\documentclass{llncs}
\usepackage{subfig}
\usepackage{graphicx} %
\usepackage{booktabs}
\usepackage{times}

\usepackage[ruled,vlined,linesnumbered,lined,noresetcount]{algorithm2e}
\usepackage{algpseudocode}
\usepackage{verbatim}
\usepackage{amsmath}
\usepackage{amssymb}
\usepackage{mathtools}
\usepackage{xspace}
\usepackage{url}
\DeclareSymbolFont{bbold}{U}{bbold}{m}{n}
\DeclareSymbolFontAlphabet{\mathbbold}{bbold}
\newcommand{\bzero}{\mathbbold{0}}
\newcommand{\bone}{\mathbbold{1}}

\usepackage{pgf,tikz}
\usepackage{mathrsfs}
\usetikzlibrary{arrows}
\definecolor{xdxdff}{rgb}{0.5,0.5,1}
\definecolor{ffqqqq}{rgb}{1.,0.,0.}
\definecolor{qqqqff}{rgb}{0.,0.,1.}
\definecolor{ccqqqq}{rgb}{0.8,0.,0.}
\definecolor{yqqqqq}{rgb}{0.5,0.,0.}
\definecolor{yqqqyq}{rgb}{0.5,0.,0.5}
\definecolor{dcrutc}{rgb}{0.86,0.1,0.23}
\definecolor{ubqqys}{rgb}{0.3,0.,0.5}
\definecolor{xfqqff}{rgb}{0.5,0.,1.}

\usepackage{amsmath}

\DeclareMathOperator*{\argmin}{\arg\!\min}
\newcommand{\Rat}{\mathbb{Q}}
\newcommand{\Real}{\mathbb{R}}
\newcommand{\name}{multi-mode system\xspace}
\newcommand{\names}{multi-mode systems\xspace}

\newcommand{\Names}{Multi-Mode Systems\xspace}

\newcommand{\cstime}[1]{\Delta t_{#1}}
\newcommand{\cscost}[1]{\Delta \pi_{#1}}

\newcommand{\tmax}{t_{\text{max}}}

\newcommand{\vmax}{V_\text{max}}
\newcommand{\vmin}{V_\text{min}}
\newcommand{\vinit}{V_\text{0}}
\newcommand{\vend}{V_\text{end}}
\newcommand{\vdiff}{(\vmax-\vmin)}

\newcommand{\nfrac}[2]{#1/#2}
\newcommand{\modesup}{M^+}
\newcommand{\modesdown}{M^-}
\newcommand{\modeszero}{M^0}
\newcommand{\allmodes}{M}

\newcommand{\shot}{leap\xspace} %
\newcommand{\shots}{leaps\xspace} %
\newcommand{\automata}{\ensuremath{\mathcal{A}}\xspace}
\DeclareMathOperator{\run}{run}

\DeclareMathOperator{\Resize}{resize}
\DeclareMathOperator{\maxresize}{resize-domain}

\newcommand{\flatm}{zero-mode\xspace}
\newcommand{\flatms}{zero-modes\xspace}
\newcommand{\size}{|\automata|}
\newcommand{\calO}{\mathcal{O}}
\newcommand{\splitatcommas}[1]{%
  \begingroup
  \begingroup\lccode`~=`, \lowercase{\endgroup
    \edef~{\mathchar\the\mathcode`, \penalty0 \noexpand\hspace{0pt plus 1em}}%
  }\mathcode`,="8000 #1%
  \endgroup
}
\newcommand{\seq}[1]{\splitatcommas{\langle #1 \rangle}}

\newcommand{\PUD}{partial-up+down\xspace}

\newcommand{\PUDU}{partial-up+down+up\xspace}
\newcommand{\U}{up\xspace}
\newcommand{\D}{down\xspace}

\newcommand{\UPDD}{up+partial-down+down\xspace}
\newcommand{\PU}{partial-up\xspace}
\newcommand{\UPD}{up+partial-down\xspace}
\newcommand{\PUUD}{partial-up+up+down\xspace}
\newcommand{\FUD}{flat+up+down\xspace}
\newcommand{\FD}{flat+down\xspace}
\newcommand{\UD}{up+down\xspace}
\newcommand{\PDUD}{partial-down+up+down\xspace}
\newcommand{\PUU}{partial-up+up\xspace}
\newcommand{\UPDDU}{up+partial-down+down+up\xspace}
\newcommand{\PDD}{partial-down+down\xspace}
\newcommand{\Wedge}{wedge\xspace}
\newcommand{\patnum}{44\xspace}

\allowdisplaybreaks
\usepackage{microtype}
\title{Optimal Control for \Names\\ with Discrete Costs}
\author{Mahmoud A. A. Mousa, Sven Schewe, and Dominik Wojtczak}
\institute{
University of Liverpool, Liverpool, U.K.
}

\begin{document}
\maketitle
\pagestyle{plain}
\begin{abstract}
This paper studies optimal time-bounded control in multi-mode systems with discrete costs. Multi-mode systems are an important subclass of linear hybrid systems, in which there are no guards on transitions and all invariants are global. Each state has a continuous cost attached to it, which is linear in the sojourn time, while a discrete cost is attached to each transition taken. We show that an optimal control for this model can be computed in {\sc NExpTime} and approximated in {\sc PSpace}. We also show that the one-dimensional case is simpler: although the problem is {\sc NP}-complete (and in {\sc LogSpace} for an infinite time horizon), we develop an {\sc FPTAS} for finding an approximate solution. 
\end{abstract}

\section{Introduction}
Multi-mode systems \cite{ATW12b} are an important subclass of linear hybrid systems \cite{alur_algorithmic_1995}, which consist of multiple continuous variables and global invariants for the values that each variable is allowed to take during a run of the system.
However, unlike for the full linear hybrid systems model, \names have no guards on transitions and no local invariants.
In this paper, we study multi-mode systems with discrete costs, which
extend linear hybrid systems by adding both continuous and discrete costs to states. 
Every time a transition is taken (i.e.~when the current state changes),
the discrete cost assigned to the target state is incurred.   
The continuous cost is the sum of the products of the sojourn time in each state and the cost assigned to this state.
Our aim is to minimise the total cost over a finite-time horizon 
or a long-time average cost over an infinite time horizon.
We exemplify this by applying this model to the optimal control of heating, ventilation, and air-conditioning (HVAC) systems.
HVAC systems account for about $50\%$ of the total energy cost in buildings \cite{perez-lombard_review_2008}, so a lot of energy can be saved by optimising their control.
Many simulation programs have been developed to analyse the influence of control on the performance of HVAC system components such as TRNSYS \cite{_trnsys_????}, EnergyPlus \cite{_energyplus_????}, and the Matlab's IBPT \cite{_ibpt.org_????}. 
Our approach has the advantage over the existing control theory techniques that it provides approximation guarantees. 
Although the actual dynamics of a HVAC system is governed by linear differential equations,
one can argue \cite{nghiem_green_2011,nghiem_event-based_2013,MSW16}
that constant rate dynamic, as in our model, can approximate well such a behaviour.

The simplest subclass of our model is \names with a single dimension.
It naturally occurs when controlling the temperature in a single room or building to stay in a pleasant range.
For this, the system can be in different modes, e.g.\ the air-conditioning can be switched on or off, 
or one can choose to switch on an electrical radiator or a gas burner. Each such a configuration can be modelled as mode of our \name.
Modes have start-up cost (gas burners, e.g.\ may suffer some wear and tear when switched on) as well as continuous costs.

When keeping an office building in a pleasant temperature range during opening hours, we face a control problem for \names with a finite time horizon.
We show that finding an optimal schedule in such a case is NP-complete and significantly more challenging than for the infinite time horizon (LogSPACE). However, we devise an FPTAS for the finite time horizon problem.

Heating multiple rooms simultaneously can be naturally modelled by \names (with multiple dimensions). 
In such a scenario, we might have different pleasant temperature ranges in different rooms
and the temperatures of the individual rooms may influence each other.
Naturally, controlling a multi-dimensional \names is more complex than controlling a one-dimensional \name.
We develop a nondeterministic exponential time algorithm for the construction of optimal control, whose complexity is only driven by potentially required high precision in exponentially many mode switches.
Allowing for an $\varepsilon$-deviation from the ranges of pleasant temperatures reduces the complexity to {\sc PSpace}.

{\noindent \bf Related work. }
Our model can be viewed as a weighted extension of the linear hybrid automata model (\cite{springerlink:10.1007/3-540-57318-6_30,Hen96}), but with global constraints only. 
Even basic questions for the general linear hybrid automata model are undecidable already for three variables 
and not known to be decidable for two variables \cite{asarin_low_2012}.
Most of the research for this model has focused on qualitative objectives such as reachability.
Various subclasses of hybrid systems with a decidable reachability problem were considered, see e.g. \cite{asarin_low_2012} for an overview. 
In particular, reachability in linear hybrid systems, where the derivative of each variable in each state is constant, can be shown to be decidable for one continuous variable by using the techniques from
\cite{laroussinie_model_2004}.   
In \cite{alur_theory_1994-1}, it has been shown that reachability is decidable for timed automata, which are a particular subclass of hybrid automata where the slope of all variables is equal to 1.

In \cite{MSW16} we only studied the one-dimensional case of our model 
with the simplifying assumption that there is exactly one mode that can
bring the temperature down and it is cost-free.
In this paper, we drop this assumption and generalise the model to multiple dimensions.
In the one-dimensional setting, we manage to prove similar nice algorithmic properties as in \cite{MSW16},
i.e.\ the existence of finitely many patterns for optimal schedules, polynomial constant-factor approximation algorithm and an FPTAS.
However, as opposed to the existence of a unique pattern for an optimal schedule in \cite{MSW16}, we show that that there can be \patnum different patterns when the simplifying assumption is dropped.
To show this, we need to devise five safety-preserving and cost-non-increasing operations on schedules,
while in \cite{MSW16} it sufficed for each mode to just lump together all timed actions that use this mode.
Also, our constant-factor approximation algorithm requires a careful analysis of the interplay between different sections of the normal form for schedules, which results in an $\calO(n^7)$ algorithm, while in \cite{MSW16} it sufficed to use one mode all the time and the algorithm ran in linear time.

Multi-mode systems were studied in \cite{ATW12b}, but with no discrete costs and with infinite time horizons only.
They were later extended in \cite{DBLP:conf/hybrid/AlurFMT13} to a setting where the rate of change of each variable in a mode belonging to an interval instead of being constant.
\cite{wojtczak_optimal_2013} studied a hybrid automaton model where the dynamics are governed by linear differential equations, but again without switching costs and only with an infinite time horizon. 
Both of these papers show that, for any number of variables, a schedule with the optimal long-time average cost can be computed in polynomial time.
In \cite{nghiem_green_2011,nghiem_event-based_2013}, the same models without switching costs have been studied over the infinite time horizon, with the objective of minimising the peak cost, rather than the long-time average cost.
In \cite{bouyer_average-price_2008}, 
long-time average and total cost games have been shown to be decidable for hybrid automata with strong resets, 
in which all variables are reset to $0$ after each discrete transition.
The long-time average and total cost optimisation for the weighted timed automata model have been shown to be {\sc PSpace}-complete (see e.g.\ \cite{bouyer_weighted_2006-1} for an overview).

There are many practical approaches to the reduction of energy consumption and peak demand in buildings.
One particularly popular one is model predictive control (MPC) \cite{camacho2013model}. 
In \cite{OUPAM10}, stochastic MPC was used to minimise the energy consumption in a building.
In \cite{li_optimal_2010}, On-Off optimal control was considered for air conditioning and refrigeration.
The drawback of using MPC is its high computational complexity and the fact that it cannot provide any worst-case guarantees.
UPPAAL Stratego \cite{david_uppaal_2015} supports the analysis of the expected cost in linear hybrid systems, 
but uses a stochastic semantics of these models \cite{david_time_2011,david_optimizing_2013}.
I.e.~a control strategy induces a stochastic model where 
the time delay in each state is uniformly or exponentially distributed.
This is different to the standard nondeterministic interpretation of the model,
which we use in this paper.
In \cite{larsen2016online}, an on-line controller synthesis combined with 
machine learning and compositional synthesis techniques was applied 
for optimal control of a floor heating system.

\smallskip
{\noindent \bf Structure of the paper.} The paper is organised as follows. We introduce all necessary notation and formally define the model in Section~\ref{sec:background}. In Section~\ref{sec:limit-safe}, 
we study the computational complexity of limit-safe and $\epsilon$-safe control in multiple dimensions.
In Section~\ref{sec:schedule-structure}, we show that in one dimension every schedule can be transformed without increasing its cost 
into a schedule following one of 44 different patterns.
In Section~\ref{sec:optimal-one-dimension}, we show that the cost optimisation decision problem in one-dimension with infinite and finite horizon is 
{\sc LogSpace} and {\sc NP}-complete, respectively. In Section~\ref{sec:approx}, still for the one-dimension case, we first show a constant factor approximation algorithm and, building on it, develop an FPTAS by a reduction to the 0-1 knapsack problem. 
To ease the exposition, some of the proofs and algorithms were moved to the appendix.

\section{Preliminaries} 
\label{sec:background}

Let $\bzero_N$ and $\bone_N$ be $N$-dimensional vectors with all entries equal to $0$ and $1$, respectively. By $\mathbb{R}_{\geq 0}$ and $\mathbb{Q}_{\geq 0}$ we denote the sets of all non-negative real and rational numbers, respectively. We assume that $0 \cdot \infty = \infty \cdot 0 = 0$.
For a vector $v$, let $\|v\|$ be its $\infty$-norm (i.e.\ the maximum coordinate in $v$).
We write $v_1 \leq v_2$ if every coordinate vector of vector $v_1$ is smaller than or equal to the corresponding coordinate in vector $v_2$,
and $v_1 < v_2$ if, additionally, $v_1 \neq v_2$ holds.

\subsection{Formal Definition of \Names}

Motivated by our application of keeping temperature in multiple rooms within comfortable range, 
we restrict ourselves to safe sets being hyperrectangles, which can be specified by 
giving its two extreme corner points. 
A {\em \name with discrete costs}, \automata, henceforth referred to simply as {\em \name},
is formally defined as a tuple
$\automata = (M,N,A,\pi_{c},\pi_{d},\vmin,\vmax,\vinit)$ %
where:  
\begin{itemize}
\item $M$ is a finite set of modes;
\item $N \geq 1$ is the number of continuous variables in the system;
\item $A : M \to \Rat^N$ is the slope of all the variables in a given mode;
\item $\pi_{c} : M \to \Rat_{\geq 0}$ is the cost per time unit spent in a given mode;
\item $\pi_{d} : M \to \Rat_{\geq 0}$ is the cost of switching to a given mode;
\item $\vmin, \vmax \in \Rat^N$: $\vmin < \vmax$, define the safe set, $S$,
as follows $\{x \in \Real^N : \vmin \leq x \leq \vmax\}$;
\item $\vinit \in \Rat^N$, such that $\vinit\in S$, defines the initial value of all the variables.
\end{itemize}

\subsection{Schedules, their cost and safety}

A {\em timed action} is a pair $(m,t) \in M \times \mathbb{R}_{\geq 0}$ of a mode $m$ and time delay $t > 0$. 
A {\emph{schedule}} $\sigma$ (of length $k$) with time horizon $\tmax$ is a finite sequence of timed actions $\sigma=\seq{(m_1,t_1), (m_2,t_2),\ldots,(m_k,t_k)}$, such that $\sum_{i=1}^{k} t_{i} = \tmax$.
A {\emph{schedule}} $\sigma$ with infinite time horizon is either an infinite sequence of timed actions $\sigma=\seq{ (m_1,t_1), (m_2,t_2),\ldots,(m_k,t_k), \ldots}$, such that $\sum_{i=1}^{\infty} t_{i} = \infty$
or a finite sequence of timed actions $\sigma=\seq{(m_1,t_1), (m_2,t_2),\ldots,(m_k,t_k)} $, such that $t_k = \infty$.
The \textit{run} of a finite schedule $\sigma = \seq{ (m_1,t_1), (m_2,t_2),\ldots,(m_k,t_k)} $ is a sequence of {\em states} $\run(\sigma)= \seq{V_0,V_1,...,V_k}$ such that, %
for all $0 \leq i \leq k-1$, we have that $V_{i+1} = V_i+t_i A(m_i)$.

A schedule and its run are called \textit{safe} if $\vmin \leq V_{i} \leq \vmax$ holds for all $1 \leq i \leq k$. 
A schedule and its run are called \textit{$\epsilon$-safe} if $\vmin-\epsilon\cdot \bone_N < V_{i} < \vmax + \epsilon \cdot \bone_N$ holds for all $1 \leq i \leq k$.
The run of an infinite schedule and its safety and $\epsilon$-safety are defined accordingly.

The {\em total cost} of a schedule $\sigma=\seq{(m_1,t_1), (m_2,t_2),\ldots,(m_k,t_k)}$ with a finite time horizon is defined as $\pi(\sigma) = \sum_{i=1}^{k} \pi_{d}(m_{i}) + \pi_{c}(m_{i})t_{i}$. 
The {\em limit-average cost} for a finite schedule $\sigma=\seq{(m_1,t_1), (m_2,t_2),\ldots,(m_k,t_k)}$ with an infinite time horizon is defined as
$\pi_{avg}(\sigma) = \pi_c(m_{k})$ and for an infinite schedule $\sigma=\seq{(m_1,t_1), (m_2,t_2),\ldots}$ 
it is defined as $$\pi_{avg}(\sigma) = \limsup_{k\to\infty}\nfrac{\left(\sum_{i=1}^{k} \pi_{d}(m_{i}) + \pi_{c}(m_{i})t_{i}\right)\Big}{\sum_{i=1}^{k} t_i}$$ 

A safe finite schedule $\sigma$ is {\em $\epsilon$-optimal} 
if, for all safe finite schedules $\sigma'$, we have that $\pi(\sigma') \geq \pi(\sigma)-\epsilon$.
A safe finite schedule is {\em optimal} if it is $0$-optimal. A safe infinite schedule $\sigma$ is {\em optimal} if, for all safe infinite schedules $\sigma'$, we have that $\pi_{avg}(\sigma') \geq \pi_{avg}(\sigma)$.

The following example shows that there may not be an optimal schedule for a \name with a finite time horizon.
\begin{example}
\label{ex:no-optimal}
Consider a \name with three modes: $M_1, M_2, M_3$.
The slope vectors in these modes are $A(M_1) = (1,1)$, $A(M_2) = (1,-1)$ and $A(M_3) = (-1,1)$, respectively.
The continuous cost of using $M_1$ is $\pi_c(M_1) = 1$ and all the other costs are $0$.
Let $\vinit = \vmin = \bzero_2$ and $\vmax = \bone_2$.
Notice that we can only use $M_2$ or $M_3$ once we get out of the initial corner $\vinit$.
This can only be done using $M_1$. 
Now let the time horizon be $\tmax$.
Note that the following schedule $\sigma_\epsilon = (M_1,\epsilon), \big((M_2, t), (M_3, t)\big)^l$, where $t'=\tmax-\epsilon$, $l = \lceil t'/\epsilon \rceil$, and $t = t' / 2l$, has time horizon $\tmax$ and total cost $\epsilon > 0$.
As $\epsilon$ can be made arbitrarily small but has to be $>0$, $\sigma_\epsilon$ is an $\epsilon$-optimal schedule for all $\epsilon > 0$, but no optimal schedule exists.
\end{example}

\newcommand{\zerocostmodes}{M^*}
Note that in Example \ref{ex:no-optimal}, for any $\epsilon > 0$, there exists an optimal $\epsilon$-safe schedule $\sigma$ with total cost $0$:
$\sigma_0 = \seq{\big((M_2, t), (M_3, t)\big)^l}$ where $l$ is defined as in Example \ref{ex:no-optimal}. 
Our aim is to find an ``abstract schedule'' that, for any given $\epsilon>0$, can be used to construct in polynomial time an $\epsilon$-safe $\epsilon$-optimal schedule.
\newcommand{\ata}{\mathbf{t}}

Let $\zerocostmodes = \{m \in M \mid \pi_d(m) = 0\}$ be the subset of modes without discrete costs.
Note that, as shown in \cite{ATW12b}, the cost and safety of a schedule with $\zerocostmodes$ modes only, depends
only on the total amount of time spent in each of the $\zerocostmodes$ modes.
We therefore lump together any sequence of timed actions that only use $\zerocostmodes$ modes and define 
an {\em abstract timed action (over $\zerocostmodes$)} as 
a function $\ata : \zerocostmodes \to \mathbb{R}_{\geq 0}$.
A finite {\em abstract schedule} with time horizon $\tmax$ (of length $k$) is a finite sequence
 $\tau=\seq{\ata_1, (m_1,t_1), 
 \ata_2, (m_2,t_2),\ldots,(m_{k-1},t_{k-1}),
 \ata_k}$ such that $\forall_i\;m_i \in M\setminus\zerocostmodes$ and 
 $\sum_{i\leq k, m\in \zerocostmodes} \ata_i(m) + \sum_{i < k} t_i = \tmax$.
The run of the abstract schedule $\tau$ is a sequence $\seq{V_0,V_1,\ldots,V_{2k+1}}$ such that,
for all $i \leq k$, we have
$V_{2i} = V_{2i-1} + A(m_i) t_i$ and $V_{2i+1} = V_{2i} + \sum_{m \in \zerocostmodes} A(m) \ata_i(m)$.
We say that an abstract schedule is {\em limit-safe} if its run is safe.
The total cost of an abstract schedule $\tau$ is defined as 
$$\sum_{i\leq k, m\in \zerocostmodes} \pi_c(m,\ata_i(m)) + \sum_{i < k} \big(\pi_d(m_i) + \pi_c(m_i) t_i\big) \, .$$
Note that any safe schedule can be turned into a limit-safe abstract schedule with the same cost by simply replacing any maximal subsequence of consecutive timed actions that only use $\zerocostmodes$ modes by a single abstract timed action.
A limit-safe abstract schedule $\sigma$ is optimal if the total cost of all other limit-safe abstract schedules is higher than $\pi(\sigma)$. 
The following statement justifies the name ``limit-safe''.

\begin{proposition}
\label{prop:limit-safe-construct}
Given a limit-safe abstract schedule $\tau$ and $\epsilon > 0$, 
we can construct in polynomial time an $\epsilon$-safe schedule $\sigma$
such that $\pi(\tau) = \pi(\sigma)$.
\end{proposition}
\begin{proof}
Let $M^* = \{m_1,m_2,\ldots,m_j\}$.
To obtain $\sigma$ from $\tau$, we replace each abstract timed action $\big\{\big(m,t_m) \mid m \in M^*\big\}$ by a sequence
$\big((m_1,t_{m_1}/l),\ldots,(m_j,t_{m_j}/l)\big)^l$ for a sufficiently large $l \in \mathbb N$.

Sufficiently large means that, for $t^*=\sum_{m\in M^*}t_m$, $l> t^* \cdot \max_{m \in M^*}\|A(m)\| / \epsilon$.
This choice guarantees that $\sum_{m\in M^*}\|A(m)\|\cdot t_m/l < \varepsilon$.
Thus, when the abstract action $\big\{\big(m,t_m) \mid m \in M^*\big\}$ joins two states $V_{2i},V_{2i+1}$ along the run $\langle V_0,V_1,\ldots,\ldots,V_{2k+1} \rangle$ of $\tau$, we know that this concrete schedule will cover the $l$-th part of $V_{2i},V_{2i+1}$ after every sequence $(m_1,t_{m_1}/l),(m_2,t_{m_2}/l),\ldots,(m_j,t_{m_j}/l)$.
As the safe set is convex, the start and end points of this sequence are safe points.
Also, $\sum_{m\in M^*}\|A(m)\|\cdot t_m/l < \varepsilon$ implies that the points in the middle are $\epsilon$-safe.
\qed\end{proof}

\smallskip
\noindent {\bf Example \ref{ex:no-optimal} continues.} 
{\itshape
An example limit-safe abstract schedule of length $1$ is $\tau = \{(m_1,\tmax/2),(m_2,\tmax/2)\}$.
Based on $\tau$ we can construct an $\epsilon$-safe schedule $\seq{\big((m_1,\tmax/2l),(m_2,\tmax/2l)\big)^l}$ where $l$ is any integer greater than $\tmax/\epsilon$.
}

\subsection{Structure of optimal schedules}

We show here that it later suffices to consider only schedules with a particular structure.

\begin{definition}
We call a finite schedule $\sigma$ {\em angular} if
there are no two consecutive timed actions $(m_i,t_i), (m_{i+1},t_{i+1})$ in $\sigma$ 
such that $A(m_i) = A(m_{i+1})$.
\end{definition}

We show that while looking for an ($\epsilon$-)safe ($\epsilon$-)optimal finite schedule, we can restrict our attention to angular schedules only. 

\begin{proposition}
\label{prop:angular}
For every finite ($\epsilon$-)safe schedule with time horizon $\tmax$ there exists an angular safe schedule with the same or lower cost.
\end{proposition}

Henceforth, we assume that all finite schedules are angular.
Let $\modeszero = \{m \mid A(m) = 0\}$, which we will also refer to as {\em \flatms}. 

\begin{proposition}
\label{prop:zero-mode}
For every finite safe schedule with time horizon $\tmax$ there exists a safe schedule with the same or lower cost, in which at most one \flatm is used at the very beginning.
\end{proposition}

Henceforth, we assume that all finite schedules use at most one \flatm timed action and only at the very beginning.

\subsection{Approximation algorithms}

We study approximation algorithms for the total cost minimisation problem in \names.
We say that an algorithm is a {\em constant factor approximation algorithm} with a {\em relative performance $\rho$}
iff,
for all inputs $x$, the cost of the solution that it computes, $f(x)$, satisfies
$OPT(x) \leq f(x) \leq (1+\rho)\cdot OPT(x)$, 
where $OPT(x)$ is the optimal cost for the input $x$. 
We are particularly interested in polynomial-time approximation algorithms.
A polynomial-time approximation scheme (PTAS) is an algorithm that, for every $\rho > 0$,
runs in polynomial-time and has relative performance $\rho$.
Note that the running time of a PTAS may depend in an arbitrary way on $\rho$.
Therefore, we typically strive to find a fully polynomial-time approximation scheme (FPTAS), which 
is an algorithm that runs in polynomial-time 
in the size of the input and $1/\rho$.

The 0-1 Knapsack problem is a well-known NP-complete optimisations problem, 
which possess multiple FPTASes (see e.g.\ \cite{kellerer_knapsack_2004}).
In this problem we are given a knapsack with a fixed volume and a list of items, each with an integer volume and value. The aim is to pick a subset of these items that together do not exceed the volume of the knapsack and have the maximum total value. 

\section{Complexity of limit-safe and $\epsilon$-safe finite control}
\label{sec:limit-safe}

As our one-dimensional model strictly generalises the simple linear hybrid automata considered in \cite{MSW16}, 
we immediately obtain the following result. 

\begin{theorem}[follows from \cite{MSW16}, Theorem 3]
\label{thm:np-hard}
Given (one-dimensional) \name \automata, constants $\tmax$ and $C$ (both in binary),
checking whether there exists a safe schedule in \automata with time horizon $\tmax$ and total cost at most $C$ is NP-hard.
\end{theorem}

In the rest of this section we fix a (multi-dimensional) \name \automata and time horizon $\tmax$. 

\begin{theorem}
\label{thm:limit-safe-existence}
If a limit-safe abstract schedule exists in \automata, then
there exists one of exponential length and it can be constructed in polynomial time.
\end{theorem}
\begin{proof}[sketch]
Before we formally prove this theorem, 
we need to introduce first a bit of terminology. We call a mode $m$ {\em safe for time $t > 0$} at $V \in S := \{x \in \Real^N : \vmin \leq x \leq \vmax\}$ if $V + A(m)t \in S$.
Also, $m$ is {\em safe} at $V$ if there exists $t > 0$ such that $m$ is safe for time $t$ at $V$.
We say that a {\em coordinate of a state}, $V \in S$, {\em is at the border} if that coordinate
in $V$ is equal to the corresponding coordinate in $\vmin$ or $\vmax$.

Our algorithm first removes from $M$ all modes that will never be safe to use in a limit-safe schedule.
This procedure can be found between lines 1 -- 8 of Algorithm \ref{alg:find-limit-safe} in Appendix \ref{app:limit-safe-schedule}.
This is an adaptation of 
\cite[Theorem 7]{ATW12b}
where an algorithm was given for finding safe modes that can ever be used in a schedule with no time horizon.
The main difference here is that the modes in $\zerocostmodes$ can always be used in a limit-safe abstract schedule even if they are not safe to use.
We find here a sequence of sets of modes $\zerocostmodes = M_0 \subset M_1 \subset M_2 \subset \ldots$
such $M_{i+1}$ is the set of modes that are safe at a state reachable from $\vinit$ via
a limit-safe abstract schedule that only uses modes from $M_i$. 
Note that at some step $k \leq |M|$ this sequence will stabilise, i.e.\ $M_k = M_{k+1}$.
Similarly as in the proof of \cite[Theorem 7]{ATW12b}, we can show that
no mode from $M\setminus M_k$ can ever be used by a limit-safe abstract schedule.
As a result, we can remove all these modes from $M$.

Next, we remove all modes that cannot be part of a limit-safe abstract schedule with time horizon $\tmax$. 
For this, for each $m$, we formulate a very similar linear programme (LP) as above (cf.\ lines 9 -- 11 of Algorithm \ref{alg:find-limit-safe}) where we ask for the time delay of $m$ to be positive and the total time delay of all the modes to be $\tmax$.
By a simple adaptation of the proof of \cite[Theorem 4]{ATW12b},
if this LP is not satisfiable then $m$ can be removed from \automata.

Next, we look for the easiest possible target state $\vend$ that can potentially be reached using a limit-safe abstract schedule from $\vinit$ with time horizon $\tmax$.
For this, $\vend$ has to have the least number of coordinates at the border of the safe set.
Note that this is well-defined, because if $V$ and $V'$ are two points reachable from $\vinit$ via a limit-safe abstract schedules $\tau$ and $\tau'$ with time horizon $\tmax$, respectively, 
then $\tau/2$ (i.e.\ divide all abstract and timed actions delays in $\tau$ by 2) followed by $\tau'/2$, is also a limit-safe abstract schedule with time horizon $\tmax$, which reaches $(V + V')/2$.
However, $(V + V')/2$ has a coordinate at the border iff both $V$ and $V'$ have it as well.
This shows that there is a state with a minimum number of coordinates at the border. 

To find the coordinates that need to be at the border we will use the following 
LP.
We have a variable $x_i$ for each dimension $i \leq N$
and a constraint that requires $x_i$ to be less or equal to the $i$-th coordinate of $\vmax - \vend$ {\bf and} $\vend - \vmin$.
We also add that $\sum_{m\in M}t_m = \tmax$ and $\vend = \vinit + \sum_{m\in M}t_m \cdot A(m)$,
with the objective {\em Maximise $\sum_i x_i$}.
If the value of the objective is $>0$, we will get to know a new coordinate that does not have to be at the border.
We then remove it from the LP and run it again.
Once the objective is $0$, then all the remaining coordinates, $I$, have to be at the border
and the solution to this LP tells us, at which border the solution has to be located (it cannot possibly be at the border of both $\vmin$ and $\vmax$ as then we could reach the middle).

Next, in order to bound the length of a limit-safe abstract schedule by an exponential in the size of the input, we not only need a state with the minimum number of coordinates at the border, but also sufficiently far way from the border. Otherwise, we may need super-exponentially many timed actions to reach it. In order to find such a point, we replace all $x_i$-s in the previously defined LP by a single variable $x$ which is 
smaller or equal to all the coordinates of $\vmax - \vend$ and $\vend - \vmin$ from~$I$. 
We then set the objective to {\em Maximise $x$}, which will give us a suitable easy target state~$\vend$. 

Now, consider $\automata'$, which is the same as \automata but with all slopes negated (i.e.\ $A'(m) = - A(m)$ for all $m \in M$). We claim that $\vend$ is reachable from $\vinit$ using a limit-safe abstract schedule $\tau$
iff $(\vinit+\vend)/2$ is reachable from $\vinit$ in \automata with time horizon $\tmax/2$ and $(\vinit+\vend)/2$ is reachable from $\vend$ in $\automata'$ with time horizon $\tmax/2$; this again follows by considering $\tau/2$. Note that a coordinate of $(\vinit+\vend)/2$ is at the border iff it is at the border in both $\vinit$ and $\vend$.

This way we reduced our problem to just checking whether a limit-safe abstract schedule exists from one point to another more permissive point (i.e.\ where the set of safe modes is at least as big) within a given time horizon. 
Algorithm \ref{alg:find-limit-safe} in the appendix \ref{app:limit-safe-schedule} solves this problem and constructs (if there exists one) a limit-safe abstract schedule of at most exponential length with these properties. 
It again reuses the same constructions as above,
e.g.\ constructs exactly the same sequence of sets of modes $\zerocostmodes = M_0 \subset M_1 \subset \ldots \subset M_k$,
and its correctness follows by a similar reasoning as above. 
We now need to invoke this algorithm twice: to check that $(\vinit+\vend)/2$ is reachable from $\vinit$ with time horizon $\tmax/2$ and that $(\vinit+\vend)/2$ is reachable from $\vend$ with time horizon $\tmax/2$ in $\automata'$.
If at least one of these calls return NO, then no limit-safe abstract schedule from $\vinit$ to $\vend$ can exist. Otherwise, let $\sigma$ and $\sigma'$ be the schedules returned by these two calls, respectively. Then the concatenation of $\sigma$ with the reverse of $\sigma'$ is a limit-safe abstract schedule that reaches $\vend$ from $\vinit$ with time horizon $\tmax$.
\qed

\end{proof}

\begin{theorem}
\label{thm:limit-safe}
Finding an optimal limit-safe abstract schedule in \automata can be done in nondeterministic exponential time.
\end{theorem}
\begin{proof}
The limit-safe abstract schedule constructed in Theorem \ref{thm:limit-safe-existence} has an exponential length.
To establish a nondeterministic exponential upper bound, we can guess the modes (and the order in which they occur).
With them, we can produce an exponentially sized linear program, which encodes that the run of the abstract schedule is safe and minimises the total cost incurred.
\qed\end{proof}

Theorem \ref{thm:limit-safe} and Proposition \ref{prop:limit-safe-construct} immediately give us the following.

\begin{corollary}
If a limit-safe abstract schedule exists in \automata, then for any $\epsilon > 0 $
an $\epsilon$-safe schedule with the same cost can be found in nondeterministic exponential time.
\end{corollary}

Moreover, from Theorem \ref{thm:limit-safe-existence} and the fact that in the case of \names with no discrete costs all abstract schedules have length $1$, we get the following.

\begin{corollary}
Finding an optimal limit-safe abstract schedule for \names with no discrete costs can be done in polynomial time.
\end{corollary} 

We can reduce the computational complexity in the general model
if we are willing to sacrifice optimality for $\epsilon$-optimality.

\newcommand{\amax}{\ensuremath{\max_{m \in \allmodes}|A(m)|}}
\begin{theorem}
If a limit-safe abstract schedule exists, then finding an $\epsilon$-safe $\epsilon$-optimal strategy 
can be done in deterministic polynomial space.
\end{theorem}
\begin{proof}
When reconsidering the linear programme from the end of the proof of Theorem \ref{thm:limit-safe}, we can guess the intermediate states in polynomial space (and thus guess and output the schedule) as long as all states along the run (including the time passed so far) are representable in polynomial space.

Otherwise we use the opportunity to deviate by up to $\epsilon$ from the safe set 
by increasing or decreasing the duration of each timed action up to some $\delta>0$,
in order to keep the intermediate values representable in space polynomial in $\size$ and $\epsilon$.
However, we apply these changes in a way that the overall time remains $\tmax$.
Clearly this is possible, because within $\delta/2$ of the actual time point of each state along the run, 
there is a value whose number of digits in the standard decimal notation is at most equal to the sum of the number of digits in $\delta/2$ and $\tmax$.
Picking any such point for every interval would induce a schedule with the required property and they can be simply guessed one by one.
The final imprecision introduced by this operation is at most $b\cdot \delta\cdot \amax$,
where $b$ is a bound on the number of timed actions in a limit-safe schedule, which is exponential in $\size$.
If we choose $\delta = \epsilon/(b \cdot \amax)$, 
then we will get the required precision.

Although our algorithm is nondeterministic, due to Savitch's theorem,
it can be implemented in deterministic polynomial space.
\qed\end{proof}

\newcommand{\flexi}{flexi\xspace}
\newcommand{\flexis}{flexis\xspace}
\section{Structure of Finite Control in One-dimension}
\label{sec:schedule-structure}

We show in this section that any finite safe schedule in one-dimension can be transformed 
without increasing its cost into 
a safe schedule, which follows one of finitely many regular patterns. 
The crucial component of this normal form will be a ``\shot'' that we define below. We first introduce some notation.
Let $\modesup = \{m \mid A(m) > 0\}$ and $\modesdown = \{m \mid A(m) < 0\}$.
Recall that $\modeszero = \{m \mid A(m) = 0\}$.
We will call a mode, $m$, an {\em up mode, down mode, or \flatm} if $m \in \modesup$, $m \in \modesdown$, or $m \in \modeszero$, respectively.
Similarly, the {\em trend} of a timed action $(m,t)$ is {\em up, down, flat} if $m$ is an up, down, \flatm, respectively.
For any subsequence of timed actions $\sigma' = \seq{(m_i,t_i),\ldots,(m_j,t_j)}$ in a schedule $\sigma$, whose
run is $run(\sigma) = \seq{V_0,V_1,\ldots,V_k}$,
we say that 
$\sigma'$ {\em starts at state $v$} and {\em ends at state $v'$}
iff $v = V_{i-1}$ and $v' = V_{j}$.
We use the same terminology for a single timed action (in this case this subsequence has length 1).

\begin{definition}
A {\em partial \shot{}} 
is a pair of consecutive timed actions $(m_{i},t_{i}),(m_{i+1},t_{i+1})$ in a safe schedule such that $m_i \in \modesup$, $m_{i+1} \in \modesdown$, and $A(m_{i})t_{i} + A(m_{i+1})t_{i+1} = 0$, i.e.\ the state of a \name does not change after any partial \shot{}.
A partial leap is {\em complete} if $A(m_{i})t_{i} = \vmax - \vmin$. We will simply refer to complete leaps as {\em leaps}.

There are $|M^+ \times M^-|$ types of \shots{}. A \shot{} is of {\em type $(m,m') \in \modesup \times \modesdown$} iff $m_i = m$ and $m_{i+1} = m'$. 
Let $\cstime{m}$ and $\cscost{m}$ denote the time and cost it takes for an up mode $m$ to get from $\vmin$ to $\vmax$ or a down mode $m$ to get from $\vmax$ to $\vmin$. 
Note that $\cstime{m} = |\nfrac{\vdiff}{A(m)}|$ and $\cscost{m} = \pi_d(m) + \pi_c(m) \cdot \cstime{m}$.
By $\cstime{m,m'}$ and $\cscost{m,m'}$ we denote the time duration and the cost of a \shot{} of type $(m,m')\in\modesup \times \modesdown$, respectively. 
Note that $\cstime{m,m'} = \cstime{m} + \cstime{m'}$ and $\cscost{m,m'} = \cscost{m} + \cscost{m'}$.
\end{definition}

Any safe schedule $\sigma$ can be decomposed into three sections
that we will call its {\em head, \shots, and tail}.
The {\em head section} ends after the first timed action that ends at $\vmin$.
The {\em leaps section} contains only leaps of possibly different types following the head section.
Finally, the {\em tail section} starts after the last leap in the leaps section has finished.
Note that any of these sections can be empty and the tail section can in principle contain further leaps.
We show here that, for any safe schedule of length at least three, there exists another safe one with the same or a smaller cost,
whose head and tail sections follow one of the 10 patterns presented in Figure \ref{fig:head} and Figure \ref{fig:tail}, respectively, where
{\em partial up/down} means that the next state is not at the border.
For each of these patterns, there exists an example which shows that an optimal safe schedule
may need to use such a pattern and hence it is necessary to consider it.
In order to prove this, we first need to define several cost-nonincreasing and safety-preserving operations that can be applied to safe schedules.
These will later be applied in Theorem \ref{thm:standard-form} to transform any safe schedule into one of the just mentioned regular patterns.
These operations are easy to explain via a picture, but cumbersome to define formally.
Therefore, we moved their formal definitions to the Appendix \ref{app:formal-def}
and present here only the intuition behind them.

Let $\sigma$ be any safe finite schedule.
Following Proposition \ref{prop:angular} and \ref{prop:zero-mode}, we can assume that $\sigma$ is angular and only contains at most one timed action with a \flatm, and if it contains one, this action occurs at the very beginning. Unless explicitly stated, the operations below are defined for timed actions with up or down trend only.

\begin{figure}
	\centering
	\begin{tikzpicture}[scale=0.5] 
	\draw (1.,3.5)-- (10.,3.5);
	\draw (1.,-1.5)-- (10.,-1.5);
	\draw [line width=1.6pt] (2.5,-1.)-- (3.,1.);
	\draw [line width=1.6pt] (3.,1.)-- (4.5,2.);
	\draw [line width=1.6pt] (4.5,2.)-- (8.38,2.48);
	\draw [line width=1.6pt,color=ccqqqq] (2.5,-1.)-- (4.,0.);
	\draw [line width=1.6pt,color=ccqqqq] (4.,0.)-- (8.,0.5);
	\draw [line width=1.6pt,color=ccqqqq] (8.,0.5)-- (8.38,2.48);
	\draw (-0.9,3.88) node[anchor=north west] {$\mathbf{V_{max}}$};
	\draw (-0.9,-1.12) node[anchor=north west] {$\mathbf{V_{min}}$};
	\begin{scriptsize}
	\draw [fill=qqqqff] (2.5,-1.) circle (2.0pt);
	\draw[color=qqqqff] (2.32,-1.25) node { $\mathbf{1}$};
	\draw [fill=qqqqff] (3.,1.) circle (1.5pt);
	\draw[color=qqqqff] (2.82,1.10) node { $\mathbf{2}$};
	\draw[color=black] (2.13,0.02) node { $\mathbf{m_1}$};
	\draw [fill=qqqqff] (4.5,2.) circle (2.0pt);
	\draw[color=qqqqff] (4.36,2.27) node { $\mathbf{3}$};
	\draw[color=black] (3.15,1.62) node { $\mathbf{m_2}$};
	\draw [fill=qqqqff] (8.38,2.48) circle (2.0pt);
	\draw[color=qqqqff] (8.66,2.55) node { $\mathbf{4}$};
	\draw[color=black] (6.17,2.44) node { $\mathbf{m_3}$};
	\draw [fill=qqqqff] (4.,0.) circle (2.0pt);
	\draw[color=qqqqff] (4.24,-0.33) node { $\mathbf{2'}$};
	\draw[color=ccqqqq] (3.59,-0.86) node { $\mathbf{m_2}$};
	\draw [fill=qqqqff] (8.,0.5) circle (2.0pt);
	\draw[color=qqqqff] (8.32,0.23) node { $\mathbf{3'}$};
	\draw[color=ccqqqq] (6.55,-0.1) node { $\mathbf{m_3}$};
	\draw[color=ccqqqq] (8.63,1.22) node { $\mathbf{m_1}$};
	\end{scriptsize}
	\end{tikzpicture}
	\begin{tikzpicture}[scale=0.5] 
	\draw (0.5,4.)-- (11.5,4.);
	\draw (0.5,-0.5)-- (11.5,-0.5);
	\draw [line width=1.6pt] (1.5,-0.5)-- (2.5,1.5);
	\draw [line width=1.6pt] (2.5,1.5)-- (4.5,-0.5);
	\draw [line width=1.6pt] (4.5,-0.5)-- (7.5,4.);
	\draw [line width=1.6pt] (7.5,4.)-- (9.5,-0.5);
	\draw [line width=1.6pt,color=ccqqqq] (1.5,-0.5)-- (4.38,4.);
	\draw [line width=1.6pt,color=ccqqqq] (4.38,4.)-- (6.5,-0.5);
	\draw [line width=1.6pt,color=ccqqqq] (6.5,-0.5)-- (7.5,1.5);
	\draw [line width=1.6pt,color=ccqqqq] (7.5,1.5)-- (9.5,-0.5);
	\draw (-1.18,4.38) node[anchor=north west] {$\mathbf{V_{max}}$};
	\draw (-1.14,-0.1) node[anchor=north west] {$\mathbf{V_{min}}$};
	\begin{scriptsize}
	\draw [fill=qqqqff] (1.5,-0.5) circle (2.0pt);
	\draw[color=qqqqff] (1.56,-0.93) node { $\mathbf{1}$};
	\draw [fill=qqqqff] (2.5,1.5) circle (2.0pt);
	\draw[color=qqqqff] (2.42,1.77) node { $\mathbf{2}$};
	\draw[color=black] (1.39,0.52) node { $\mathbf{m_1}$};
	\draw [fill=qqqqff] (4.5,-0.5) circle (2.0pt);
	\draw[color=qqqqff] (4.62,-0.97) node { $\mathbf{3}$};
	\draw[color=black] (3.05,0.24) node { $\mathbf{m_2}$};
	\draw [fill=qqqqff] (7.5,4.) circle (2.0pt);
	\draw[color=qqqqff] (7.54,4.27) node { $\mathbf{4}$};
	\draw[color=black] (6.41,1.56) node { $\mathbf{m_3}$};
	\draw [fill=qqqqff] (9.5,-0.5) circle (2.0pt);
	\draw[color=qqqqff] (9.58,-0.89) node { $\mathbf{5}$};
	\draw[color=black] (9.09,1.56) node { $\mathbf{m_4}$};
	\draw [fill=qqqqff] (4.38,4.) circle (2.0pt);
	\draw[color=qqqqff] (4.52,4.31) node { $\mathbf{2'}$};
	\draw[color=ccqqqq] (3.67,1.96) node { $\mathbf{m_3}$};
	\draw [fill=qqqqff] (6.5,-0.5) circle (2.0pt);
	\draw[color=qqqqff] (6.64,-0.91) node { $\mathbf{3'}$};
	\draw[color=ccqqqq] (4.95,1.72) node { $\mathbf{m_4}$};
	\draw [fill=qqqqff] (7.5,1.5) circle (2.0pt);
	\draw[color=qqqqff] (7.6,1.85) node { $\mathbf{4'}$};
	\draw[color=ccqqqq] (7.41,0.02) node { $\mathbf{m_1}$};
	\draw[color=ccqqqq] (8.29,0.22) node { $\mathbf{m_2}$};
	\end{scriptsize}
	\end{tikzpicture}
	\caption{\label{fig:Shift_Leaps} On the left, the rearrange operation applied to
		three timed actions 1-2-3 with modes $m_1,m_2,m_3$ results in 1'-2'-3' with modes $m_2,m_3,m_1$.
		On the right, the shift operation is being applied to a partial leap 1-2-3 which will be moved after the (complete) leap 3-4-5. }
\end{figure}
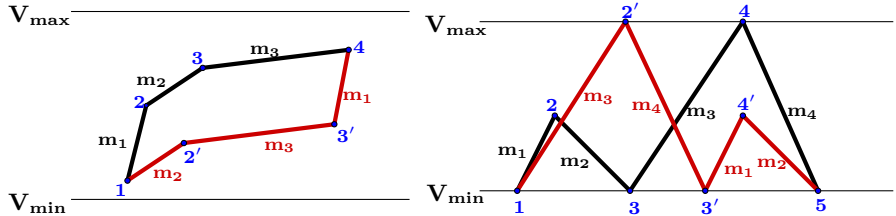 

The first operation that we need is the {\em rearrange} operation, which simply
changes the order of any subsequence of timed actions with the same trend.
The next one is the {\em shift} operation.
It cuts any subsequence of timed actions that start and end at the same state, $V$,
and pastes this subsequence after any timed action that ends at $V$.
The effect of these two operations can be seen in Figure \ref{fig:Shift_Leaps}.

\begin{figure}
	\begin{tikzpicture}[x={17.0pt},y={17.0pt}]
	\draw (-1.5,5.)-- (8.936,4.986);
	\draw (-1.5,-0.5)-- (8.936,-0.492);
	\draw [line width=1.6pt] (-0.9980520850392902,-0.4990338064233125)-- (1.4999949306288896,4.995496950062013);
	\draw [line width=1.6pt] (1.4999949306288896,4.995496950062013)-- (3.5,3.5);
	\draw [line width=1.6pt] (3.5,3.5)-- (3.9999917622719448,4.992682543850771);
	\draw [line width=1.6pt] (3.9999917622719448,4.992682543850771)-- (8.088057374570461,-0.49318876978155646);
	\draw [line width=1.6pt,color=ccqqqq] (1.4999949306288896,4.995496950062013)-- (5.499996689387469,-0.4948536528212367);
	\draw [line width=1.6pt,color=ccqqqq] (5.499996689387469,-0.4948536528212367)-- (6.064,1.048);
	\draw [line width=1.6pt,color=ccqqqq] (6.064,1.048)-- (8.088057374570461,-0.49318876978155646);
	\draw (-3.338,0.08) node[anchor=north west] {\normalsize $\mathbf{V_{min}}$};
	\draw (-3.318,5.58) node[anchor=north west] {\normalsize $\mathbf{V_{max}}$};
	\begin{scriptsize}
	\draw [fill=qqqqff] (-0.9980520850392902,-0.4990338064233125) circle (2pt);
	\draw[color=qqqqff] (-1.018,-1.05) node { 1};
	\draw [fill=qqqqff] (1.4999949306288896,4.995496950062013) circle (2pt);
	\draw[color=qqqqff] (1.442,5.43) node { 2};
	\draw[color=black] (0.692,2.) node { $\mathbf{m_i}$};
	\draw [fill=qqqqff] (3.5,3.5) circle (2pt);
	\draw[color=qqqqff] (3.642,3.29) node { 3};
	\draw[color=black] (2.992,4.44) node { $\mathbf{m_{i+1}}$};
	\draw [fill=qqqqff] (3.9999917622719448,4.992682543850771) circle (2 pt);
	\draw[color=qqqqff] (4.022,5.43) node { 4};
	\draw[color=black] (4.212,3.72) node { $\mathbf{m_{i+2}}$};
	\draw [fill=qqqqff] (8.088057374570461,-0.49318876978155646) circle (2pt);
	\draw[color=qqqqff] (8.062,-0.89) node { 5};
	\draw[color=black] (6.972,2.18) node { $\mathbf{m_{i+3}}$};
	\draw [fill=qqqqff] (5.499996689387469,-0.4948536528212367) circle (2pt);
	\draw[color=qqqqff] (5.522,-0.97) node { 3'};
	\draw[color=ccqqqq] (3.272,1.52) node { $\mathbf{m_{i+3}}$};
	\draw [fill=qqqqff] (6.064,1.048) circle (2pt);
	\draw[color=qqqqff] (6.082,1.41) node { 4'};
	\draw[color=ccqqqq] (5.292,0.64) node { $\mathbf{m_{i+2}}$};
	\draw[color=ccqqqq] (6.792,-0.08) node { $\mathbf{m_{i+1}}$};
	\end{scriptsize}
	\end{tikzpicture}
	\begin{tikzpicture}[x={19.0pt},y={20.0pt}] 
	\draw (1.,4.)-- (6.0,4.);
	\draw (1.,-0.5)-- (6.0,-0.5);
	\draw [line width=1.6pt] (1.54,-0.5)-- (2.,1.);
	\draw [line width=1.6pt] (2.,1.)-- (5.,2.);
	\draw [line width=1.6pt] (5.,2.)-- (5.7,-0.5);
	\draw [line width=2.pt,color=yqqqyq] (5.,2.)-- (4.42,4.);
	\draw [line width=1.6pt,color=yqqqyq] (2.,1.)-- (2.58,3.34);
	\draw [line width=1.6pt,color=yqqqyq] (2.58,3.34)-- (4.42,4.);
	\draw [line width=1.6pt,color=ffqqqq] (1.54,-0.5)-- (5.368902077151335,0.6824925816023741);
	\draw (-0.7,0.0) node[anchor=north west] {$\mathbf{V_{min}}$};
	\draw (-0.7,4.52) node[anchor=north west] {$\mathbf{V_{max}}$};
	\begin{scriptsize}
	\draw [fill=qqqqff] (1.54,-0.5) circle (2.0pt);
	\draw[color=qqqqff] (1.36,-1.07) node {$1$};
	\draw [fill=qqqqff] (2.,1.) circle (2.0pt);
	\draw[color=qqqqff] (1.58,0.81) node {$2$};
	\draw[color=black] (1.45,1.38) node {$m_1$};
	\draw [fill=qqqqff] (5.,2.) circle (2.0pt);
	\draw[color=qqqqff] (5.3,2.07) node {$3$};
	\draw[color=black] (3.61,1.16) node {$m_2$};
	\draw [fill=qqqqff] (5.7,-0.5) circle (2.0pt);
	\draw[color=qqqqff] (5.86,-1.03) node {$4$};
	\draw[color=black] (5.59,1.38) node {$m_3$};
	\draw [fill=qqqqff] (4.42,4.) circle (2.0pt);
	\draw[color=qqqqff] (4.72,4.19) node {$7$};
	\draw [fill=qqqqff] (2.58,3.34) circle (2.0pt);
	\draw[color=qqqqff] (2.24,3.27) node {$6$};
	\draw[color=yqqqyq] (3.49,3.24) node {$m_2$};
	\draw [fill=qqqqff] (5.368902077151335,0.6824925816023741) circle (2.0pt);
	\draw[color=qqqqff] (5.66,0.73) node {$5$};
	\draw[color=ffqqqq] (3.95,-0.14) node {$m_2$};
	\end{scriptsize}
	\end{tikzpicture}
	\caption{\label{fig:shift-wedge} On the left, an example of applying the shift-down operation to timed actions $m_{i+1}, m_{i+2}$. These actions are rearranged to move after point 5, which becomes point 3' (i.e.\ following timed action $m_{i+3}$).
		On the right, an example of applying the \Wedge operation to three timed actions $m_1,m_2,m_3$.
		This operation is a (parallel) translation of the action $m_2$, which changes the time duration of each of theses actions.
		After this operation either the $m_2$ line touches $\vmin$, which would remove $m_1$ from the schedule, 
		or the $m_2$ line touches $\vmax$, which would change a state along the run of the schedule to be at the border.
	}
\end{figure}
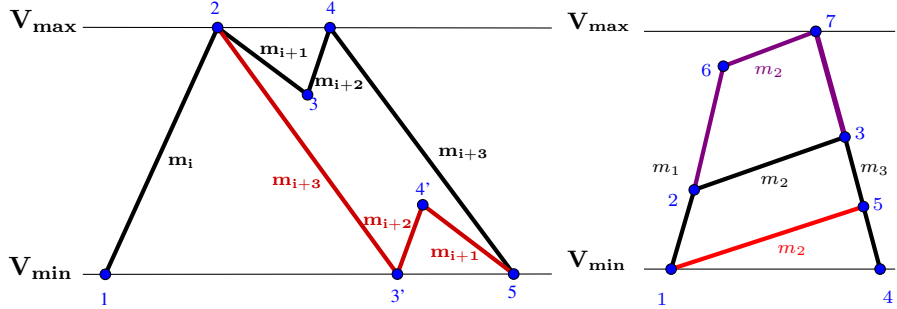 

Next is the {\em shift-down} operation. 
We can see an example of applying this operation in Figure \ref{fig:shift-wedge}.
Intuitively, it can rearrange any subsequence of timed actions 
that start and end at the same state and move 
them after any timed action that ends at $\vmin$.
The most complicated operation we define is the {\em \Wedge} operation. 
It acts on three consecutive timed actions in a safe schedule and simultaneously 
shrinks the middle action while extending the other two, or stretches the middle action while shrinking the other two.
We can see its behaviour in Figure~\ref{fig:shift-wedge}. Intuitively, it moves the timed action $m_2$ parallelly up or down,
until either the timed action $m_1$ is removed or $m_2$ ends at $\vmax$.
The direction depends on the cost gradient, but as the cost delta function of this operation is linear, %
one of these directions is cost-nonincreasing.

Finally, we define the {\em resize} operation that will be used the most in our procedure.
The resize operation requires one parameter $t \in \Real$ and can act on any two consecutive timed actions in a safe schedule.
Intuitively, if $t < 0$, this operation decreases the total time of this pair of timed actions by $|t|$ while changing only the middle state between these two timed actions along the run of the schedule. 
If $t > 0$, this operation increases the duration of this pair of timed actions by $t$ while again changing only the state between them along the run.
If $t>0$ then we will also refer to this operation as the {\em stretch} operation and if $t < 0$ as the {\em shrink} operation with parameter $-t > 0$. 
If the stretch and shrink operations are simultaneously applied with the same parameter $t$ to two non-overlapping pairs of timed actions, the result is a safe schedule with the same time horizon as before, but with a possibly different total cost. 
We will call a {\em \flexi} any subsequence of length $2$ in a safe schedule such that
both shrink and stretch operations can be applied to it for some $t > 0$ without compromising its safety.
A simultaneous application of these two operations to \flexis is demonstrated in Figure \ref{fig:shrink-stretch-up-up} and \ref{fig:shrink-stretch-up-down}.

Consider two non-overlapping \flexis at positions $i$ and $j$ in a safe schedule $\sigma$.
Let $\sigma' = \Resize(\sigma,i,t)$ be the resulting schedule of applying the resize operation with parameter $t$ to the $i$-th and $i+1$-th timed actions in $\sigma$ and $\maxresize(\sigma,i)$ be the maximal closed interval from which $t$ can be picked to ensure that $\sigma'$ is safe.
Similarly, let $\sigma'' = \Resize(\sigma,j,-t)$ and $\sigma''' = \Resize(\Resize(\sigma,i,t),j,-t))$.
Note that $\sigma'''$ has the same time horizon as $\sigma$ and is safe as long as 
$t \in \maxresize(\sigma,i) \cap \maxresize(\sigma,j)$ and let us denote this closed interval by $I$.
Furthermore, $\pi(\sigma''') - \pi(\sigma) = \pi(\sigma') - \pi(\sigma) + \pi(\sigma'') - \pi(\sigma)$ because
the two \flexis did not overlap.
As it is shown in Appendix \ref{app:formal-def}, 
both $\pi(\sigma') - \pi(\sigma)$ and $\pi(\sigma'') - \pi(\sigma)$ are linear functions in $t$ in the interior of $I$. As a result, $\pi(\sigma''') - \pi(\sigma)$ is also a linear function in $t$ and so its minimum value is achieved at one of the endpoints of $I$. Also, at such an endpoint, one of the time actions in these two \flexis will disappear and as a result the total cost would be reduced even further. It follows, that there is an endpoint of $I$ such that selecting it as $t$ will not increase the cost of the schedule, but it will remove a \flexi from $\sigma$.
As the \flatm timed action and the last timed action in a schedule can have flexible time delay,
we can also define the resize operation for them in a similar way (see Appendix \ref{app:formal-def}). 
As a result, we can apply the resize operation with parameter $t$ to any of these (including a \flexi) and with parameter $-t$ to the other. 
Reasoning as above, there is a value for $t$ such that the cost of the resulting schedule does not increase,
the schedule remains safe, and at least one of the timed actions is 
removed from $\sigma$ or one more state along the run of $\sigma$ becomes $\vmin$ or $\vmax$.

\begin{figure}
	\centering
	\resizebox{\textwidth}{!}{
		\begin{tikzpicture}[line cap=round,line join=round,>=triangle 45,x=1.0cm,y=1.0cm]
		\draw (-3.,4.5)-- (2.5,4.5);
		\draw (-3.,2.)-- (2.5,2.);
		\draw [dash pattern=on 5pt off 5pt] (-3.,5.)-- (-3.,1.5);
		\draw [dash pattern=on 5pt off 5pt] (1.,5.)-- (1.,1.5);
		\draw (-3.2,1.74) node[anchor=north west] {\large $\mathbf{t=0}$};
		\draw (3.52,4.52)-- (9.02,4.52);
		\draw (3.52,2.02)-- (9.02,2.02);
		\draw [dash pattern=on 5pt off 5pt] (3.52,5.02)-- (3.52,1.52);
		\draw [dash pattern=on 5pt off 5pt] (7.5,5.)-- (7.5,1.5);
		\draw (9.98,4.5)-- (15.48,4.5);
		\draw (9.98,2.)-- (15.48,2.);
		\draw [dash pattern=on 5pt off 5pt] (9.98,5.)-- (9.98,1.5);
		\draw [dash pattern=on 5pt off 5pt] (14.,5.)-- (14.,1.5);
		\draw (-2.98,0.02)-- (2.52,0.02);
		\draw (-2.98,-2.48)-- (2.52,-2.48);
		\draw [dash pattern=on 5pt off 5pt] (-2.98,0.52)-- (-2.98,-2.98);
		\draw [dash pattern=on 5pt off 5pt] (1.,0.5)-- (1.,-3.);
		\draw (3.48,-0.02)-- (8.98,-0.02);
		\draw (3.48,-2.52)-- (8.98,-2.52);
		\draw [dash pattern=on 5pt off 5pt] (3.48,0.48)-- (3.48,-3.02);
		\draw [dash pattern=on 5pt off 5pt] (7.5,0.5)-- (7.5,-3.);
		\draw (10.,0.)-- (15.5,0.);
		\draw (10.,-2.5)-- (15.5,-2.5);
		\draw [dash pattern=on 5pt off 5pt] (10.,0.5)-- (10.,-3.);
		\draw [dash pattern=on 5pt off 5pt] (14.,0.5)-- (14.,-3.);
		\draw (3.28,1.78) node[anchor=north west] {\large $\mathbf{t=0}$};
		\draw (9.74,1.76) node[anchor=north west] {\large $\mathbf{t=0}$};
		\draw (-3.24,-2.78) node[anchor=north west] {\large $\mathbf{t=0}$};
		\draw (3.24,-2.8) node[anchor=north west] {\large $\mathbf{t=0}$};
		\draw (9.78,-2.82) node[anchor=north west] {\large $\mathbf{t=0}$};
		\draw [line width=1.6pt] (-3.,3.7)-- (-1.3,3.7);
		\draw [line width=1.6pt] (-1.3,3.7)-- (1.,2.);
		\draw [line width=1.6pt] (3.52,3.7)-- (7.5,2.02);
		\draw [line width=1.6pt] (9.98,2.7)-- (11.22,4.04);
		\draw [line width=1.6pt] (11.22,4.04)-- (14.,2.);
		\draw [line width=1.6pt] (-2.98,-1.5)-- (-1.76,-1.52);
		\draw [line width=1.6pt] (-1.76,-1.52)-- (-0.76,0.02);
		\draw [line width=1.6pt] (-0.76,0.02)-- (1.,-2.48);
		\draw [line width=1.6pt] (3.48,-1.64)-- (5.08,-0.02);
		\draw [line width=1.6pt] (5.08,-0.02)-- (7.5,-2.52);
		\draw [line width=1.6pt] (10.,-0.78)-- (11.,-1.74);
		\draw [line width=1.6pt] (11.,-1.74)-- (12.,0.);
		\draw [line width=1.6pt] (12.,0.)-- (14.,-2.5);
		\draw (-4.22,4.86) node[anchor=north west] {\large $\mathbf{V_{max}}$};
		\draw (-4.18,0.38) node[anchor=north west] {\large $\mathbf{V_{max}}$};
		\draw (-4.2,2.38) node[anchor=north west] {\large $\mathbf{V_{min}}$};
		\draw (-4.18,-2.12) node[anchor=north west] {\large $\mathbf{V_{min}}$};
		\draw (-1.18,1.16) node[anchor=north west] {\large \textbf{(a)}};
		\draw (5.3,1.18) node[anchor=north west] {\large \textbf{(b)}};
		\draw (11.82,1.24) node[anchor=north west] {\large \textbf{(c)}};
		\draw (-1.14,-3.3) node[anchor=north west] {\large \textbf{(d)}};
		\draw (5.32,-3.3) node[anchor=north west] {\large \textbf{(e)}};
		\draw (11.86,-3.32) node[anchor=north west] {\large \textbf{(f)}};
		\draw (-2.98,-4.48)-- (2.52,-4.48);
		\draw (-2.98,-6.98)-- (2.52,-6.98);
		\draw [dash pattern=on 5pt off 5pt] (-2.98,-3.98)-- (-2.98,-7.48);
		\draw [dash pattern=on 5pt off 5pt] (1.,-4.)-- (1.,-7.5);
		\draw (3.48,-4.52)-- (8.98,-4.52);
		\draw (3.48,-7.02)-- (8.98,-7.02);
		\draw [dash pattern=on 5pt off 5pt] (3.48,-4.02)-- (3.48,-7.52);
		\draw [dash pattern=on 5pt off 5pt] (7.5,-4.)-- (7.5,-7.5);
		\draw (10.,-4.5)-- (15.5,-4.5);
		\draw (10.,-7.)-- (15.5,-7.);
		\draw [dash pattern=on 5pt off 5pt] (10.,-4.)-- (10.,-7.5);
		\draw [dash pattern=on 5pt off 5pt] (14.,-4.)-- (14.,-7.5);
		\draw (-3.24,-7.28) node[anchor=north west] {\large $\mathbf{t=0}$};
		\draw (3.24,-7.3) node[anchor=north west] {\large $\mathbf{t=0}$};
		\draw (9.78,-7.32) node[anchor=north west] {\large $\mathbf{t=0}$};
		\draw [line width=1.6pt] (-2.98,-6.32)-- (-2.26,-5.2);
		\draw [line width=1.6pt] (-2.26,-5.2)-- (-0.76,-4.48);
		\draw [line width=1.6pt] (-0.76,-4.48)-- (1.,-6.98);
		\draw [line width=1.6pt] (10.,-6.18)-- (11.2,-4.5);
		\draw (-1.14,-7.8) node[anchor=north west] {\large \textbf{(g)}};
		\draw (5.32,-7.8) node[anchor=north west] {\large \textbf{(h)}};
		\draw (11.86,-7.82) node[anchor=north west] {\large \textbf{(i)}};
		\draw [line width=1.6pt] (3.48,-5.14)-- (4.72,-6.46);
		\draw [line width=1.6pt] (4.72,-6.46)-- (7.5,-7.02);
		\draw [line width=1.6pt] (11.2,-4.5)-- (11.82,-5.74);
		\draw [line width=1.6pt] (11.82,-5.74)-- (14.,-7.);
		\draw (-4.16,-4.18) node[anchor=north west] {\large $\mathbf{V_{max}}$};
		\draw (-4.12,-6.6) node[anchor=north west] {\large $\mathbf{V_{min}}$};
		\begin{scriptsize}
		\draw [fill=qqqqff] (2.02,3.22) circle (1.0pt);
		\draw [fill=qqqqff] (1.7,3.22) circle (1.0pt);
		\draw [fill=qqqqff] (1.4,3.22) circle (1.0pt);
		\draw [fill=qqqqff] (8.54,3.24) circle (1.0pt);
		\draw [fill=qqqqff] (8.22,3.24) circle (1.0pt);
		\draw [fill=qqqqff] (7.92,3.24) circle (1.0pt);
		\draw [fill=qqqqff] (15.,3.22) circle (1.0pt);
		\draw [fill=qqqqff] (14.68,3.22) circle (1.0pt);
		\draw [fill=qqqqff] (14.38,3.22) circle (1.0pt);
		\draw [fill=qqqqff] (2.04,-1.26) circle (1.0pt);
		\draw [fill=qqqqff] (1.72,-1.26) circle (1.0pt);
		\draw [fill=qqqqff] (1.42,-1.26) circle (1.0pt);
		\draw [fill=qqqqff] (8.5,-1.3) circle (1.0pt);
		\draw [fill=qqqqff] (8.18,-1.3) circle (1.0pt);
		\draw [fill=qqqqff] (7.88,-1.3) circle (1.0pt);
		\draw [fill=qqqqff] (15.02,-1.28) circle (1.0pt);
		\draw [fill=qqqqff] (14.7,-1.28) circle (1.0pt);
		\draw [fill=qqqqff] (14.4,-1.28) circle (1.0pt);
		\draw [fill=qqqqff] (-3.,3.7) circle (2.0pt);
		\draw[color=qqqqff] (-3.26,3.63) node {\large $1$};
		\draw [fill=qqqqff] (-1.3,3.7) circle (2.0pt);
		\draw[color=qqqqff] (-0.96,3.73) node {\large $2$};
		\draw[color=black] (-2.25,3.86) node {\large $m_1$};
		\draw [fill=qqqqff] (1.,2.) circle (2.0pt);
		\draw[color=qqqqff] (1.28,1.61) node {\large $3$};
		\draw[color=black] (0.17,3.) node {\large $m_2$};
		\draw [fill=qqqqff] (3.52,3.7) circle (2.0pt);
		\draw[color=qqqqff] (3.26,3.63) node {\large $1$};
		\draw [fill=qqqqff] (7.5,2.02) circle (2.0pt);
		\draw[color=qqqqff] (7.82,1.65) node {\large $2$};
		\draw[color=black] (5.77,2.96) node {\large $m_1$};
		\draw [fill=qqqqff] (9.98,2.7) circle (2.0pt);
		\draw[color=qqqqff] (9.7,2.63) node {\large $1$};
		\draw [fill=qqqqff] (11.22,4.04) circle (2.0pt);
		\draw[color=qqqqff] (11.56,4.09) node {\large $2$};
		\draw[color=black] (10.87,3.06) node {\large $m_1$};
		\draw [fill=qqqqff] (14.,2.) circle (2.0pt);
		\draw[color=qqqqff] (14.26,1.61) node {\large $3$};
		\draw[color=black] (13.1,3.) node {\large $m_2$};
		\draw [fill=qqqqff] (-2.98,-1.5) circle (2.0pt);
		\draw[color=qqqqff] (-3.24,-1.61) node {\large $1$};
		\draw [fill=qqqqff] (-1.76,-1.52) circle (2.0pt);
		\draw[color=qqqqff] (-1.64,-1.83) node {\large $2$};
		\draw[color=black] (-2.43,-1.28) node {\large $m_1$};
		\draw [fill=qqqqff] (-0.76,0.02) circle (2.0pt);
		\draw[color=qqqqff] (-0.69,0.29) node {\large $3$};
		\draw[color=black] (-1.01,-1.08) node {\large $m_2$};
		\draw [fill=qqqqff] (1.,-2.48) circle (2.0pt);
		\draw[color=qqqqff] (1.26,-2.8) node {\large $4$};
		\draw[color=black] (0.55,-1.2) node {\large $m_3$};
		\draw [fill=qqqqff] (3.48,-1.64) circle (2.0pt);
		\draw[color=qqqqff] (3.18,-1.75) node {\large $1$};
		\draw [fill=qqqqff] (5.08,-0.02) circle (2.0pt);
		\draw[color=qqqqff] (5.15,0.25) node {\large $2$};
		\draw[color=black] (4.41,-1.26) node {\large $m_1$};
		\draw [fill=qqqqff] (7.5,-2.52) circle (2.0pt);
		\draw[color=qqqqff] (7.76,-2.89) node {\large $3$};
		\draw[color=black] (6.7,-1.28) node {\large $m_2$};
		\draw [fill=qqqqff] (10.,-0.78) circle (2.0pt);
		\draw[color=qqqqff] (9.66,-0.95) node {\large $1$};
		\draw [fill=qqqqff] (11.,-1.74) circle (2.0pt);
		\draw[color=qqqqff] (11.02,-2.09) node {\large $2$};
		\draw[color=black] (10.83,-1.12) node {\large $m_1$};
		\draw [fill=qqqqff] (12.,0.) circle (2.0pt);
		\draw[color=qqqqff] (12.05,0.26) node {\large $3$};
		\draw[color=black] (11.79,-1.22) node {\large $m_2$};
		\draw [fill=qqqqff] (14.,-2.5) circle (2.0pt);
		\draw[color=qqqqff] (14.24,-2.85) node {\large $4$};
		\draw[color=black] (13.47,-1.28) node {\large $m_3$};
		\draw [fill=qqqqff] (2.04,-5.76) circle (1.0pt);
		\draw [fill=qqqqff] (1.72,-5.76) circle (1.0pt);
		\draw [fill=qqqqff] (1.42,-5.76) circle (1.0pt);
		\draw [fill=qqqqff] (8.5,-5.8) circle (1.0pt);
		\draw [fill=qqqqff] (8.18,-5.8) circle (1.0pt);
		\draw [fill=qqqqff] (7.88,-5.8) circle (1.0pt);
		\draw [fill=qqqqff] (15.02,-5.78) circle (1.0pt);
		\draw [fill=qqqqff] (14.7,-5.78) circle (1.0pt);
		\draw [fill=qqqqff] (14.4,-5.78) circle (1.0pt);
		\draw [fill=qqqqff] (-2.98,-6.32) circle (2.0pt);
		\draw[color=qqqqff] (-3.24,-6.43) node {\large $1$};
		\draw [fill=qqqqff] (-2.26,-5.2) circle (2.0pt);
		\draw[color=qqqqff] (-2.44,-4.91) node {\large $2$};
		\draw[color=black] (-2.31,-6.06) node {\large $m_1$};
		\draw [fill=qqqqff] (-0.76,-4.48) circle (2.0pt);
		\draw[color=qqqqff] (-0.74,-4.15) node {\large $3$};
		\draw[color=black] (-1.49,-5.28) node {\large $m_2$};
		\draw [fill=qqqqff] (1.,-6.98) circle (2.0pt);
		\draw[color=qqqqff] (1.26,-7.30) node {\large $4$};
		\draw[color=black] (0.55,-5.7) node {\large $m_3$};
		\draw [fill=qqqqff] (3.48,-5.14) circle (2.0pt);
		\draw[color=qqqqff] (3.18,-5.25) node {\large $1$};
		\draw [fill=qqqqff] (7.5,-7.02) circle (2.0pt);
		\draw[color=qqqqff] (7.76,-7.39) node {\large $3$};
		\draw [fill=qqqqff] (10.,-6.18) circle (2.0pt);
		\draw[color=qqqqff] (9.66,-6.35) node {\large $1$};
		\draw [fill=qqqqff] (11.2,-4.5) circle (2.0pt);
		\draw[color=qqqqff] (11.2,-4.2) node {\large $2$};
		\draw[color=black] (10.35,-5.18) node {\large $m_1$};
		\draw [fill=qqqqff] (14.,-7.) circle (2.0pt);
		\draw[color=qqqqff] (14.24,-7.35) node {\large $4$};
		\draw [fill=qqqqff] (4.72,-6.46) circle (2.0pt);
		\draw[color=qqqqff] (4.48,-6.73) node {\large $2$};
		\draw[color=black] (4.53,-5.82) node {\large $m_1$};
		\draw[color=black] (5.71,-6.42) node {\large $m_2$};
		\draw [fill=qqqqff] (11.82,-5.74) circle (2.0pt);
		\draw[color=qqqqff] (11.58,-6.07) node {\large $3$};
		\draw[color=black] (11.9,-5.14) node {\large $m_2$};
		\draw[color=black] (13.2,-6.28) node {\large $m_3$};
		\end{scriptsize}
		\end{tikzpicture}
	}
	\caption{\label{fig:head} Ten possible head patterns:
		(a) \FD (b) \D (c) \PUD (d) \FUD (e) \UD (f) \PDUD 
		(g) \PUUD (h) \PDD (i) \UPDD and (j) empty (not depicted).
	}
\end{figure}
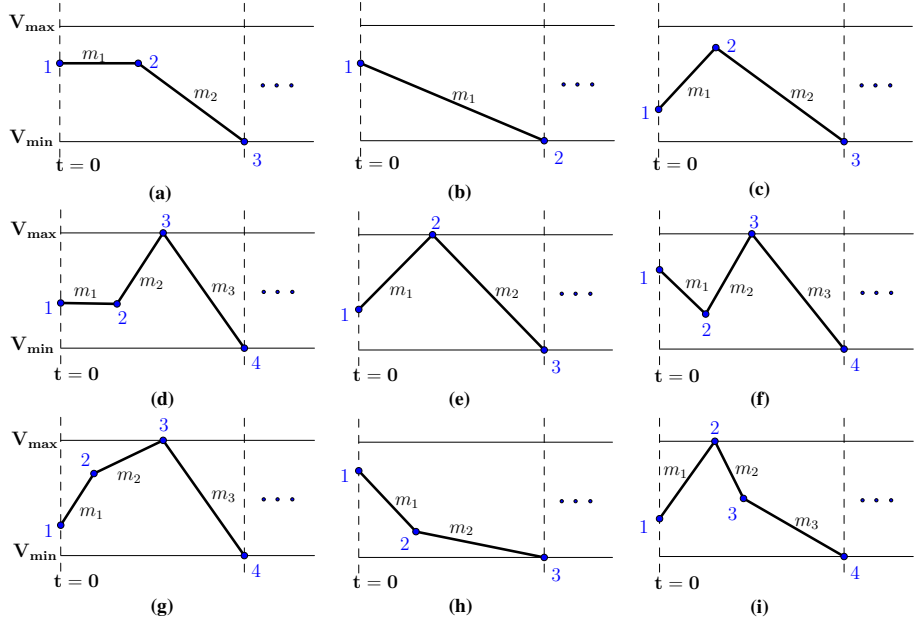

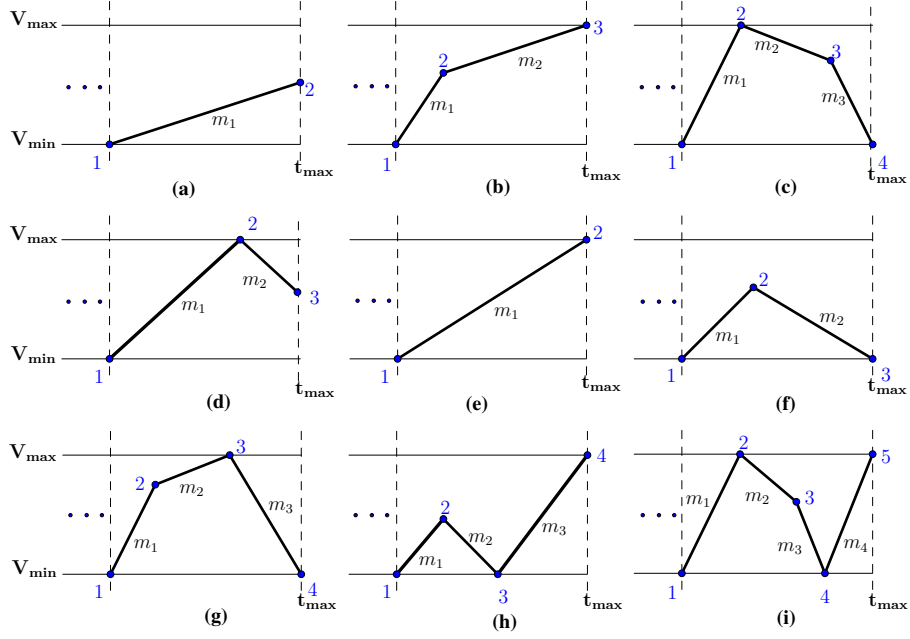
\begin{figure}
	\centering
	\resizebox{\textwidth}{!}{
		\begin{tikzpicture}[line cap=round,line join=round,>=triangle 45,x=1.0cm,y=1.0cm]
		\draw (-3.,4.5)-- (2.,4.5);
		\draw (-3.,2.)-- (2.,2.);
		\draw [dash pattern=on 5pt off 5pt] (2.,5.)-- (2.,1.5);
		\draw [dash pattern=on 5pt off 5pt] (-2.,5.)-- (-2.,1.5);
		\draw (-4.2,4.9) node[anchor=north west] {\large $ \mathbf{V_{max}}$};
		\draw (-4.18,2.38) node[anchor=north west] {\large $ \mathbf{V_{min}}$};
		\draw (1.76,1.78) node[anchor=north west] {\large $\mathbf{t_{max}}$};
		\draw (-4.2,0.4) node[anchor=north west] {\large $\mathbf{V_{max}}$};
		\draw (-4.18,-2.12) node[anchor=north west] {\large $\mathbf{V_{min}}$};
		\draw [line width=1.6pt] (-2.,2.)-- (2.,3.3);
		\draw (-0.8,1.34) node[anchor=north west] {\large \textbf{(a)}};
		\draw (-3.,0.)-- (2.,0.);
		\draw (-3.,-2.5)-- (2.,-2.5);
		\draw [dash pattern=on 5pt off 5pt] (1.96,0.38)-- (1.96,-3.12);
		\draw [dash pattern=on 5pt off 5pt] (-2.,0.5)-- (-2.,-3.);
		\draw (1.78,-2.86) node[anchor=north west] {\large $\mathbf{t_{max}}$};
		\draw [line width=2.pt] (-2.,-2.5)-- (0.74,0.);
		\draw [line width=1.6pt] (0.74,0.)-- (1.94,-1.1);
		\draw (-0.08,-3.12) node[anchor=north west] {\large \textbf{(d)}};
		\draw (3.04,0.)-- (8.,0.);
		\draw (3.04,-2.5)-- (8.,-2.5);
		\draw [dash pattern=on 5pt off 5pt] (8.,0.5)-- (8.,-3.);
		\draw [dash pattern=on 5pt off 5pt] (4.04,0.5)-- (4.04,-3.);
		\draw (7.78,-2.74) node[anchor=north west] {\large $\mathbf{t_{max}}$};
		\draw [line width=1.6pt] (4.04,-2.5)-- (8.,0.);
		\draw (5.36,-3.18) node[anchor=north west] {\large \textbf{(e)}};
		\draw (-2.98,-4.52)-- (2.02,-4.52);
		\draw (-2.98,-7.02)-- (2.02,-7.02);
		\draw [dash pattern=on 5pt off 5pt] (2.02,-4.02)-- (2.02,-7.52);
		\draw [dash pattern=on 5pt off 5pt] (-1.98,-4.02)-- (-1.98,-7.52);
		\draw (1.82,-7.3) node[anchor=north west] {\large $\mathbf{t_{max}}$};
		\draw [line width=1.6pt] (2.02,-7.02)-- (0.52,-4.52);
		\draw [line width=1.6pt] (0.52,-4.52)-- (-1.04,-5.14);
		\draw [line width=1.6pt] (-1.04,-5.14)-- (-1.98,-7.02);
		\draw (3.02,-4.52)-- (8.02,-4.52);
		\draw (3.02,-7.02)-- (8.02,-7.02);
		\draw [dash pattern=on 5pt off 5pt] (8.02,-4.02)-- (8.02,-7.52);
		\draw [dash pattern=on 5pt off 5pt] (4.02,-4.02)-- (4.02,-7.52);
		\draw (7.82,-7.3) node[anchor=north west] {\large $\mathbf{t_{max}}$};
		\draw [line width=2.pt] (4.02,-7.02)-- (5.,-5.86);
		\draw [line width=1.6pt] (5.,-5.86)-- (6.14,-7.02);
		\draw [line width=2.pt] (6.14,-7.02)-- (8.02,-4.52);
		\draw (-0.12,-7.62) node[anchor=north west] {\large \textbf{(g)}};
		\draw (5.92,-7.74) node[anchor=north west] {\large \textbf{(h)}};
		\draw (3.,4.5)-- (8.,4.5);
		\draw (3.,2.)-- (8.,2.);
		\draw [dash pattern=on 5pt off 5pt] (4.,5.)-- (4.,1.5);
		\draw [dash pattern=on 5pt off 5pt] (8.,5.)-- (8.,1.5);
		\draw (9.,4.5)-- (14.,4.5);
		\draw (9.,2.)-- (14.,2.);
		\draw [dash pattern=on 5pt off 5pt] (13.96,4.88)-- (13.96,1.38);
		\draw [dash pattern=on 5pt off 5pt] (10.,5.)-- (10.,1.5);
		\draw (7.78,1.7) node[anchor=north west] {\large $\mathbf{t_{max}}$};
		\draw (13.78,1.64) node[anchor=north west] {\large $\mathbf{t_{max}}$};
		\draw [line width=1.6pt] (4.,2.)-- (5.,3.5);
		\draw [line width=1.6pt] (5.,3.5)-- (8.,4.5);
		\draw [line width=1.6pt] (10.,2.)-- (11.24,4.5);
		\draw [line width=1.6pt] (11.24,4.5)-- (13.12,3.76);
		\draw [line width=1.6pt] (13.12,3.76)-- (14.,2.);
		\draw (5.78,1.38) node[anchor=north west] {\large \textbf{(b)}};
		\draw (11.84,1.38) node[anchor=north west] {\large \textbf{(c)}};
		\draw (9.,0.)-- (14.,0.);
		\draw (9.,-2.5)-- (14.,-2.5);
		\draw [dash pattern=on 5pt off 5pt] (10.,0.5)-- (10.,-3.);
		\draw [dash pattern=on 5pt off 5pt] (14.,0.5)-- (14.,-3.);
		\draw (13.78,-2.8) node[anchor=north west] {\large $\mathbf{t_{max}}$};
		\draw [line width=1.6pt] (10.,-2.5)-- (11.5,-1.);
		\draw [line width=1.6pt] (11.5,-1.)-- (14.,-2.5);
		\draw (11.86,-3.16) node[anchor=north west] {\large \textbf{(f)}};
		\draw (9.,-4.5)-- (14.,-4.5);
		\draw (9.,-7.)-- (14.,-7.);
		\draw [dash pattern=on 5pt off 5pt] (10.,-4.)-- (10.,-7.5);
		\draw [line width=1.6pt] (10.,-7.)-- (11.22,-4.5);
		\draw [line width=1.6pt] (11.22,-4.5)-- (12.4,-5.5);
		\draw [line width=1.6pt] (12.4,-5.5)-- (13.,-7.);
		\draw [line width=1.6pt] (13.,-7.)-- (14.,-4.5);
		\draw [dash pattern=on 5pt off 5pt] (14.,-4.)-- (14.,-7.5);
		\draw (11.88,-7.66) node[anchor=north west] {\large \textbf{(i)}};
		\draw (-4.22,-4.18) node[anchor=north west] {\large $\mathbf{V_{max}}$};
		\draw (-4.16,-6.62) node[anchor=north west] {\large $\mathbf{V_{min}}$};
		\draw (13.8,-7.3) node[anchor=north west] {\large $\mathbf{t_{max}}$};
		\begin{scriptsize}
		\draw [fill=qqqqff] (-2.,2.) circle (2.0pt);
		\draw[color=qqqqff] (-2.26,1.63) node {\large $1$};
		\draw [fill=qqqqff] (2.,3.3) circle (2.0pt);
		\draw[color=qqqqff] (2.2,3.15) node {\large $2$};
		\draw[color=black] (0.41,2.46) node {\large $m_1$};
		\draw [fill=qqqqff] (-2.86,3.2) circle (1.0pt);
		\draw [fill=qqqqff] (-2.52,3.2) circle (1.0pt);
		\draw [fill=qqqqff] (-2.2,3.2) circle (1.0pt);
		\draw [fill=qqqqff] (-2.,-2.5) circle (2.0pt);
		\draw[color=qqqqff] (-2.22,-2.87) node {\large $1$};
		\draw [fill=qqqqff] (0.74,0.) circle (2.0pt);
		\draw[color=qqqqff] (1.,0.33) node {\large $2$};
		\draw[color=black] (-0.23,-1.42) node {\large $m_1$};
		\draw [fill=qqqqff] (1.94,-1.1) circle (2.0pt);
		\draw[color=qqqqff] (2.3,-1.21) node {\large $3$};
		\draw[color=black] (1.05,-.9) node {\large $m_2$};
		\draw [fill=qqqqff] (-2.86,-1.3) circle (1.0pt);
		\draw [fill=qqqqff] (-2.52,-1.3) circle (1.0pt);
		\draw [fill=qqqqff] (-2.2,-1.3) circle (1.0pt);
		\draw [fill=qqqqff] (4.04,-2.5) circle (2.0pt);
		\draw[color=qqqqff] (3.78,-2.87) node {\large $1$};
		\draw [fill=qqqqff] (8.,0.) circle (2.0pt);
		\draw[color=qqqqff] (8.24,0.15) node {\large $2$};
		\draw[color=black] (6.33,-1.54) node {\large $m_1$};
		\draw [fill=qqqqff] (3.2,-1.28) circle (1.0pt);
		\draw [fill=qqqqff] (3.54,-1.28) circle (1.0pt);
		\draw [fill=qqqqff] (3.86,-1.28) circle (1.0pt);
		\draw [fill=qqqqff] (2.02,-7.02) circle (2.0pt);
		\draw[color=qqqqff] (2.26,-7.35) node {\large $4$};
		\draw [fill=qqqqff] (0.52,-4.52) circle (2.0pt);
		\draw[color=qqqqff] (0.76,-4.33) node {\large $3$};
		\draw[color=black] (1.6,-5.62) node {\large $m_3$};
		\draw [fill=qqqqff] (-1.04,-5.14) circle (2.0pt);
		\draw[color=qqqqff] (-1.36,-5.11) node {\large $2$};
		\draw[color=black] (-0.27,-5.26) node {\large $m_2$};
		\draw [fill=qqqqff] (-1.98,-7.02) circle (2.0pt);
		\draw[color=qqqqff] (-2.22,-7.39) node {\large $1$};
		\draw[color=black] (-1.23,-6.42) node {\large $m_1$};
		\draw [fill=qqqqff] (4.02,-7.02) circle (2.0pt);
		\draw[color=qqqqff] (3.82,-7.39) node {\large $1$};
		\draw [fill=qqqqff] (5.,-5.86) circle (2.0pt);
		\draw[color=qqqqff] (5.04,-5.61) node {\large $2$};
		\draw[color=black] (4.75,-6.76) node {\large $m_1$};
		\draw [fill=qqqqff] (6.14,-7.02) circle (2.0pt);
		\draw[color=qqqqff] (6.26,-7.49) node {\large $3$};
		\draw[color=black] (5.79,-6.28) node {\large $m_2$};
		\draw [fill=qqqqff] (8.02,-4.52) circle (2.0pt);
		\draw[color=qqqqff] (8.32,-4.53) node {\large $4$};
		\draw[color=black] (7.33,-6.08) node {\large $m_3$};
		\draw [fill=qqqqff] (-2.82,-5.78) circle (1.0pt);
		\draw [fill=qqqqff] (-2.48,-5.78) circle (1.0pt);
		\draw [fill=qqqqff] (-2.16,-5.78) circle (1.0pt);
		\draw [fill=qqqqff] (3.16,-5.78) circle (1.0pt);
		\draw [fill=qqqqff] (3.5,-5.78) circle (1.0pt);
		\draw [fill=qqqqff] (3.82,-5.78) circle (1.0pt);
		\draw [fill=qqqqff] (14.,2.) circle (2.0pt);
		\draw[color=qqqqff] (14.22,1.63) node {\large $4$};
		\draw [fill=qqqqff] (4.,2.) circle (2.0pt);
		\draw[color=qqqqff] (3.82,1.57) node {\large $1$};
		\draw [fill=qqqqff] (5.,3.5) circle (2.0pt);
		\draw[color=qqqqff] (5.0,3.75) node {\large $2$};
		\draw[color=black] (4.99,2.78) node {\large $m_1$};
		\draw [fill=qqqqff] (8.,4.5) circle (2.0pt);
		\draw[color=qqqqff] (8.25,4.45) node { \large $3$};
		\draw[color=black] (6.83,3.7) node {\large $m_2$};
		\draw [fill=qqqqff] (10.,2.) circle (2.0pt);
		\draw[color=qqqqff] (9.8,1.63) node {\large $1$};
		\draw [fill=qqqqff] (11.24,4.5) circle (2.0pt);
		\draw[color=qqqqff] (11.24,4.75) node { \large $2$};
		\draw[color=black] (11.13,3.32) node {\large $m_1$};
		\draw [fill=qqqqff] (13.12,3.76) circle (2.0pt);
		\draw[color=qqqqff] (13.24,3.95) node {\large $3$};
		\draw[color=black] (11.79,4.04) node {\large $m_2$};
		\draw[color=black] (13.19,2.94) node {\large $m_3$};
		\draw [fill=qqqqff] (3.16,3.22) circle (1.0pt);
		\draw [fill=qqqqff] (3.5,3.22) circle (1.0pt);
		\draw [fill=qqqqff] (3.82,3.22) circle (1.0pt);
		\draw [fill=qqqqff] (9.16,3.2) circle (1.0pt);
		\draw [fill=qqqqff] (9.5,3.2) circle (1.0pt);
		\draw [fill=qqqqff] (9.82,3.2) circle (1.0pt);
		\draw [fill=qqqqff] (10.,-2.5) circle (2.0pt);
		\draw[color=qqqqff] (9.78,-2.85) node {\large $1$};
		\draw [fill=qqqqff] (11.5,-1.) circle (2.0pt);
		\draw[color=qqqqff] (11.72,-0.83) node {\large $2$};
		\draw[color=black] (10.97,-2.06) node {\large $m_1$};
		\draw [fill=qqqqff] (14.,-2.5) circle (2.0pt);
		\draw[color=qqqqff] (14.28,-2.83) node {\large $3$};
		\draw[color=black] (13.15,-1.66) node {\large $m_2$};
		\draw [fill=qqqqff] (9.18,-1.3) circle (1.0pt);
		\draw [fill=qqqqff] (9.52,-1.3) circle (1.0pt);
		\draw [fill=qqqqff] (9.84,-1.3) circle (1.0pt);
		
		\draw [fill=qqqqff] (9.18,-5.78) circle (1.0pt);
		\draw [fill=qqqqff] (9.52,-5.78) circle (1.0pt);
		\draw [fill=qqqqff] (9.84,-5.78) circle (1.0pt);

		\draw [fill=qqqqff] (14.,-4.5) circle (2.0pt);
		\draw[color=qqqqff] (14.28,-4.55) node {\large $5$};
		\draw [fill=qqqqff] (10.,-7.) circle (2.0pt);
		\draw[color=qqqqff] (9.8,-7.41) node {\large $1$};
		\draw [fill=qqqqff] (11.22,-4.5) circle (2.0pt);
		\draw[color=qqqqff] (11.3,-4.25) node {\large $2$};
		\draw[color=black] (10.35,-5.48) node {\large $m_1$};
		\draw [fill=qqqqff] (12.4,-5.5) circle (2.0pt);
		\draw[color=qqqqff] (12.68,-5.41) node {\large $3$};
		\draw[color=black] (11.59,-5.42) node {\large $m_2$};
		\draw [fill=qqqqff] (13.,-7.) circle (2.0pt);
		\draw[color=qqqqff] (13.,-7.49) node {\large $4$};
		\draw[color=black] (12.27,-6.48) node {\large $m_3$};
		\draw[color=black] (13.65,-6.42) node {\large $m_4$};
		\end{scriptsize}
		\end{tikzpicture}
	}
	\caption{\label{fig:tail} Ten possible tail patterns:
		(a) \PU (b) \PUU (c) \UPDD (d) \UPD (e) \U (f) \PUD (g) \PUUD (h) \PUDU (i) \UPDDU and (j) empty (not depicted).
	}
\end{figure}

\begin{theorem}\label{thm:standard-form}
For every safe schedule $\sigma$ in a one-dimensional \name 
there exists a safe schedule $\sigma'$ 
whose head section matches one of the patterns in Figure \ref{fig:head},
tail section matches one of the patterns in Figure \ref{fig:tail},
and $\pi(\sigma') \leq \pi(\sigma)$ holds.
Furthermore, it suffices to consider only \patnum combinations of these head and tail patterns,
and the length of all of them is at most five.
\end{theorem}
\begin{proof}
We will repeatedly apply combination of shrink and stretch operations to \flexis
until we remove all non-overlapping ones.
Note that after each such an application either a timed action is removed or one more state along the run of $\sigma$ becomes equal to $\vmax$ or $\vmin$.
We claim that the following steps will transform $\sigma$ into a suitable $\sigma'$:
\begin{enumerate}
\item as long as there are at least one pair of non-overlapping \flexis then 
shrink one and stretch the other until a timed action is removed or
a new state at the border is created; 
\item once there is only one \flexi left or two overlapping ones, 
use the shift or shift-down operation to move them to the end of the schedule;
\item if the first timed action is flat, pair it with the remaining \flexi to
remove one of them using the shrink-stretch operation combination;
\item if the last state of $run(\sigma)$ is not at the border and 
a \flexi or flat timed action remains after the previous step, 
they should be paired with each other for the shrink-stretch operation combination;
\item if two overlapping \flexis exist, use the wedge operation to resolve them;
\item finally, if the tail section still does not follow any of the patterns,
apply the shift-down operation to the (unique) segment that starts and ends
at $\vmax$. 
\end{enumerate}
A graphical representation of this procedure when applied to an example schedule can be seen in Appendix
\ref{sec:example-transformation}.
It is easy to see that the first step of this procedure will stop eventually because 
$\sigma$ has a finite number of timed actions and states along its run.
The rest of the steps of this procedure just try to reduce the number of possibilities for 
the head and tail sections. 
Note that, apart from the initial state, there can be only one state, along the run
of the resulting $\sigma'$, which is not at the border.
This is because otherwise a shrink-stretch or wedge operation could still be applied.
Drawing all possible patterns with one point not at the border 
and eliminating the ones that are inter-reducible using one of these operations,
results in Figure \ref{fig:head} for the head section and Figure \ref{fig:tail} for the tail section.

If we try to combine all these head and tail pattern together then this would result in $10\cdot 10 = 100$ possible combinations.
However, as just mentioned, there can be only one point not at the border or a \flatm timed action in a schedule so these combinations of head and tail patterns can be reduced further.
In particular, any head pattern can be combined with tail patterns (e) and (j), but only
(b), (e), (j) head patterns can be combined with the remaining tail ones. 
Therefore, there are $10\cdot 2 + 3\cdot 8 = \patnum$ combined patterns and 
it is easy to check that none of them has length larger than five (this is important for the computational complexity stated in Theorem~\ref{thm:constant-factor}).\qed
\end{proof}

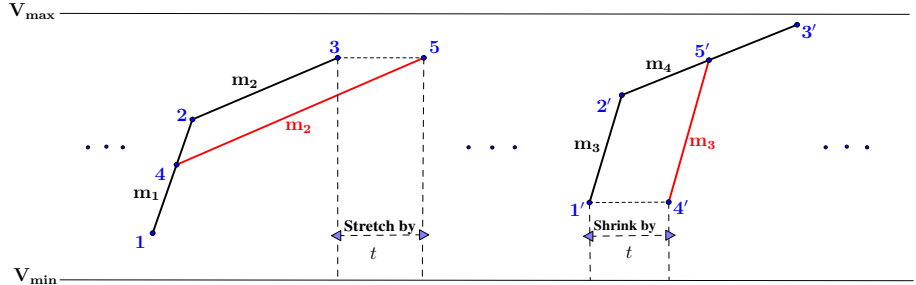
\begin{figure}
	\centering
	\resizebox{\columnwidth}{0.2\textheight}{
		\begin{tikzpicture}[line cap=round,line join=round,>=triangle 45,x=1.0cm,y=1.0cm]
		\draw (-1.,4.)-- (17.35224303060281,3.978710516200727);
		\draw (-1.,-2.)-- (17.437322834272518,-2.0194156425134575);
		\draw [line width=1.2pt] (1.,-0.9485267187798286)-- (1.86,1.6114732812201713);
		\draw [line width=1.2pt] (1.86,1.6114732812201713)-- (5.,3.);
		\draw [line width=1.2pt] (10.44,-0.26)-- (11.14,2.16);
		\draw [line width=1.2pt] (11.14,2.16)-- (14.927468626016232,3.744741056109039);
		\draw [line width=1.2pt,color=ffqqqq] (6.86,3.)-- (1.5176602862940822,0.5924155287932532);
		\draw [dash pattern=on 2pt off 2pt] (5.,3.)-- (6.86,3.);
		\draw (-2.116966043455037,-1.6234456367363685) node[anchor=north west] {\large $\mathbf{V_{min}}$};
		\draw (-2.222197711568463,4.332298057081143) node[anchor=north west] {\large $\mathbf{V_{max}}$};
		\draw [dash pattern=on 4pt off 4pt] (6.86,3.)-- (6.827787770946935,-2.0082431451896805);
		\draw [dash pattern=on 4pt off 4pt] (5.,3.)-- (4.9999933463709825,-2.006318364490517);
		\draw [dash pattern=on 4pt off 4pt] (4.999994677096786,-1.0050546915924148)-- (6.83435906567376,-0.986561543570986);
		\draw [line width=1.2pt,color=ffqqqq] (12.152727272727269,-0.24909090909090986)-- (13.021985819207918,2.947454759156864);
		\draw [dash pattern=on 2pt off 2pt] (10.44,-0.26)-- (12.152727272727269,-0.24909090909090986);
		\draw [dash pattern=on 4pt off 4pt] (10.44,-0.26)-- (10.458276294868353,-2.0120662743914233);
		\draw [dash pattern=on 4pt off 4pt] (12.152727272727269,-0.24909090909090986)-- (12.159538178682746,-2.013857808490751);
		\draw [dash pattern=on 4pt off 4pt] (10.447903891472164,-1.0177105635786279)-- (12.155606396919241,-0.9950978436262381);
		\draw (5.028941833666851,-0.5412834010110604) node[anchor=north west] {\normalsize \textbf{Stretch by}};
		\draw (10.423449964212742,-0.5412834010110604) node[anchor=north west] {\textbf{Shrink by}};
		\begin{scriptsize}
		\draw [fill=qqqqff] (0.,1.) circle (1.0pt);
		\draw [fill=qqqqff] (0.3559781572104634,0.9845757866336147) circle (1.0pt);
		\draw [fill=qqqqff] (1.,-0.9485267187798286) circle (1.5pt);
		\draw[color=qqqqff] (0.7365959388462649,-1.1392902175378734) node {\large $\mathbf{1}$};
		\draw [fill=qqqqff] (1.86,1.6114732812201713) circle (1.5pt);
		\draw[color=qqqqff] (1.6362967859238076,1.7124281577514768) node {\large $\mathbf{2}$};
		\draw[color=black] (0.8733671921014724,-0.08527950585288169) node {\large $\mathbf{m_1}$};
		\draw [fill=qqqqff] (5.,3.) circle (1.5pt);
		\draw[color=qqqqff] (4.933555720144476,3.233055677282726) node {\large $\mathbf{3}$};
		\draw[color=black] (2.9955391657222212,2.4402805288693394) node {\large $\mathbf{m_2}$};
		\draw [fill=qqqqff] (7.827426593263678,0.9845757866336147) circle (1.0pt);
		\draw [fill=qqqqff] (8.34,0.98) circle (1.0pt);
		\draw [fill=qqqqff] (8.82712744034122,0.9845757866336147) circle (1.0pt);
		\draw [fill=qqqqff] (10.44,-0.26) circle (1.5pt);
		\draw[color=qqqqff] (10.177600514463515,-0.4097438158692781) node {\large $\mathbf{1'}$};
		\draw [fill=qqqqff] (11.14,2.16) circle (1.5pt);
		\draw[color=qqqqff] (10.756374689087355,1.9404301053305661) node {\large $\mathbf{2'}$};
		\draw[color=black] (10.396833156366485,1.0021143979858522) node {\large $\mathbf{m_3}$};
		\draw [fill=qqqqff] (14.927468626016232,3.744741056109039) circle (1.5pt);
		\draw[color=qqqqff] (15.18656808607808,3.6065981837931425) node {\large $\mathbf{3'}$};
		\draw[color=black] (11.940230955363392,2.8085913672663296) node {\large $\mathbf{m_4}$};
		\draw [fill=qqqqff] (15.544415588248219,0.9845757866336147) circle (1.0pt);
		\draw [fill=qqqqff] (16.438884767212336,0.9845757866336147) circle (1.0pt);
		\draw [fill=qqqqff] (16.,1.) circle (1.0pt);
		\draw [fill=qqqqff] (6.86,3.) circle (1.5pt);
		\draw[color=qqqqff] (7.0908049164697,3.233055677282726) node {\large $\mathbf{5}$};
		\draw [fill=qqqqff] (1.5176602862940822,0.5924155287932532) circle (1.5pt);
		\draw[color=qqqqff] (1.1627542794133925,0.3794936949814159) node {\large $\mathbf{4}$};
		\draw[color=ffqqqq] (4.1530875149699025,1.4405796817917933) node {\large $\mathbf{m_2}$};
		\draw [fill=qqqqff] (-0.39818213093575305,0.9845757866336147) circle (1.0pt);
		\draw [fill=xdxdff,shift={(4.999994677096786,-1.0050546915924148)},rotate=90] (0,0) ++(0 pt,3.75pt) -- ++(3.2475952641916446pt,-5.625pt)--++(-6.495190528383289pt,0 pt) -- ++(3.2475952641916446pt,5.625pt);
		\draw [fill=xdxdff,shift={(6.83435906567376,-0.986561543570986)},rotate=270] (0,0) ++(0 pt,3.75pt) -- ++(3.2475952641916446pt,-5.625pt)--++(-6.495190528383289pt,0 pt) -- ++(3.2475952641916446pt,5.625pt);
		\draw[color=black] (5.77540906505188,-1.4094446629468238) node {\large $t$};
		\draw [fill=qqqqff] (12.152727272727269,-0.24909090909090986) circle (1.5pt);
		\draw[color=qqqqff] (12.415467492424541,-0.36220520451704044) node {\large $\mathbf{4'}$};
		\draw [fill=qqqqff] (13.021985819207918,2.947454759156864) circle (1.5pt);
		\draw[color=qqqqff] (12.878546662708105,3.180439843226013) node {\large $ \mathbf{5'}$};
		\draw[color=ffqqqq] (12.887315968384222,1.0722688433948029) node {\large $\mathbf{m_3}$};
		\draw [fill=xdxdff,shift={(10.447903891472164,-1.0177105635786279)},rotate=90] (0,0) ++(0 pt,3.75pt) -- ++(3.2475952641916446pt,-5.625pt)--++(-6.495190528383289pt,0 pt) -- ++(3.2475952641916446pt,5.625pt);
		\draw [fill=xdxdff,shift={(12.155606396919241,-0.9950978436262381)},rotate=270] (0,0) ++(0 pt,3.75pt) -- ++(3.2475952641916446pt,-5.625pt)--++(-6.495190528383289pt,0 pt) -- ++(3.2475952641916446pt,5.625pt);
		\draw[color=black] (11.30007164100672,-1.3743674402423485) node {\large $t$};
		\end{scriptsize}
		\end{tikzpicture}
	}
	\caption{\label{fig:shrink-stretch-up-up} Shrink and stretch operations being applied to two up-up \flexis. The 1-2-3 one is stretched by $t$, which results in 1-4-5, and 1'-2'-3' is shrunk by $t$, which results in 4'-5'-3'. Note that 3 and 5 (also, 1' and 4') are the same states but shifted in time. In fact, all states along the run of the schedule stay the same apart from 2 and 2', and as a result the schedule stays safe.
	}
\end{figure}

\begin{figure}
	\centering
	\resizebox{\columnwidth}{0.2\textheight}{
		\begin{tikzpicture}[line cap=round,line join=round,>=triangle 45,x=1.0cm,y=1.0cm]
		\draw (-0.44,4.98)-- (20.,5.);
		\draw (-0.44,-2.)-- (20.,-2.);
		\draw [line width=1.2pt] (1.02,1.38)-- (3.46,3.3);
		\draw [line width=1.2pt] (3.46,3.3)-- (5.32,0.);
		\draw [line width=1.2pt] (11.32,0.62)-- (14.98,4.54);
		\draw [line width=1.2pt] (14.98,4.54)-- (17.88,2.72);
		\draw [line width=1.2pt,color=ffqqqq] (3.46,3.3)-- (4.78,4.26);
		\draw [line width=1.2pt,color=ffqqqq] (4.78,4.26)-- (7.28,0.);
		\draw [line width=1.2pt,color=ffqqqq] (13.32,0.62)-- (16.154407971063947,3.8029577560909025);
		\draw [dash pattern=on 5pt off 5pt] (5.32,0.)-- (5.32,-2.);
		\draw [dash pattern=on 5pt off 5pt] (7.28,0.)-- (7.28,-2.);
		\draw [dash pattern=on 5pt off 5pt] (11.32,0.62)-- (11.34,-2.);
		\draw [dash pattern=on 5pt off 5pt] (13.32,0.62)-- (13.32,-2.);
		\draw [dash pattern=on 5pt off 5pt] (5.32,-1.34)-- (7.28,-1.3417923868518333);
		\draw [dash pattern=on 5pt off 5pt] (11.334656488549617,-1.3)-- (13.32,-1.3105263157894735);
		\draw [dash pattern=on 5pt off 5pt] (5.32,0.)-- (7.28,0.);
		\draw [dash pattern=on 5pt off 5pt] (11.32,0.62)-- (13.32,0.62);
		\draw (5.28,-0.8) node[anchor=north west] {\large \textbf{Stretch by}};
		\draw (11.33,-0.74) node[anchor=north west] {\large \textbf{Shrink by}};
		\draw (-1.76,-1.65) node[anchor=north west] { \Large $\mathbf{V_{min}}$};
		\draw (-1.80,5.35) node[anchor=north west] {\Large $\mathbf{V_{max}}$};
		\begin{scriptsize}
		\draw [fill=qqqqff] (1.02,1.38) circle (1.5pt);
		\draw[color=qqqqff] (0.72,1.45) node {\Large $\mathbf{1}$};
		\draw [fill=qqqqff] (3.46,3.3) circle (1.5pt);
		\draw[color=qqqqff] (3.10,3.41) node {\Large $\mathbf{2}$};
		\draw[color=black] (1.9,2.56) node {\Large $\mathbf{m_i}$};
		\draw [fill=qqqqff] (5.32,0.) circle (1.5pt);
		\draw[color=qqqqff] (5.,-0.33) node {\Large $\mathbf{3}$};
		\draw[color=black] (3.75,1.46) node {\Large $\mathbf{m_{i+1}}$};
		\draw [fill=qqqqff] (11.32,0.62) circle (1.5pt);
		\draw[color=qqqqff] (10.98,0.25) node {\Large $\mathbf{4}$};
		\draw [fill=qqqqff] (14.98,4.54) circle (1.5pt);
		\draw[color=qqqqff] (14.6,4.6) node {\Large $\mathbf{5}$};
		\draw[color=black] (12.71,2.74) node {\Large $\mathbf{m_j}$};
		\draw [fill=qqqqff] (17.88,2.72) circle (1.5pt);
		\draw[color=qqqqff] (18.2,2.59) node {\Large $\mathbf{6}$};
		\draw[color=black] (17.25,3.56) node {\Large $\mathbf{m_{j+1}}$};
		\draw [fill=qqqqff] (4.78,4.26) circle (1.5pt);
		\draw[color=qqqqff] (4.56,4.47) node {\Large $\mathbf{2'}$};
		\draw [fill=qqqqff] (7.28,0.) circle (1.5pt);
		\draw[color=qqqqff] (7.72,-0.27) node {\Large $\mathbf{3'}$};
		\draw[color=ffqqqq] (6.65,2.24) node {\Large $\mathbf{m_{i+1}}$};
		\draw [fill=qqqqff] (13.32,0.62) circle (1.5pt);
		\draw[color=qqqqff] (13.68,0.25) node {\Large $\mathbf{4'}$};
		\draw [fill=qqqqff] (16.154407971063947,3.8029577560909025) circle (1.5pt);
		\draw[color=qqqqff] (16.4,4.07) node {\Large $\mathbf{5'}$};
		\draw[color=ffqqqq] (15.1,1.86) node {\Large $\mathbf{m_j}$};
		\draw [fill=xdxdff,shift={(5.32,-1.34)},rotate=90] (0,0) ++(0 pt,3.75pt) -- ++(3.2475952641916446pt,-5.625pt)--++(-6.495190528383289pt,0 pt) -- ++(3.2475952641916446pt,5.625pt);
		\draw [fill=xdxdff,shift={(7.28,-1.3417923868518333)},rotate=270] (0,0) ++(0 pt,3.75pt) -- ++(3.2475952641916446pt,-5.625pt)--++(-6.495190528383289pt,0 pt) -- ++(3.2475952641916446pt,5.625pt);
		\draw[color=black] (6.24,-1.75) node {\Large $t$};
		\draw [fill=xdxdff,shift={(11.334656488549617,-1.3)},rotate=90] (0,0) ++(0 pt,3.75pt) -- ++(3.2475952641916446pt,-5.625pt)--++(-6.495190528383289pt,0 pt) -- ++(3.2475952641916446pt,5.625pt);
		\draw [fill=xdxdff,shift={(13.32,-1.3105263157894735)},rotate=270] (0,0) ++(0 pt,3.75pt) -- ++(3.2475952641916446pt,-5.625pt)--++(-6.495190528383289pt,0 pt) -- ++(3.2475952641916446pt,5.625pt);
		\draw[color=black] (12.24,-1.69) node {\Large $t$};
		\draw [fill=black] (-0.46,1.38) circle (1.0pt);
		\draw [fill=black] (-0.08,1.36) circle (1.0pt);
		\draw [fill=black] (0.3,1.36) circle (1.0pt);
		\draw [fill=black] (8.86,1.42) circle (1.0pt);
		\draw [fill=black] (9.32,1.42) circle (1.0pt);
		\draw [fill=black] (9.76,1.42) circle (1.0pt);
		\draw [fill=black] (18.62,1.4) circle (1.0pt);
		\draw [fill=black] (19.06,1.4) circle (1.0pt);
		\draw [fill=black] (19.48,1.4) circle (1.0pt);
		\end{scriptsize}
		\end{tikzpicture}
	}
	\caption{\label{fig:shrink-stretch-up-down}Shrink and stretch operations being applied to two up-down \flexis.}
\end{figure}
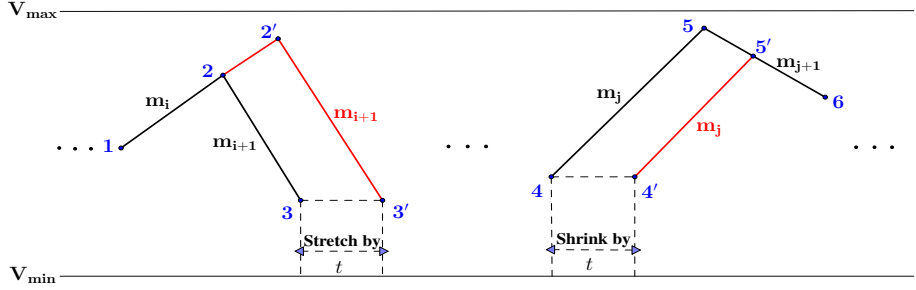 

\section{Complexity of Optimal Control in One-dimension}
\label{sec:optimal-one-dimension}

We start with considering the easy case of infinite time horizons, before turning to the interesting case of finite time horizons.

\subsection{Infinite time horizon}
First let us consider the case $\modeszero = \emptyset$.
If also $\modesup \times \modesdown = \emptyset$ then there are no safe schedules with infinite horizon at all.
Otherwise, let $(i',j') = \argmin_{(i,j) \in \modesup \times \modesdown} \cscost{i,j} /\cstime{i,j}$.
Let us pick any mode $m^{-} \in \modesdown$ and denote $t^{-} := (\vmin-\vinit)/A(m^{-})$. 
Consider the infinite schedule $\sigma$,
which starts with the timed action $(m^{-},t^{-})$ followed by infinitely many complete \shots of type $(i',j')$. 
Obviously, at all times $t = t^{-} + k \cdot \cstime{i',j'}$ where $k \in \mathbb{N}$,
$\sigma$ is more expensive by at most $\pi_d(m^{-}) + \pi_c(m^{-}) t^{-}$ from the cheapest
schedule with time horizon $t$.
Consequently, as $k \to \infty$, this shows that the limit superior of the average cost cannot be smaller than $\cscost{i',j'} / \cstime{i',j'}$. At the same time, $\sigma$ realises this long-time average.

If $\modeszero \neq \emptyset$, then let $m' = \min_{m\in \modeszero} \pi_c(m)$ be the \flatm with the lowest continuous cost to run.
We claim that if $\pi_c(m') < \cscost{i',j'} / \cstime{i',j'}$ or $\modesup \times \modesdown = \emptyset$
then an optimal safe schedule is simply $(m',\infty)$, whose limit-average cost is $\pi_c(m')$,
and otherwise $\sigma$ defined above is an optimal safe schedule.
This is because, if $\pi_c(m') < \cscost{i',j'} / \cstime{i',j'}$, then, 
at any time point of $\sigma$ where a \shot of some type $(i,j)$ is used, 
removing this \shot and increasing the time $m'$ is used for by $\cstime{i,j}$ reduces the total cost up to this time point. 

Taking into account that $\argmin_{(i,j) \in \modesup \times \modesdown} \cscost{i,j} /\cstime{i,j}$ can be computed using logarithmic space (because multiplication, division and comparison can be \cite{DBLP:journals/ita/ChiuDL01}) we get the following theorem.

\begin{theorem}
An optimal safe infinite schedule for one-dimensional \names can be computed in deterministic {\sc LogSpace}.
\end{theorem}

\subsection{Finite Time Horizon}

Due to Theorem \ref{thm:np-hard}, we already know that the decision problem for optimal schedules in one-dimensional \names is at least NP-hard.
Here, we show that the problem is NP-complete by showing that an optimal schedule exists and that each section of an optimal schedule can be guessed.

Note that the existence of an optimal schedule for the one-dimensional case sets it apart from the general case.
In Example \ref{ex:no-optimal}, we have shown that optimal schedules are not even guaranteed to exist for two-dimensional \names.

\begin{theorem}
For any one-dimensional \names \automata and $\tmax \geq 0$, there exists an optimal schedule with time horizon $\tmax$, and 
checking for the existence of an optimal schedule with cost $\leq C$ is {\sc NP}-complete. (When $\tmax$ and $C$ are given in binary.)
\end{theorem}
\begin{proof}
First, we can simply iterate over all schedules of length one and directly calculate their costs.
Next, we can iterate over pairs of modes, $m_1$ and $m_2$, and for each of them solve a linear program (LP) which 
will give us the cheapest schedule of length two using these two modes. 
This LP finds the cheapest partition of $\tmax$ between the two modes and has the following form:
Minimise
$\pi_c(m_1) t_1 + \pi_c(m_2) (\tmax-t_1) + \pi_d(m_1) + \pi_d(m_2)$ 
{\setlength{\abovedisplayskip}{0pt}
	\setlength{\belowdisplayskip}{0pt}
\begin{flalign*}
&& \text{Subject to: } & 0 \leq t_1 \leq \tmax, \ \ \ \vmin \leq \vinit {+} A(m_1) t_1 \leq \vmax \ \text{ and }\\
&& & \vmin \leq \vinit {+} A(m_1) t_1 {+} A(m_2) (\tmax-t_1)\leq \vmax. 
\end{flalign*}
}
This can be done in $\calO(\size^2)$ time.

Now, for schedules of length at least three, we showed in Section \ref{sec:schedule-structure} that any such a schedule can be transformed without increasing its cost
into one that can be split into three sections: the head section, the \shots section, and the tail section 
(some of which may be empty).
Due to Theorem \ref{thm:standard-form}, there are \patnum combined patterns for the tail and head sections.
Note that, when considering only the cost of the whole schedule, 
it suffices for us to know the number of \shots of each type
in the \shots section and not their precise order.
Notice that a schedule with time horizon $\tmax$ can contain at most $\lfloor \tmax / \cscost{i,j} \rfloor $ \shots of type $(i,j)$.
The size of this number is polynomial in the size of the input \automata.
There are $\calO(|\allmodes|^2)$ types of \shots so the number of \shots of each type and the combined pattern of the schedule can be guessed non-deterministically with polynomially many bits.
This guess uniquely determines the cost of the schedule.
This is because, after the total time of the \shots section is deducted from $\tmax$,
we get the exact time the head and tail section have to last for.
Each combined pattern has at most one of the following: a \flexi, a \flatm, or the last state not at the border. 
The time remaining will determinate exactly (if at all possible) the value of this single flexible point along this schedule.
Now, computing the cost of the resulting schedule and checking whether it is lower than $C$ can be done in polynomial time.
This shows that the problem is in NP.
It also shows that optimal schedules exist, because there are only finitely many options to choose from.
\qed\end{proof}

\section{Approximate Optimal Control in One-Dimension}
\label{sec:approx}

\subsection{Constant Factor Approximation}

We first show an approximation algorithm with a $3$-relative performance for the cost minimisation problem in one-dimensional \names, 
which runs in $\mathcal{O}(\size^7)$ time.
Our algorithm tries all possible patterns for an optimal schedule and 
for the leaps section always picks leaps of the same type.
It then adds, if necessary or for cost efficiency, a partial leap to the leaps section and
minimises the total cost of the just constructed schedule by optimising 
the time duration of this partial leap.
This constant approximation algorithm is crucial for showing the existence of an FPTAS for the same problem in the next subsection.

\begin{theorem}
\label{thm:constant-factor}
Computing a safe schedule with total cost at most three times larger than the optimal one for one-dimensional \name \automata can be done in
$\mathcal{O}(\size^7)$ time.
\end{theorem}

\subsection{FPTAS algorithm}
\label{sec:fptas-reduction}
We now show that the cost minimisation problem for one dimensional \names is in FPTAS by a polynomial time reduction to the 0-1 Knapsack problem, for which many FPTAS algorithms are available (see e.g.\ \cite{kellerer_knapsack_2004}).
This is similar to the FPTAS construction in \cite{MSW16}, but differs in how the modes with fractional duration are handled.
First we iterate over all possible schedules of length at most two
and find the cheapest one in polynomial time.
Next, thanks to Theorem \ref{thm:standard-form}, 
all optimal schedules longer than two can be transformed
into one of \patnum different patterns.
Each of these patterns results in a slightly different FPTAS formulation.
An FPTAS for the general model consists of all of these individual FPTASes 
executed one after another. The details of the proof can be found in Appendix \ref{sec:app-fptas}.

\begin{theorem}
\label{thm:fptas}
Solving the optimal control problem for \names with relative performance $\rho$ 
takes $\mathcal{O}(\mathsf{poly}(1/\rho) \mathsf{poly}\text{(size of the instance)})$ time
and is therefore in FPTAS.
\end{theorem}

\section*{Acknowledgement}
This work was supported by EPSRC EP/M027287/1 grant ``Energy Efficient Control''.

\bibliographystyle{plain}
\bibliography{HybridSystems,papers}

\newpage
\appendix

\section{Proof of Proposition \ref{prop:angular}}
\begin{proof}
Let $\sigma$ be a finite safe schedule with
two timed actions $(m_i,t_i), (m_{i+1},t_{i+1})$ in $\sigma$
such that $A(m_i) = A(m_{i+1})$. (If no such timed actions exist then $\sigma$ is angular and we are done.)
We can now replace these timed actions by a single timed action $(m,t_i + t_{i+1})$ 
such that $m$ is the mode from $m_i$ or $m_{i+1}$ with the lower continuous cost, and $m'$ the other mode. (I.e.\ $\{m,m'\} =\{m_i,m_{i+1}\}$ and $\pi_c(m) \leq \pi_c(m')$)
For the resulting safe schedule $\sigma'$, it now holds that $\pi(\sigma') \leq \pi(\sigma) - \pi_d(m')$.
\qed\end{proof}

\section{Proof of Proposition \ref{prop:zero-mode}}
\begin{proof}
Let $\sigma$ be a finite safe schedule with
timed actions $(m^1_0,t^1_0), (m^2_0,t^2_0), \ldots, (m^l_0,t^l_0)$ 
that use \flatms (i.e.~$m^i_0 \in \modeszero$ for all $i \leq l$).
(If no such timed actions exist then $\sigma$ is already in the form requested and we are done.)
Let $m_0 = \argmin_{i\leq l} \pi_c(m^i_0)$ be the \flatm among the ones used by $\sigma$ with the lowest continuous cost.
We construct a new safe schedule $\sigma'$ by first removing from $\sigma$ all timed actions that use a \flatm.
We then add at the very beginning a single 
timed action $(m_0, \sum_{i\leq l} t^i_0)$.
It is easy to see that such defined $\sigma'$ is safe and its total cost is equal or lower than that of $\sigma$.
\qed\end{proof}

\newcommand{\px}{x}

\section{Algorithm for Optimal Reachability Problem for Multi-Mode Systems with no discrete costs 
(adopted \cite[Algorithm 2]{ATW12b})}
\label{sec:optreach}
\begin{algorithm}
\caption{\label{alg:safe-interior-reach2}
  An algorithm checking whether any safe schedule exists 
  and if so finding one with the minimal total continuos-cost.
}
\KwIn{MMS $\automata = (M=\{m_1,\ldots,m_k\},N,A,\pi_{c},\pi_{d} \equiv 0,\vmin,\vmax,\vinit)$, target point 
$\vend$ and $t > 0$ such that all modes of $\automata$ are safe at $\vinit$ and $\vend$ for time $t$.}
\KwOut{NO, if no safe schedule from $\vinit$ to $\vend$ exists, and a continuos-cost-optimal 
  schedule (of at most exponential length), otherwise.
} 

\label{alg-line:safe-LP}
Check whether the following linear programming problem with variables
$\{t^{(m)}\}_{m\in M}$ has a solution.
\begin{align*}
   \text{Minimise} \sum_{m\in M} \pi_c(m) t^{(m)} & \text{ subject to:}\\
   \vinit + \sum_{m\in M} A(m) t^{(m)} =\ & \vend \text{ and }\\
   t^{(m)} \geq\ & 0 \text{ for all $m \in M$}. 
\end{align*}

\If{no satisfying assignment exists}{\Return NO}
\Else{
  Find a polynomial sized assignment $\{t^{(m)}\}_{m\in M}$.

  Let $l$ be the smallest natural number greater or equal to $\sum_{m \in M} t^{(m)}/t$. (Note that this number is at most exponential in the size of the input and can be written down using polynomially many bits.)

  {\bf return} the schedule 
  $\big((m_1,t^{(m_1)}/l),(m_2,t^{(m_2)}/l),\ldots,(m_k,t^{(m_k)}/l)\big)^l$.
}

\end{algorithm}

\section{Algorithm for finding a limit-safe schedule}
\label{app:limit-safe-schedule}
\begin{algorithm}[H]
\small
\label{alg:find-limit-safe}
\caption{Finding a limit-safe schedule to target state $\vend$ with time horizon $\tmax$. 
} 
\KwIn{Multi-mode system $\automata = (M,N,A,\pi_{c},\pi_{d},\vmin,\vmax,\vinit)$, set of modes $\zerocostmodes$ with zero discrete costs, time horizon $\tmax$, and target state $\vend$ such that any mode safe at $\vinit$ is safe as $\vend$.} 
\KwOut{NO if no safe schedule with time horizon $\tmax$ exists from $\vinit$ to $\vend$,
       and such a schedule, otherwise. 
} 
\setcounter{AlgoLine}{0}

$k:= 0; M_0 := \zerocostmodes;$

\Repeat{$M_{k} = M_{k-1}$}{
  $k := k + 1$; $M_{k} := M_{k-1}$;
  
  \ForEach{mode $q \in M \setminus M_{k-1}$}{
    \If{the following set of linear constraints is satisfiable for some
      assignment to the variables $t, \{ t^{(m)}_0 \}_{m \in M_0}$, $\{ t^{(m)}_1
      \}_{m \in M_1}$, \ldots, $\{ t^{(m)}_{k-1} \}_{m \in M_{k-1}}$: 
      \begin{align}
        \boldsymbol{\cdot}\ \nonumber  & t > 0&\\
        \nonumber &\emph{For all } i=0,\ldots,k-1:\\
        \nonumber &\qquad\boldsymbol{\cdot}\ t^{(m)}_i \geq 0\,\emph{ for all } m \in M_i\\
        \nonumber &\qquad\boldsymbol{\cdot}\ V_{i+1} = V_{i} + \sum_{m \in M_{i}} A(m) t^{(m)}_i \\
        \nonumber &\qquad\boldsymbol{\cdot}\ \vmin \leq V_{i+1} \leq \vmax \\
        \label{con:new-safe-mode}
        \boldsymbol{\cdot}\ & \vmin \leq V_{k} + A(q) t \leq \vmax %
      \end{align}
    }
    { $M_{k} := M_{k-1} \cup \{q\}$;}
    
  }
}
  $k := k - 1$;

\ForEach{$j=0,\ldots,k$ {\bf and} $q \in M_j$}{
    \If{the following set of linear constraints is {\bf not} satisfiable for any
      assignment to the variables $t, \{ t^{(m)}_0 \}_{m \in M_0}$, $\{ t^{(m)}_1
      \}_{m \in M_1}$, \ldots, $\{ t^{(m)}_{k} \}_{m \in M_{k}}$: 
      \begin{align}
        \boldsymbol{\cdot}\ \nonumber  & t^{(q)}_{j} > 0&\\
        \nonumber &\emph{For all } i=0,\ldots,k-1:\\
        \nonumber &\qquad\boldsymbol{\cdot}\ t^{(m)}_i \geq 0\,\emph{ for all } m \in M_i\\
        \nonumber &\qquad\boldsymbol{\cdot}\ V_{i+1} = V_{i} + \sum_{m \in M_{i}} A(m) t^{(m)}_i \\
        \nonumber &\qquad\boldsymbol{\cdot}\ \vmin \leq V_{i+1} \leq \vmax \\
        \nonumber \boldsymbol{\cdot}\ & \sum_{i=0}^{k} \sum_{m \in M_i} t^{(m)}_i = \tmax 
      \end{align}
    }
    { $M_{j} := M_{j} \setminus \{q\}$;}
    \tcc{the algorithm continues below...}
}

\end{algorithm}

\begin{algorithm}
    \If{the following set of linear constraints is {\bf not} satisfiable for any
      assignment to the variables $\{ t^{(m)}_0 \}_{m \in M_0}$, $\{ t^{(m)}_1
      \}_{m \in M_1}$, \ldots, $\{ t^{(m)}_{k} \}_{m \in M_{k}}$:
      \nllabel{alg:LP-new-safe-mode}
      \begin{align}
        \nonumber &\emph{For all } i=0,\ldots,k:\\
        \nonumber &\qquad\boldsymbol{\cdot}\ t^{(m)}_i > 0\,\emph{ for all } m \in M_i\\
        \nonumber &\qquad\boldsymbol{\cdot}\ V_{i+1} = V_{i} + \sum_{m \in M_{i}} A(m) t^{(m)}_i \\
        \nonumber &\qquad\boldsymbol{\cdot}\ \vmin \leq V_{i+1} \leq \vmax \\
        \nonumber \boldsymbol{\cdot}\ & \sum_{i=0}^{k} \sum_{m \in M_i} t^{(m)}_i = \tmax 
      \end{align}
    }
     {{\bf return} NO} 

Compute a polynomial sized solution to the linear program in step
\ref{alg:LP-new-safe-mode} and use it in the next line.

{\bf return} the schedule created by composing the following schedules obtained by 
repeatedly calling \cite[Algorithm 2]{ATW12b} (see Appendix \ref{sec:optreach}) to find a safe schedule: 
\begin{itemize}
\item from $V_0$ to $V_1$ using only modes in $M_0$ with 
the safe time bound $t = \min_{m \in M_0}{t^{(m)}_0}$, 
\item from $V_1$ to $V_2$ using only modes in $M_1$ with
the safe time bound $t = \min_{m \in M_1}{t^{(m)}_1}$, 
\item \ldots, 
\item from $V_{k}$ to $V_{k+1}$ using only modes in $M_{k}$ with
the safe time bound $t = \min_{m \in M_k}{t^{(m)}_k}$.  
\end{itemize}

\end{algorithm}

\section{Formal Definition of Operations}
\label{app:formal-def}

\begin{definition}[Rearrange Operation]
Let $(m_{i},t_{i}),\ldots,(m_{j},t_{j})$ be any subsequence of $\sigma$ such that 
either $\forall_{\:i\leq\:l\:\leq\:j\:}m_l \in \modesdown$ or $\forall_{\:i\leq\:l\:\leq\:j\:}m_l \in \modesup$ hold.
Note that any permutation of the timed actions $(m_{i},t_{i}),\ldots,(m_{j},t_{j})$ 
will result in a new schedule $\sigma'$ which is safe and has the same total cost as $\sigma$. 
\end{definition}

\begin{definition}[Shift Operation]
Let the run of our finite schedule $\sigma = \seq{ \allowdisplaybreaks (m_1,t_1), (m_2,t_2),\ldots,(m_k,t_k)}$ be $\seq{V_0,V_1,...,V_k}$.
For any $i\leq l\leq j$ such that $V_i = V_l = V_j$ holds, 
we can move the whole subsequence of timed actions 
$(m_{i},t_{i}),\ldots,(m_{j-1},t_{j-1})$ just after $(m_{l-1},t_{l-1})$ in $\sigma$
to obtain a new safe schedule with the same cost.
Specifically, the new schedule will look as follows:  
$\seq{(m_{1},t_{1}),\ldots,(m_{i-1},t_{i-1}),(m_{j},t_{j}),\ldots,(m_{l-1},t_{l-1}),(m_{i},t_{i}),\ldots,(m_{j-1},t_{j-1}),(m_l,t_l),\ldots,(m_k,t_k)}$
Analogously, in the same situation, we can also move 
the whole subsequence of timed actions 
$(m_{j},t_{j}),\ldots,(m_{l-1},t_{l-1})$ just after $(m_{i-1},t_{i-1})$ in $\sigma$
to obtain a new safe schedule with the same cost.
\end{definition}

\begin{definition}[Shift-Down Operation]
Let the run of our finite schedule $\sigma = \seq{ \allowdisplaybreaks (m_1,t_1), (m_2,t_2),\ldots,(m_k,t_k)}$ be $\seq{V_0,V_1,...,V_k}$.
For any $i\leq j$ and $l$ such that $V_i = V_{j+1} = \vmax$ and $V_{l+1} = \vmin$, we can 
``rotate'' the whole subsequence of timed actions 
$(m_{i},t_{i}),\ldots,(m_{j},t_{j})$ and move it just after $(m_{l},t_{l})$ in $\sigma$
to obtain a new safe schedule $\sigma'$ with the same cost.
Specifically, let $d = \argmin_{\:i\leq\:b\:<\:j\:} V_{b+1}$.
Note that if we rotate the subsequence of actions in the way to start with timed action $(m_d,t_d)$ then we will never encounter a lower state the start state, because $d$ was the lowest point along this subsequence of timed actions.
Specifically, the new schedule $\sigma'$ will look as follows
$\seq{(m_{1},t_{1}),\ldots,(m_{i-1},t_{i-1}),(m_{j+1},t_{j+1}),\ldots,(m_{l},t_{l}),(m_{d},t_{d}),\ldots,(m_{j},t_{j}),(m_i,t_i),\ldots,(m_{d-1},t_{d-1}),(m_{l+1},t_{l+1}),\ldots,(m_k,t_k)}$.
\end{definition}

\begin{definition}[Resize Operation]
Let $\sigma = \seq{(m_1,t_1),\ldots,(m_k,t_k)}$ whose run is $\seq{V_0,V_1,\ldots,V_k}$.
For $i < k$ and $t \in \Real$, let $\Resize(\sigma,i,t)$ be a schedule $\sigma'$ identical to $\sigma$ apart from
timed actions $(m_i,t_i),(m_{i+1},t_{i+1})$ being replaced by $(m_i,t'_i),(m_{i+1},t'_{i+1})$ in the following way, where we distinguish among several cases. If $t>0$ then we will also refer to this operation as the {\em stretch} operation and if $t < 0$ as the {\em shrink} operation. 
\begin{itemize}

\item[\bf (up-up)] If $0 < A(m_i) < A(m_{i+1})$ then let
$t'_i = t_i + \beta t + t$ and $t'_{i+1} = t_{i+1} - \beta t$ where 
$$\beta = \frac{A(m_i)}{A(m_{i+1}) - A(m_i)} \geq 0$$
Let $\maxresize(\sigma,i) := [-t_i/(\beta+1), t_{i+1}/\beta]$.
Note that $\pi_c(\sigma') - \pi_c(\sigma) = ((\beta + 1)\pi_c(m_i) - \beta \pi_c(m_{i+1})) \cdot t$.

If $0 < A(m_{i+1}) < A(m_{i})$ then let
$t'_i = t_i - \beta\cdot t$ and $t'_{i+1} = t_{i+1} + \beta\cdot t + t$
where 
$$\beta = \frac{A(m_{i+1})}{A(m_{i+1}) - A(m_i)} \geq 0$$
Let $\maxresize(\sigma,i) := [-t_{i+1}/(\beta+1), t_{i}/\beta]$.
Note that $\pi_c(\sigma') - \pi_c(\sigma) = ((\beta + 1)\pi_c(m_{i+1}) - \beta \pi_c(m_{i})) \cdot t$.

\item[\bf (up-down)] Here $0 < A(m_i)$ and $A(m_{i+1}) < 0$ holds. Let
$t'_i = t_i + \beta t $ and $t'_{i+1} = t_{i+1} - \beta t + t$ where 
$$\beta = \frac{-A(m_{i+1})}{A(m_{i}) - A(m_{i+1})} \geq 0$$
Let $\maxresize(\sigma,i) := [-\min \{t_i/\beta, t_{i+1}/(1-\beta)\}, (\vmax-V_i)/(\beta A(m_{i}))]$. \\
Note that $\pi_c(\sigma') - \pi_c(\sigma) = (\beta \pi_c(m_i) + (1-\beta) \pi_c(m_{i+1})) \cdot t$.
\item[\bf (down-up)] Analogous to {\bf up-down} case.
\item[\bf (down-down)] Analogous to {\bf up-up} case.
\item[\bf (flat)] If $(m_1,t_1)$ is a \flatm action in $\sigma$, then let $\Resize(\sigma,0,t)$ be equal to $\sigma$ where the first action is replaced by $(m_1,t_1+t)$. Let $\maxresize(\sigma,0) := [-t_1, \tmax-t_1]$ and
notice that $\pi_c(\sigma') - \pi_c(\sigma) = \pi_c(m_{1}) \cdot t$
\item[\bf (last-action)] If $(m_k,t_k)$ is the last action in $\sigma$, then let $\Resize(\sigma,k,t)$ be equal to $\sigma$ where the last action is replaced by $(m_k,t_k+t)$.\\
Let $\maxresize(\sigma,k) := [-t_k, \max{\{(\vmax - V)/A(m_k),(V_k-\vmin)/A(m_k)\}}]$ and
notice that $\pi_c(\sigma') - \pi_c(\sigma) = - \pi_c(m_{k}) \cdot t$
\end{itemize}
\end{definition}

\begin{definition}[Wedge Operation]
Let $\sigma=\langle (m_1,t_1), (m_2,t_2),\ldots,(m_k,t_k)\rangle$ be a finite safe schedule
whose run is $\seq{V_0,V_1,...,V_k}$.
Let $\tau = \seq{(m_i,t_i),(m_{i+1},t_{i+1}),(m_{i+2},t_{i+2})}$ be any three consecutive timed actions
in which exactly two consecutive timed actions have the same trend.
It suffices to consider the case where $A(i) > A(i+1) > 0$ and $A(i+2) < 0$ as all other cases are very similar. Notice that if $A(i+1) > A(i)$ then we can simply change the order of $(m_{i},t_{i}),(m_{i+1},t_{i+1})$ using the rearrange operation defined earlier.
Furthermore, we only define this operation in the case where $V_{i-1} = V_{i+2}$.
This is the only situation we need this operation for and it is easy to generalise this further.
Let $\alpha = A(i+2) (t_i+t_{i+1}+t_{i+3}) / (A(i+2) - A(i+1))$.
For any $t \geq 0$, consider the sequence of timed actions
$\tau' = \seq{(m_i, (t + \alpha)t_i / (t_{i+1} - \alpha)), (m_{i+1},t), (m_{i+2}, t_i+t_{i+1}+t_{i+3} - t - (t + \alpha)t_i / (t_{i+1} - \alpha))}$.
Let us replace $\tau$ by $\tau'$ in $\sigma$ to get $\sigma'$ whose run is $\seq{V'_0,V'_1,...,V'_k}$.
We claim that $V_{i-1} = V'_{i-1} = V'_{i+2} = V'_{i+2}$, so the runs of $\sigma$ and $\sigma'$ can only differ at their $i$-th and $i+1$-th states.
At the same time notice that $\pi_c(\sigma') - \pi_c(\sigma)$ is a linear function of $t$ as a sum of linear functions.
As a result its minimum is attained at the smallest or largest permissible value of $t$. 
Moreover, the permissible value of $t$ is a closed interval $[\beta,\alpha]$ where $\beta$ can be calculated using the following
linear constraint $V'_{i+1} \leq \vmax$.
\end{definition}

\section{Proof of Theorem \ref{thm:fptas} from Section \ref{sec:approx}}
\label{sec:app-fptas}
\begin{proof}
We consider here only one of the \patnum possible pattern cases, because all these FPTAS algorithms will look essentially the same.
These can be later combined all these FPTASes into a single FPTAS for the general model by running them one by one.
The case we will look at is \UD pattern, with modes $m_1,m_2$, for the head section and \PUUD, with modes $m_3,m_4,m_5$, for the tail section.
We consider all combinations of these five modes $m_i$ individually, and therefore consider them given.
(Note that there are only quintically many such combinations.)
Wlog.\ we assume that $\cscost{m_3} - \pi_{d}(m_3) \geq \cscost{m_4} - \pi_{d}(m_4)$, because otherwise we could swap the role of $m_3$ with $m_4$ in our algorithm below.
Note that any schedule with this pattern which picks $m_3$ in the tail for $\alpha\cstime{m_3}$ amount of time, uses $m_4$ for $(1-\alpha)\cstime{m_3}$ amount of time in the tail section.

Let $c^*$ be the $3$-approximation, which can be computed
using the procedure from Theorem \ref{thm:constant-factor}, of the optimal cost $o^*$.
To get an approximation to our optimal cost problem with a relative performance $\rho$, 
it suffices to find a solution with $c^* \rho/3$ absolute performance.
We split this into two equal parts of 
$\epsilon = c^* \rho/6$.
An optimal solution to the knapsack instance that we produce
will provide us with a schedule with cost no greater than $\epsilon$ over the optimal one. 
Moreover, a solution to the knapsack instance with $\delta$ absolute error will 
provide a schedule with an $\epsilon + \delta$ absolute error.
Therefore, it suffices to set $\delta = \epsilon$ to find 
a schedule with $\rho$ relative performance.
In our reduction, the total value of all the items in the resulting knapsack problem is at most $4|M|^2$ times the optimal cost for safe schedules, 
so by using $\rho' = \rho/(12|M|^2)$, for the resulting knapsack problem
we will find a near optimal solution with a relative performance $\rho$ for \names.
The running time of this procedures is in
$\mathcal{O}\mathcal(\mathsf{poly}(1/\rho) \mathsf{poly}(|M|) \mathsf{poly}\text{(size of the knapsack instance)})$.
This suffices to establish the inclusion of the cost minimisation problem for \names in FPTAS.

For each type of leaps, $(m,m') \in \modesup \times \modesdown$, we build the following items for this knapsack problem instance:
$\{(2^i\cdot \cstime{m,m'},2^i\cdot \cscost{m,m'}) \mid i \in\mathbb{N} \wedge 2^i\cdot \cscost{m,m'} \leq c^* \wedge 2^i\cdot \cstime{m,m'} \leq \tmax\}$. 
Let $i^* \in \mathbb{N}$ be smallest such that $2^{-i^*}\cdot (\cscost{m_3} - \pi_{d}(m_3)) \leq \epsilon$.  
For both $m_3$ and $m_4$ we add the following extra multiset of items:
$\{(2^{-i}\cdot \cstime{m_3},2^{-i}\cdot (\cscost{m_3} - \pi_{d}(m_3) - \cscost{m_4} + \pi_{d}(m_4)) ) \mid i\in \mathbb{Z}_+ \wedge i \leq i^* \wedge 
2^{-i}\cdot (\cscost{m_3} - \pi_{d}(m_3)) \leq c^* \}$
and additionally $(2^{-i^*}\cdot \cstime{m_3},2^{-i^*}\cdot (\cscost{m_3} - \pi_{d}(m_3) - \cscost{m_4} + \pi_{d}(m_4))$, 
which is a copy of an element already in the multiset.
Note that this models the fact that the more $m_3$ is used in the tail section the less mode $m_4$ is used in tail section and with the same proportion.
Also, all costs are nonnegative because of the assumption that $\cscost{m_3} \geq \cscost{m_4}$.
Let $t_\Sigma$ be the time span of all items in this knapsack instance.
We set the volume of this 0-1 knapsack instance to be $t_\Sigma - \tmax + (\vmax-\vinit)/A(m_1) + \cstime{m_2} + \cstime{m_5}$.

The just produced knapsack problem has the following properties:
\begin{itemize}
\item
the size of its description is polynomial in the size of the original problem including the relative performance;
\item
fractional time duration of $m_3$ in the tail section can be overestimated by joining together the fractional items for both $m_3$ and $m_4$ (which do not include discrete costs), so that we do not exceed the volume by $2^{-i^*}\cdot \cstime{m_3}$ or more;
\item $n$ leaps of of type $(m,m')$ in $\sigma$ can be achieved by picking the items for this type and corresponding to the binary representation of $n$; and
\item
The volume of these items is $\geq \tmax - (\vmax-\vinit)/A(m_1) - \cstime{m_2} - \cstime{m_5}$,
which leaves enough space for modes $m_1$ and $m_2$ in the head section, and mode $m_5$ in the tail section. 
Let $v^*$ be the value of the items in this knapsack and $o^*$ denotes the optimal cost.
Then $$0 \leq v^* + \pi_d(m_1) + \pi_c(m_1) (\vmax-\vinit)/A(m_1) + \cscost{m_2} + \pi_d(m_3) + \cscost{m_4} + \cscost{m5} - o^* \leq \epsilon$$
\item 
Let $V_\Sigma$ be the value of all items in the multiset. 
For any solution to the knapsack problem with value $V$ 
we get a schedule $\sigma'$ with cost $\leq V_\Sigma - V + \pi_d(m_1) + \pi_c(m_1) (\vmax-\vinit)/A(m_1) + \cscost{m_2} + \pi_d(m_3) + \cscost{m_4} + \cscost{m5}$.
\end{itemize}
All of this shows that solving this knapsack instance with a relative performance of $\rho/(12|M|^2)$ gives
us a safe schedule with relative performance of $\rho$. 
\qed\end{proof}

\section{Transformation of an Example Schedule}
\label{sec:example-transformation}
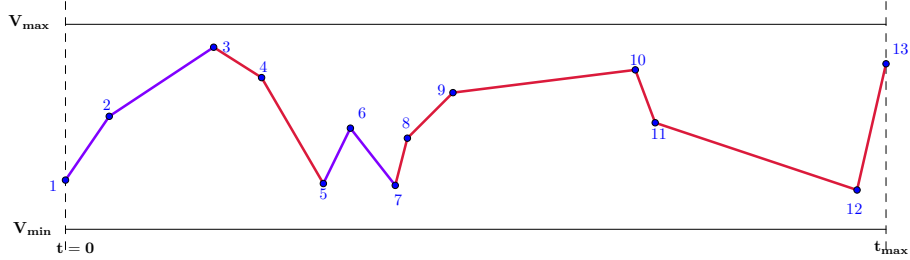
\begin{figure}[htp]
\centering
\resizebox{\textwidth}{!}{
\begin{tikzpicture}[line cap=round,line join=round,>=triangle 45,x=1.0cm,y=1.0cm]
\draw (-2.5,-16.)-- (15.5,-16.);
\draw (-2.5,-20.5)-- (15.5,-20.5);
\draw (-3.859731197569333,-15.690812833218225) node[anchor=north west] {\normalsize $\mathbf{V_{max}}$};
\draw (-3.7586278255167347,-20.22024390117464) node[anchor=north west] {\normalsize $\mathbf{V_{min}}$};
\draw [dash pattern=on 5pt off 5pt] (-2.5,-15.5)-- (-2.5,-21.);
\draw [dash pattern=on 5pt off 5pt] (15.5,-15.5)-- (15.5,-21.);
\draw (15.12748207390858,-20.648188179953348) node[anchor=north west] {\normalsize $\mathbf{t_{max}}$};
\draw (-2.8082561282223133,-20.68862952877439) node[anchor=north west] {\normalsize $\mathbf{t=0}$};
\draw [line width=1.6pt,color=dcrutc] (10.,-17.)-- (10.436285610668046,-18.161045227459507);
\draw [line width=1.6pt,color=dcrutc] (6.,-17.5)-- (10.,-17.);
\draw [line width=1.6pt,color=dcrutc] (10.436285610668046,-18.161045227459507)-- (14.864613306571846,-19.637154459427453);
\draw [line width=1.6pt,color=dcrutc] (14.864613306571846,-19.637154459427453)-- (15.5,-16.86692206518624);
\draw [line width=1.6pt,color=dcrutc] (0.7505825680291393,-16.50294992579685)-- (1.8020576373761599,-17.170232181344005);
\draw [line width=1.6pt,color=dcrutc] (1.8020576373761599,-17.170232181344005)-- (3.1568428228809746,-19.495609738553778);
\draw [line width=1.6pt,color=xfqqff] (3.1568428228809746,-19.495609738553778)-- (3.75,-18.28);
\draw [line width=1.6pt,color=xfqqff] (3.75,-18.28)-- (4.734055426901509,-19.536051087374855);
\draw [line width=1.6pt,color=xfqqff] (0.7505825680291393,-16.50294992579685)-- (-1.54,-18.02);
\draw [line width=1.6pt,color=xfqqff] (-1.54,-18.02)-- (-2.5,-19.42);
\draw [line width=1.6pt,color=dcrutc] (6.,-17.5)-- (5.,-18.5);
\draw [line width=1.6pt,color=dcrutc] (5.,-18.5)-- (4.734055426901509,-19.536051087374855);
\begin{scriptsize}
\draw [fill=qqqqff] (-2.5,-19.42) circle (2.0pt);
\draw[color=qqqqff] (-2.7678147794012737,-19.546161424580024) node {\normalsize $1$};
\draw [fill=qqqqff] (4.734055426901509,-19.536051087374855) circle (2.0pt);
\draw[color=qqqqff] (4.79471745013306,-19.852781656895616) node {\normalsize $7$};
\draw [fill=qqqqff] (10.,-17.) circle (2.0pt);
\draw[color=qqqqff] (10.052092796868157,-16.775929030338823) node {\normalsize $10$};
\draw [fill=qqqqff] (10.436285610668046,-18.161045227459507) circle (2.0pt);
\draw[color=qqqqff] (10.517168308310108,-18.396893099338197) node {\normalsize $11$};
\draw [fill=qqqqff] (6.,-17.5) circle (2.0pt);
\draw[color=qqqqff] (5.745089147427482,-17.443211285885976) node {\normalsize $9$};
\draw [fill=qqqqff] (14.864613306571846,-19.637154459427453) circle (2.0pt);
\draw[color=qqqqff] (14.803951283340266,-20.053885028948215) node {\normalsize $12$};
\draw [fill=qqqqff] (15.5,-16.86692206518624) circle (2.0pt);
\draw[color=qqqqff] (15.814985003866246,-16.533280937412588) node {\normalsize $13$};
\draw [fill=qqqqff] (0.7505825680291393,-16.50294992579685) circle (2.0pt);
\draw[color=qqqqff] (1.0336720097764127,-16.492839588591547) node {\normalsize $3$};
\draw [fill=qqqqff] (1.8020576373761599,-17.170232181344005) circle (2.0pt);
\draw[color=qqqqff] (1.842498986197197,-16.937694425622983) node {\normalsize $4$};
\draw [fill=qqqqff] (3.1568428228809746,-19.495609738553778) circle (2.0pt);
\draw[color=qqqqff] (3.1568428228809715,-19.731457610432497) node {\normalsize $5$};
\draw [fill=qqqqff] (3.75,-18.28) circle (2.0pt);
\draw[color=qqqqff] (4.006111148122796,-17.948728146148966) node {\normalsize $6$};
\draw [fill=qqqqff] (-1.54,-18.02) circle (2.0pt);
\draw[color=qqqqff] (-1.5950156635911366,-17.76674207645429) node {\normalsize $2$};
\draw [fill=qqqqff] (5.,-18.5) circle (2.0pt);
\draw[color=qqqqff] (4.976703519827737,-18.171155564664684) node {\normalsize $8$};
\end{scriptsize}
\end{tikzpicture}
}
\caption{\label{fig:EX_Original} The original schedule. For any two non-overlapping \flexis, 
we try to shrink one by $t$ and stretch the other by $t$ for the maximum possible time $t>0$.
We repeat this until there is at most one \flexi left. 
Here, we start off by shrinking
\flexi 1-2-3 (of type up-up) and stretching \flexi 5-6-7 (of type up-down).
This will result in straightening 
the 1-2-3 \flexi and removal of its midpoint 2 
(we can see the end result in the next figure).
}
\end{figure}

\begin{figure}[htp]
\centering
\resizebox{\textwidth}{!}{

\begin{tikzpicture}[line cap=round,line join=round,>=triangle 45,x=1.0cm,y=1.0cm]
\draw (-2.5,-16.)-- (15.5,-16.);
\draw (-2.5,-20.5)-- (15.5,-20.5);
\draw (-3.8597311975693342,-15.690812833218215) node[anchor=north west] {\normalsize $\mathbf{V_{max}}$};
\draw (-3.758627825516736,-20.220243901174626) node[anchor=north west] {\normalsize $\mathbf{V_{min}}$};
\draw [dash pattern=on 5pt off 5pt] (-2.5,-15.5)-- (-2.5,-21.);
\draw [dash pattern=on 5pt off 5pt] (15.5,-15.5)-- (15.5,-21.);
\draw (15.127482073908585,-20.648188179953337) node[anchor=north west] {\normalsize $\mathbf{t_{max}}$};
\draw (-2.808256128222314,-20.688629528774374) node[anchor=north west] {\normalsize $\mathbf{t=0}$};
\draw [line width=1.6pt,color=dcrutc] (10.,-17.)-- (10.436285610668046,-18.161045227459507);
\draw [line width=1.6pt,color=dcrutc] (6.,-17.5)-- (10.,-17.);
\draw [line width=1.6pt,color=dcrutc] (10.436285610668046,-18.161045227459507)-- (14.864613306571846,-19.637154459427453);
\draw [line width=1.6pt,color=dcrutc] (14.864613306571846,-19.637154459427453)-- (15.5,-16.86692206518624);
\draw [line width=1.6pt,color=dcrutc] (6.,-17.5)-- (5.,-18.5);
\draw [line width=1.6pt,color=dcrutc] (5.,-18.5)-- (4.734055426901509,-19.536051087374855);
\draw [line width=1.6pt,color=dcrutc] (-2.5,-19.42)-- (-0.5,-16.5);
\draw [line width=1.6pt,color=xfqqff] (-0.5,-16.5)-- (0.5281551495134233,-17.230894204575563);
\draw [line width=1.6pt,color=xfqqff] (0.5281551495134233,-17.230894204575563)-- (1.9031610094287579,-19.556271761785336);
\draw [line width=1.6pt,color=xfqqff] (1.9031610094287579,-19.556271761785336)-- (3.,-17.5);
\draw [line width=1.6pt,color=xfqqff] (3.,-17.5)-- (4.734055426901509,-19.536051087374855);
\begin{scriptsize}
\draw [fill=qqqqff] (-2.5,-19.42) circle (2.0pt);
\draw[color=qqqqff] (-2.767814779401275,-19.54616142458001) node {\normalsize $1$};
\draw [fill=qqqqff] (4.734055426901509,-19.536051087374855) circle (2.0pt);
\draw[color=qqqqff] (4.794717450133062,-19.852781656895602) node {\normalsize $6$};
\draw [fill=qqqqff] (10.,-17.) circle (2.0pt);
\draw[color=qqqqff] (10.294740889794397,-16.99835644885453) node {\normalsize $9$};
\draw [fill=qqqqff] (10.436285610668046,-18.161045227459507) circle (2.0pt);
\draw[color=qqqqff] (10.456506285078554,-18.437334448159224) node {\normalsize $10$};
\draw [fill=qqqqff] (6.,-17.5) circle (2.0pt);
\draw[color=qqqqff] (5.745089147427484,-17.443211285885962) node {\normalsize $8$};
\draw [fill=qqqqff] (14.864613306571846,-19.637154459427453) circle (2.0pt);
\draw[color=qqqqff] (14.80395128334027,-20.0538850289482) node {\normalsize $11$};
\draw [fill=qqqqff] (15.5,-16.86692206518624) circle (2.0pt);
\draw[color=qqqqff] (15.814985003866251,-16.836591053570373) node {\normalsize $12$};
\draw [fill=qqqqff] (5.,-18.5) circle (2.5pt);
\draw[color=qqqqff] (4.976703519827739,-18.171155564664673) node {\normalsize $7$};
\draw [fill=qqqqff] (-0.5,-16.5) circle (2.0pt);
\draw[color=qqqqff] (-0.179568454854764,-16.472618914181016) node {\normalsize $2$};
\draw [fill=qqqqff] (0.5281551495134233,-17.230894204575563) circle (2.0pt);
\draw[color=qqqqff] (0.7505825680291384,-17.09945982090713) node {\normalsize $3$};
\draw [fill=qqqqff] (1.9031610094287579,-19.556271761785336) circle (2.0pt);
\draw[color=qqqqff] (1.8222783117866779,-19.83256098248508) node {\normalsize $4$};
\draw [fill=qqqqff] (3.,-17.5) circle (2.0pt);
\draw[color=qqqqff] (3.197284171702012,-17.321887239422846) node {\normalsize $5$};
\end{scriptsize}
\end{tikzpicture}
}
\caption{\label{fig:EX_Step1}
Next, we will apply the procedure to 
\flexis 2-3-4 (of type down-down) and 4-5-6 (of type up-down).
This will result in straightening 
the 2-3-4 \flexi and removal of its midpoint 3 
(we can see the end result in the next figure).
}
\end{figure}
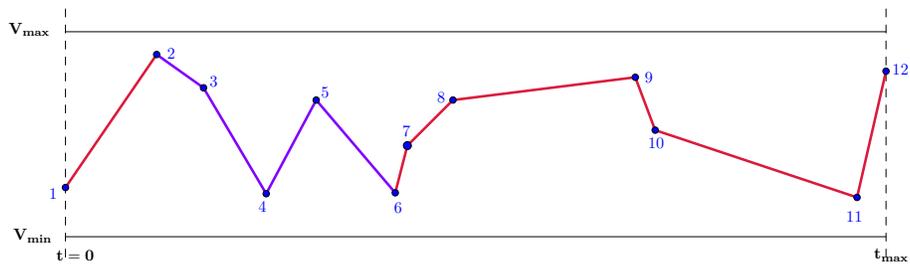

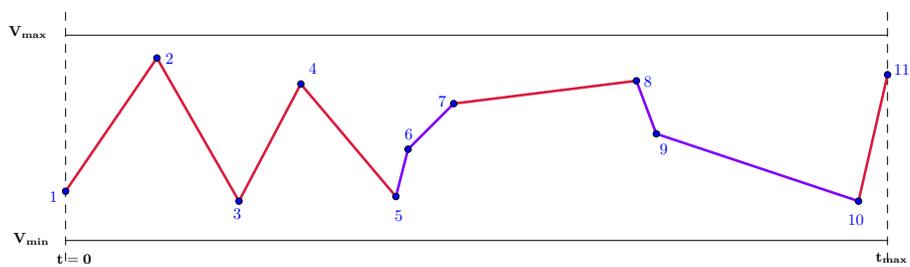
\begin{figure}[htp]
\centering
\resizebox{\textwidth}{!}{
\begin{tikzpicture}[line cap=round,line join=round,>=triangle 45,x=1.0cm,y=1.0cm]
\draw (-2.5,-16.)-- (15.5,-16.);
\draw (-2.5,-20.5)-- (15.5,-20.5);
\draw (-3.8597311975693342,-15.690812833218138) node[anchor=north west] {$\mathbf{V_{max}}$};
\draw (-3.758627825516736,-20.22024390117453) node[anchor=north west] {$\mathbf{V_{min}}$};
\draw [dash pattern=on 5pt off 5pt] (-2.5,-15.5)-- (-2.5,-21.);
\draw [dash pattern=on 5pt off 5pt] (15.5,-15.5)-- (15.5,-21.);
\draw (15.127482073908585,-20.648188179953234) node[anchor=north west] {$\mathbf{t_{max}}$};
\draw (-2.808256128222314,-20.688629528774275) node[anchor=north west] {$\mathbf{t=0}$};
\draw [line width=1.6pt,color=xfqqff] (10.,-17.)-- (10.436285610668046,-18.161045227459507);
\draw [line width=1.6pt,color=dcrutc] (6.,-17.5)-- (10.,-17.);
\draw [line width=1.6pt,color=xfqqff] (10.436285610668046,-18.161045227459507)-- (14.864613306571846,-19.637154459427453);
\draw [line width=1.6pt,color=dcrutc] (14.864613306571846,-19.637154459427453)-- (15.5,-16.86692206518624);
\draw [line width=1.6pt,color=xfqqff] (6.,-17.5)-- (5.,-18.5);
\draw [line width=1.6pt,color=xfqqff] (5.,-18.5)-- (4.734055426901509,-19.536051087374855);
\draw [line width=1.6pt,color=dcrutc] (-2.5,-19.42)-- (-0.5,-16.5);
\draw [line width=1.6pt,color=dcrutc] (-0.5,-16.5)-- (1.2965407771131692,-19.637154459427414);
\draw [line width=1.6pt,color=dcrutc] (1.2965407771131692,-19.637154459427414)-- (2.651325962617984,-17.069128809291403);
\draw [line width=1.6pt,color=dcrutc] (2.651325962617984,-17.069128809291403)-- (4.734055426901509,-19.536051087374855);
\begin{scriptsize}
\draw [fill=qqqqff] (-2.5,-19.42) circle (2.0pt);
\draw[color=qqqqff] (-2.767814779401275,-19.546161424579918) node {\normalsize $1$};
\draw [fill=qqqqff] (4.734055426901509,-19.536051087374855) circle (2.0pt);
\draw[color=qqqqff] (4.794717450133062,-19.952781656895506) node {\normalsize $5$};
\draw [fill=qqqqff] (10.,-17.) circle (2.0pt);
\draw[color=qqqqff] (10.254299540973358,-17.01857712326497) node {\normalsize $8$};
\draw [fill=qqqqff] (10.436285610668046,-18.161045227459507) circle (2.0pt);
\draw[color=qqqqff] (10.598051005952192,-18.496893099338094) node {\normalsize $9$};
\draw [fill=qqqqff] (6.,-17.5) circle (2.0pt);
\draw[color=qqqqff] (5.745089147427484,-17.44321128588588) node {\normalsize $7$};
\draw [fill=qqqqff] (14.864613306571846,-19.637154459427453) circle (2.0pt);
\draw[color=qqqqff] (14.80395128334027,-20.053885028948105) node {\normalsize $10$};
\draw [fill=qqqqff] (15.5,-16.86692206518624) circle (2.0pt);
\draw[color=qqqqff] (15.814985003866251,-16.755708355928213) node {\normalsize $11$};
\draw [fill=qqqqff] (5.,-18.5) circle (2.0pt);
\draw[color=qqqqff] (5.017144868648778,-18.171155564664584) node {\normalsize $6$};
\draw [fill=qqqqff] (-0.5,-16.5) circle (2.0pt);
\draw[color=qqqqff] (-0.22000980367580328,-16.51306026300198) node {\normalsize $2$};
\draw [fill=qqqqff] (1.2965407771131692,-19.637154459427414) circle (2.0pt);
\draw[color=qqqqff] (1.2560994282921287,-19.912340308074465) node {\normalsize $3$};
\draw [fill=qqqqff] (2.651325962617984,-17.069128809291403) circle (2.0pt);
\draw[color=qqqqff] (2.9141947299547373,-16.735487681517693) node {\normalsize $4$};
\end{scriptsize}
\end{tikzpicture}
}
\caption{\label{fig:EX_Step2}
Next, we will apply the procedure to 
\flexis 5-6-7 (of type up-up) and 8-9-10 (of type down-down).
This will result in straightening of 
the 5-6-7 \flexi and removal of its midpoint 6 
(we can see the end result in the next figure).
  }
\end{figure}

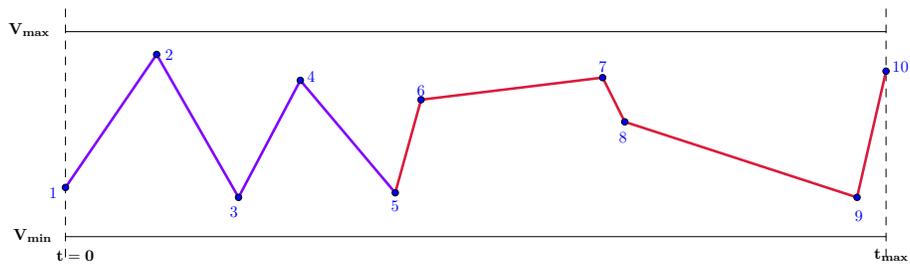
\begin{figure}[htp]
\centering
\resizebox{\textwidth}{!}{
\begin{tikzpicture}[line cap=round,line join=round,>=triangle 45,x=1.0cm,y=1.0cm]
\draw (-2.5,-16.)-- (15.5,-16.);
\draw (-2.5,-20.5)-- (15.5,-20.5);
\draw (-3.8597311975693347,-15.690812833218265) node[anchor=north west] {\normalsize $\mathbf{V_{max}}$};
\draw (-3.758627825516737,-20.220243901174693) node[anchor=north west] {\normalsize $\mathbf{V_{min}}$};
\draw [dash pattern=on 5pt off 5pt] (-2.5,-15.5)-- (-2.5,-21.);
\draw [dash pattern=on 5pt off 5pt] (15.5,-15.5)-- (15.5,-21.);
\draw (15.12748207390859,-20.648188179953404) node[anchor=north west] {\normalsize $\mathbf{t_{max}}$};
\draw (-2.8082561282223146,-20.68862952877444) node[anchor=north west] {\normalsize $\mathbf{t=0}$};
\draw [line width=1.6pt,color=dcrutc] (9.28370716926842,-17.008466786059845)-- (9.76900335512089,-17.979059157764794);
\draw [line width=1.6pt,color=dcrutc] (9.76900335512089,-17.979059157764794)-- (14.864613306571846,-19.637154459427453);
\draw [line width=1.6pt,color=dcrutc] (14.864613306571846,-19.637154459427453)-- (15.5,-16.86692206518624);
\draw [line width=1.6pt,color=xfqqff] (-2.5,-19.42)-- (-0.5,-16.5);
\draw [line width=1.6pt,color=xfqqff] (-0.5,-16.5)-- (1.2965407771131692,-19.637154459427414);
\draw [line width=1.6pt,color=xfqqff] (1.2965407771131692,-19.637154459427414)-- (2.651325962617984,-17.069128809291403);
\draw [line width=1.6pt,color=xfqqff] (2.651325962617984,-17.069128809291403)-- (4.734055426901509,-19.536051087374855);
\draw [line width=1.6pt,color=dcrutc] (4.734055426901509,-19.536051087374855)-- (5.300234310396053,-17.49376297191232);
\draw [line width=1.6pt,color=dcrutc] (9.28370716926842,-17.008466786059845)-- (5.300234310396053,-17.49376297191232);
\begin{scriptsize}
\draw [fill=qqqqff] (-2.5,-19.42) circle (2.0pt);
\draw[color=qqqqff] (-2.7678147794012755,-19.546161424580074) node {\normalsize $1$};
\draw [fill=qqqqff] (4.734055426901509,-19.536051087374855) circle (2.0pt);
\draw[color=qqqqff] (4.734055426901505,-19.83145761043255) node {\normalsize $5$};
\draw [fill=qqqqff] (9.28370716926842,-17.008466786059845) circle (2.0pt);
\draw[color=qqqqff] (9.28370716926842,-16.77592903033887) node {\normalsize $7$};
\draw [fill=qqqqff] (9.76900335512089,-17.979059157764794) circle (2.0pt);
\draw[color=qqqqff] (9.728562006299851,-18.314907029643567) node {\normalsize $8$};
\draw [fill=qqqqff] (14.864613306571846,-19.637154459427453) circle (2.0pt);
\draw[color=qqqqff] (14.905054655392874,-20.03256098248515) node {\normalsize $9$};
\draw [fill=qqqqff] (15.5,-16.86692206518624) circle (2.0pt);
\draw[color=qqqqff] (15.814985003866257,-16.77592903033887) node {\normalsize $10$};
\draw [fill=qqqqff] (-0.5,-16.5) circle (2.0pt);
\draw[color=qqqqff] (-0.22000980367580278,-16.51306026300211) node {\normalsize $2$};
\draw [fill=qqqqff] (1.2965407771131692,-19.637154459427414) circle (2.0pt);
\draw[color=qqqqff] (1.195437405060571,-19.95278165689567) node {\normalsize $3$};
\draw [fill=qqqqff] (2.651325962617984,-17.069128809291403) circle (2.0pt);
\draw[color=qqqqff] (2.893974055544219,-16.978135774444066) node {\normalsize $4$};
\draw [fill=qqqqff] (5.300234310396053,-17.49376297191232) circle (2.0pt);
\draw[color=qqqqff] (5.300234310396054,-17.281445890601862) node {\normalsize $6$};
\end{scriptsize}
\end{tikzpicture}
}
\caption{\label{fig:EX_Step3}
Next, we will apply the procedure to 
\flexis 1-2-3 (of type up-down) and 3-4-5 (of type up-down).
This will result in moving the midpoint 2 up until it reaches
$\vmax$.
}
\end{figure}

\begin{figure}[htp]
\centering
\resizebox{\textwidth}{!}{

\begin{tikzpicture}[line cap=round,line join=round,>=triangle 45,x=1.0cm,y=1.0cm]
\draw (-2.5,-16.)-- (15.5,-16.);
\draw (-2.5,-20.5)-- (15.5,-20.5);
\draw (-3.859731197569335,-15.690812833218275) node[anchor=north west] {\normalsize $\mathbf{V_{max}}$};
\draw (-3.758627825516737,-20.220243901174704) node[anchor=north west] {\normalsize $\mathbf{V_{min}}$};
\draw [dash pattern=on 5pt off 5pt] (-2.5,-15.5)-- (-2.5,-21.);
\draw [dash pattern=on 5pt off 5pt] (15.5,-15.5)-- (15.5,-21.);
\draw (15.1274820739086,-20.64818817995342) node[anchor=north west] {\normalsize $\mathbf{t_{max}}$};
\draw (-2.808256128222314,-20.688629528774456) node[anchor=north west] {\normalsize $\mathbf{t=0}$};
\draw [line width=1.6pt,color=dcrutc] (-2.5,-19.42)-- (-0.15934778044424203,-16.);
\draw [line width=1.6pt,color=xfqqff] (-0.15934778044424203,-16.)-- (1.8424989861972014,-19.394506366501215);
\draw [line width=1.6pt,color=xfqqff] (1.8424989861972014,-19.394506366501215)-- (2.9546360787757813,-17.251114878986115);
\draw [line width=1.6pt,color=dcrutc] (2.9546360787757813,-17.251114878986115)-- (4.734055426901509,-19.536051087374855);
\draw [line width=1.6pt,color=xfqqff] (9.223045146036867,-16.96802543723884)-- (9.728562006299859,-17.99927983217535);
\draw [line width=1.6pt,color=dcrutc] (4.734055426901509,-19.536051087374855)-- (5.2395722871645,-17.513983646322874);
\draw [line width=1.6pt,color=dcrutc] (5.2395722871645,-17.513983646322874)-- (9.223045146036867,-16.96802543723884);
\draw [line width=1.6pt,color=xfqqff] (9.728562006299859,-17.99927983217535)-- (14.864613306571846,-19.637154459427453);
\draw [line width=1.6pt,color=dcrutc] (14.864613306571846,-19.637154459427453)-- (15.5,-16.86692206518624);
\begin{scriptsize}
\draw [fill=qqqqff] (-2.5,-19.42) circle (2.0pt);
\draw[color=qqqqff] (-2.7678147794012746,-19.546161424580088) node {\normalsize $1$};
\draw [fill=qqqqff] (4.734055426901509,-19.536051087374855) circle (2.0pt);
\draw[color=qqqqff] (4.794717450133068,-19.852781656895684) node {\normalsize $5$};
\draw [fill=qqqqff] (9.223045146036867,-16.96802543723884) circle (2.0pt);
\draw[color=qqqqff] (9.223045146036867,-16.73548768151784) node {\normalsize $7$};
\draw [fill=qqqqff] (-0.15934778044424203,-16.) circle (2.0pt);
\draw[color=qqqqff] (-0.11890643162320294,-15.663791937760292) node {\normalsize $2$};
\draw [fill=qqqqff] (1.8424989861972014,-19.394506366501215) circle (2.0pt);
\draw[color=qqqqff] (1.8424989861972014,-19.689912889558925) node {\normalsize $3$};
\draw [fill=qqqqff] (2.9546360787757813,-17.251114878986115) circle (2.0pt);
\draw[color=qqqqff] (2.995077427596821,-16.93769442562304) node {\normalsize $4$};
\draw [fill=qqqqff] (9.728562006299859,-17.99927983217535) circle (2.0pt);
\draw[color=qqqqff] (9.690327401584015,-18.2351277040541) node {\normalsize $8$};
\draw [fill=qqqqff] (5.2395722871645,-17.513983646322874) circle (2.0pt);
\draw[color=qqqqff] (4.976703519827745,-17.463431960296553) node {\normalsize $6$};
\draw [fill=qqqqff] (14.864613306571846,-19.637154459427453) circle (2.0pt);
\draw[color=qqqqff] (14.743289260108728,-19.953885028948282) node {\normalsize $9$};
\draw [fill=qqqqff] (15.5,-16.86692206518624) circle (2.0pt);
\draw[color=qqqqff] (15.81498500386627,-16.53328093741264) node {\normalsize $10$};
\end{scriptsize}
\end{tikzpicture}
}
\caption{\label{fig:EX_Step4} 
Next, we will apply the procedure to 
\flexis 2-3-4 (of type down-up) and 7-8-9 (of type down-down).
This will result in moving the midpoint 3 down until it reaches
$\vmin$.
 }
\end{figure}

\begin{figure}[htp]
\centering
\resizebox{\textwidth}{!}{
\begin{tikzpicture}[line cap=round,line join=round,>=triangle 45,x=1.0cm,y=1.0cm]
\draw (-2.5,-16.)-- (15.5,-16.);
\draw (-2.5,-20.5)-- (15.5,-20.5);
\draw (-3.859731197569338,-15.690812833218241) node[anchor=north west] {\normalsize $\mathbf{V_{max}}$};
\draw (-3.7586278255167396,-20.22024390117466) node[anchor=north west] {\normalsize $\mathbf{V_{min}}$};
\draw [dash pattern=on 5pt off 5pt] (-2.5,-15.5)-- (-2.5,-21.);
\draw [dash pattern=on 5pt off 5pt] (15.5,-15.5)-- (15.5,-21.);
\draw (15.127482073908608,-20.64818817995337) node[anchor=north west] {\normalsize $\mathbf{t_{max}}$};
\draw (-2.8082561282223164,-20.68862952877441) node[anchor=north west] {\normalsize $\mathbf{t=0}$};
\draw [line width=1.6pt,color=dcrutc] (-2.5,-19.42)-- (-0.15934778044424203,-16.);
\draw [line width=1.6pt,color=xfqqff] (10.5,-17.)-- (11.164229889446759,-18.302589948333083);
\draw [line width=1.6pt,color=xfqqff] (11.164229889446759,-18.302589948333083)-- (14.864613306571846,-19.637154459427453);
\draw [line width=1.6pt,color=dcrutc] (14.864613306571846,-19.637154459427453)-- (15.5,-16.86692206518624);
\draw [line width=1.6pt,color=dcrutc] (-0.15934778044424203,-16.)-- (2.530001916154869,-20.506643459079804);
\draw [line width=1.6pt,color=xfqqff] (2.530001916154869,-20.506643459079804)-- (4.208317892227998,-17.372438925449234);
\draw [line width=1.6pt,color=xfqqff] (4.208317892227998,-17.372438925449234)-- (5.906854542711651,-19.515830412964267);
\draw [line width=1.6pt,color=dcrutc] (5.906854542711651,-19.515830412964267)-- (6.5,-17.5);
\draw [line width=1.6pt,color=dcrutc] (6.5,-17.5)-- (10.5,-17.);
\begin{scriptsize}
\draw [fill=qqqqff] (-2.5,-19.42) circle (2.0pt);
\draw[color=qqqqff] (-2.767814779401277,-19.546161424580045) node {\normalsize $1$};
\draw [fill=qqqqff] (10.5,-17.) circle (2.0pt);
\draw[color=qqqqff] (10.739595726825845,-16.978135774444038) node {\normalsize $7$};
\draw [fill=qqqqff] (-0.15934778044424203,-16.) circle (2.0pt);
\draw[color=qqqqff] (0.0226382892504339,-15.845778007454935) node {\normalsize $2$};
\draw [fill=qqqqff] (11.164229889446759,-18.302589948333083) circle (2.0pt);
\draw[color=qqqqff] (11.125995284730915,-18.638437820211855) node {\normalsize $8$};
\draw [fill=qqqqff] (14.864613306571846,-19.637154459427453) circle (2.0pt);
\draw[color=qqqqff] (14.743289260108735,-20.053885028948236) node {\normalsize $9$};
\draw [fill=qqqqff] (15.5,-16.86692206518624) circle (2.0pt);
\draw[color=qqqqff] (15.45101286447692,-16.79614970474936) node {\normalsize $10$};
\draw [fill=qqqqff] (2.530001916154869,-20.506643459079804) circle (2.0pt);
\draw[color=qqqqff] (2.5300019161548697,-20.981829307726947) node {\normalsize $3$};
\draw [fill=qqqqff] (4.208317892227998,-17.372438925449234) circle (2.0pt);
\draw[color=qqqqff] (4.329641938691118,-17.220783867370276) node {\normalsize $4$};
\draw [fill=qqqqff] (5.906854542711651,-19.515830412964267) circle (2.0pt);
\draw[color=qqqqff] (5.92707521712217,-19.811236936021998) node {\normalsize $5$};
\draw [fill=qqqqff] (6.5,-17.5) circle (2.0pt);
\draw[color=qqqqff] (6.473033426206201,-17.281445890601834) node {\normalsize $6$};
\end{scriptsize}
\end{tikzpicture}

}
\caption{\label{fig:EX_Step5} 
Next, we will apply the procedure to 
\flexis 3-4-5 (of type up-down) and 7-8-9 (of type down-down).
This will result in moving the midpoint 4 up until it reaches
$\vmax$.
}
\end{figure}

\begin{figure}[htp]
\centering
\resizebox{\textwidth}{!}{

\begin{tikzpicture}[line cap=round,line join=round,>=triangle 45,x=1.0cm,y=1.0cm]
\draw (-2.5,-16.)-- (15.5,-16.);
\draw (-2.5,-20.5)-- (15.5,-20.5);
\draw (-3.859731197569333,-15.69081283321823) node[anchor=north west] {\normalsize $\mathbf{V_{max}}$};
\draw (-3.7586278255167347,-20.220243901174647) node[anchor=north west] {\normalsize $\mathbf{V_{min}}$};
\draw [dash pattern=on 5pt off 5pt] (-2.5,-15.5)-- (-2.5,-21.);
\draw [dash pattern=on 5pt off 5pt] (15.5,-15.5)-- (15.5,-21.);
\draw (15.127482073908602,-20.648188179953358) node[anchor=north west] {\normalsize $\mathbf{t_{max}}$};
\draw (-2.808256128222312,-20.688629528774396) node[anchor=north west] {\normalsize $\mathbf{t=0}$};
\draw [line width=1.6pt,color=dcrutc] (-2.5,-19.42)-- (-0.15934778044424203,-16.);
\draw [line width=1.6pt,color=dcrutc] (14.864613306571846,-19.637154459427453)-- (15.5,-16.86692206518624);
\draw [line width=1.6pt,color=dcrutc] (-0.15934778044424203,-16.)-- (2.5097812417443492,-20.506643459079804);
\draw [line width=1.6pt,color=dcrutc] (2.5097812417443492,-20.506643459079804)-- (4.895820822185666,-16.);
\draw [line width=1.6pt,color=xfqqff] (4.895820822185666,-16.)-- (7.888480634942572,-19.515830412964334);
\draw [line width=1.6pt,color=xfqqff] (7.888480634942572,-19.515830412964334)-- (8.373776820795046,-17.5139836463228);
\draw [line width=1.6pt,color=dcrutc] (8.373776820795046,-17.5139836463228)-- (12.316808330846374,-16.98824611164929);
\draw [line width=1.6pt,color=xfqqff] (12.316808330846374,-16.98824611164929)-- (13.327842051372356,-18.88898950623814);
\draw [line width=1.6pt,color=xfqqff] (13.327842051372356,-18.88898950623814)-- (14.864613306571846,-19.637154459427453);
\begin{scriptsize}
\draw [fill=qqqqff] (-2.5,-19.42) circle (2.0pt);
\draw[color=qqqqff] (-2.7678147794012724,-19.54616142458003) node {\normalsize $1$};
\draw [fill=qqqqff] (-0.15934778044424203,-16.) circle (2.0pt);
\draw[color=qqqqff] (-0.11890643162320066,-15.663791937760246) node {\normalsize $2$};
\draw [fill=qqqqff] (14.864613306571846,-19.637154459427453) circle (2.0pt);
\draw[color=qqqqff] (14.74328926010873,-19.953885028948222) node {\normalsize $9$};
\draw [fill=qqqqff] (15.5,-16.86692206518624) circle (2.0pt);
\draw[color=qqqqff] (15.794764329455752,-16.81637037915987) node {\normalsize $10$};
\draw [fill=qqqqff] (2.5097812417443492,-20.506643459079804) circle (2.0pt);
\draw[color=qqqqff] (2.509781241744352,-20.961608633316413) node {\normalsize $3$};
\draw [fill=qqqqff] (4.895820822185666,-16.) circle (2.0pt);
\draw[color=qqqqff] (4.996924194238266,-15.845778007454925) node {\normalsize $4$};
\draw [fill=qqqqff] (7.888480634942572,-19.515830412964334) circle (2.0pt);
\draw[color=qqqqff] (7.908701309353094,-19.911236936021984) node {\normalsize $5$};
\draw [fill=qqqqff] (8.373776820795046,-17.5139836463228) circle (2.0pt);
\draw[color=qqqqff] (8.414218169616085,-17.34210791383338) node {\normalsize $6$};
\draw [fill=qqqqff] (12.316808330846374,-16.98824611164929) circle (2.0pt);
\draw[color=qqqqff] (12.57967709818313,-17.059018472086105) node {\normalsize $7$};
\draw [fill=qqqqff] (13.327842051372356,-18.88898950623814) circle (2.0pt);
\draw[color=qqqqff] (13.1660766560882,-19.283292657243276) node {\normalsize $8$};
\end{scriptsize}
\end{tikzpicture}
}
\caption{\label{fig:EX_Step6} 
Next, we will apply the procedure to 
\flexis 4-5-6 (of type down-up) and 7-8-9 (of type down-down).
This will result in moving the midpoint 5 down until it reaches
$\vmin$.
}
\end{figure}

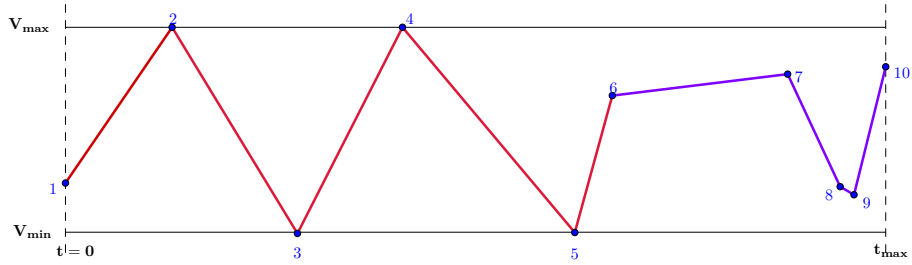
\begin{figure}[htp]
\centering
\resizebox{\textwidth}{!}{
\begin{tikzpicture}[line cap=round,line join=round,>=triangle 45,x=1.0cm,y=1.0cm]
\draw (-2.5,-16.)-- (15.5,-16.);
\draw (-2.5,-20.5)-- (15.5,-20.5);
\draw (-3.8597311975693356,-15.690812833218265) node[anchor=north west] {\normalsize $\mathbf{V_{max}}$};
\draw (-3.7586278255167374,-20.22024390117469) node[anchor=north west] {\normalsize $\mathbf{V_{min}}$};
\draw [dash pattern=on 5pt off 5pt] (-2.5,-15.5)-- (-2.5,-21.);
\draw [dash pattern=on 5pt off 5pt] (15.5,-15.5)-- (15.5,-21.);
\draw (15.1274820739086,-20.6481881799534) node[anchor=north west] {\normalsize $\mathbf{t_{max}}$};
\draw (-2.8082561282223146,-20.68862952877444) node[anchor=north west] {\normalsize $\mathbf{t=0}$};
\draw [line width=1.6pt,color=ccqqqq] (-2.5,-19.42)-- (-0.15934778044424203,-16.);
\draw [line width=1.6pt,color=dcrutc] (-0.15934778044424203,-16.)-- (2.590663939386428,-20.526864133490324);
\draw [line width=1.6pt,color=dcrutc] (2.590663939386428,-20.526864133490324)-- (4.895820822185666,-16.);
\draw [line width=1.6pt,color=dcrutc] (4.895820822185666,-16.)-- (8.677086936952838,-20.506643459079765);
\draw [line width=1.6pt,color=dcrutc] (8.677086936952838,-20.506643459079765)-- (9.5,-17.5);
\draw [line width=1.6pt,color=xfqqff] (9.5,-17.5)-- (13.348062725782873,-17.028687460470366);
\draw [line width=1.6pt,color=xfqqff] (14.803951283340286,-19.677595808248455)-- (14.5,-19.5);
\draw [line width=1.6pt,color=xfqqff] (14.5,-19.5)-- (13.348062725782873,-17.028687460470366);
\draw [line width=1.6pt,color=xfqqff] (14.803951283340286,-19.677595808248455)-- (15.5,-16.866922065186206);
\begin{scriptsize}
\draw [fill=qqqqff] (-2.5,-19.42) circle (2.0pt);
\draw[color=qqqqff] (-2.767814779401275,-19.546161424580074) node {\normalsize $1$};
\draw [fill=qqqqff] (-0.15934778044424203,-16.) circle (2.0pt);
\draw[color=qqqqff] (-0.13912710603372316,-15.80533665863392) node {\normalsize $2$};
\draw [fill=qqqqff] (2.590663939386428,-20.526864133490324) circle (2.0pt);
\draw[color=qqqqff] (2.5906639393864275,-20.96160863331646) node {\normalsize $3$};
\draw [fill=qqqqff] (4.895820822185666,-16.) circle (2.0pt);
\draw[color=qqqqff] (5.057586217469823,-15.80533665863392) node {\normalsize $4$};
\draw [fill=qqqqff] (8.677086936952838,-20.506643459079765) circle (2.0pt);
\draw[color=qqqqff] (8.677086936952838,-20.98182930772698) node {\normalsize $5$};
\draw [fill=qqqqff] (9.5,-17.5) circle (2.0pt);
\draw[color=qqqqff] (9.526355262194663,-17.3218872394229) node {\normalsize $6$};
\draw [fill=qqqqff] (13.348062725782873,-17.028687460470366) circle (2.0pt);
\draw[color=qqqqff] (13.590710818709109,-17.079239146496665) node {\normalsize $7$};
\draw [fill=qqqqff] (14.803951283340286,-19.677595808248455) circle (2.0pt);
\draw[color=qqqqff] (15.087040725087562,-19.84947154073787) node {\normalsize $9$};
\draw [fill=qqqqff] (14.5,-19.5) circle (2.0pt);
\draw[color=qqqqff] (14.257993074256257,-19.687706145453713) node {\normalsize $8$};
\draw [fill=qqqqff] (15.5,-16.866922065186206) circle (2.0pt);
\draw[color=qqqqff] (15.855426352687308,-16.998356448854587) node {\normalsize $10$};
\end{scriptsize}
\end{tikzpicture}

}
\caption{\label{fig:EX_Step7} 
Next, we will apply the procedure to 
\flexis 6-7-8 (of type up-down) and 8-9-10 (of type down-up).
This will result in straightening of the 8-8-10 \flexi and removal of the midpoint 9.
}
\end{figure}

\begin{figure}[htp]
\centering
\resizebox{\textwidth}{!}{

\begin{tikzpicture}[line cap=round,line join=round,>=triangle 45,x=1.0cm,y=1.0cm]
\draw (-2.5,-16.)-- (15.5,-16.);
\draw (-2.5,-20.5)-- (15.5,-20.5);
\draw (-3.8597311975693356,-15.690812833218265) node[anchor=north west] {\normalsize $\mathbf{V_{max}}$};
\draw (-3.7586278255167374,-20.22024390117469) node[anchor=north west] {\normalsize $\mathbf{V_{min}}$};
\draw [dash pattern=on 5pt off 5pt] (-2.5,-15.5)-- (-2.5,-21.);
\draw [dash pattern=on 5pt off 5pt] (15.5,-15.5)-- (15.5,-21.);
\draw (15.1274820739086,-20.6481881799534) node[anchor=north west] {\normalsize $\mathbf{t_{max}}$};
\draw (-2.8082561282223146,-20.68862952877444) node[anchor=north west] {\normalsize $\mathbf{t=0}$};
\draw [line width=1.6pt,color=ccqqqq] (-2.5,-19.42)-- (-0.15934778044424203,-16.);
\draw [line width=1.6pt,color=dcrutc] (-0.15934778044424203,-16.)-- (2.590663939386428,-20.526864133490324);
\draw [line width=1.6pt,color=dcrutc] (2.590663939386428,-20.526864133490324)-- (4.895820822185666,-16.);
\draw [line width=1.6pt,color=dcrutc] (4.895820822185666,-16.)-- (8.677086936952838,-20.506643459079765);
\draw [line width=1.6pt,color=xfqqff] (8.677086936952838,-20.506643459079765)-- (9.5,-17.5);
\draw [line width=1.6pt,color=xfqqff] (9.5,-17.5)-- (13.348062725782873,-17.028687460470366);
\draw [line width=1.6pt,color=xfqqff] (14.702847911287687,-19.96068524999573)-- (14.5,-19.5);
\draw [line width=1.6pt,color=xfqqff] (14.5,-19.5)-- (13.348062725782873,-17.028687460470366);
\draw [line width=1.6pt,color=xfqqff] (14.702847911287687,-19.96068524999573)-- (15.5,-16.866922065186206);
\begin{scriptsize}
\draw [fill=qqqqff] (-2.5,-19.42) circle (2.0pt);
\draw[color=qqqqff] (-2.767814779401275,-19.546161424580074) node {\normalsize $1$};
\draw [fill=qqqqff] (-0.15934778044424203,-16.) circle (2.0pt);
\draw[color=qqqqff] (-0.13912710603372316,-15.80533665863392) node {\normalsize $2$};
\draw [fill=qqqqff] (2.590663939386428,-20.526864133490324) circle (2.0pt);
\draw[color=qqqqff] (2.5906639393864275,-20.96160863331646) node {\normalsize $3$};
\draw [fill=qqqqff] (4.895820822185666,-16.) circle (2.0pt);
\draw[color=qqqqff] (5.057586217469823,-15.80533665863392) node {\normalsize $4$};
\draw [fill=qqqqff] (8.677086936952838,-20.506643459079765) circle (2.0pt);
\draw[color=qqqqff] (8.677086936952838,-20.98182930772698) node {\normalsize $5$};
\draw [fill=qqqqff] (9.5,-17.5) circle (2.0pt);
\draw[color=qqqqff] (9.526355262194663,-17.3218872394229) node {\normalsize $6$};
\draw [fill=qqqqff] (13.348062725782873,-17.028687460470366) circle (2.0pt);
\draw[color=qqqqff] (13.590710818709109,-17.079239146496665) node {\normalsize $7$};
\draw [fill=qqqqff] (14.702847911287687,-19.96068524999573) circle (2.0pt);
\draw[color=qqqqff] (14.985937353034963,-20.13256098248515) node {\normalsize $8$};
\draw [fill=qqqqff] (15.5,-16.866922065186206) circle (2.0pt);
\draw[color=qqqqff] (15.794764329455749,-16.998356448854587) node {\normalsize $9$};
\end{scriptsize}
\end{tikzpicture}
}
\caption{\label{fig:EX_Step8}
Next, we will apply the procedure to 
\flexis 5-6-7 (of type up-up) and 7-8-9 (of type down-up).
This will result in moving the midpoint 8 down until it reaches $\vmin$.
 }
\end{figure}

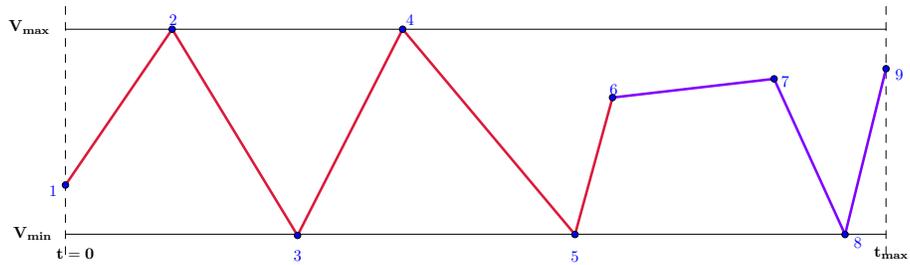
\begin{figure}[htp]
\centering
\resizebox{\textwidth}{!}{

\begin{tikzpicture}[line cap=round,line join=round,>=triangle 45,x=1.0cm,y=1.0cm]
\draw (-2.5,-16.)-- (15.5,-16.);
\draw (-2.5,-20.5)-- (15.5,-20.5);
\draw (-3.8597311975693356,-15.690812833218265) node[anchor=north west] {\normalsize $\mathbf{V_{max}}$};
\draw (-3.7586278255167374,-20.22024390117469) node[anchor=north west] {\normalsize $\mathbf{V_{min}}$};
\draw [dash pattern=on 5pt off 5pt] (-2.5,-15.5)-- (-2.5,-21.);
\draw [dash pattern=on 5pt off 5pt] (15.5,-15.5)-- (15.5,-21.);
\draw (15.1274820739086,-20.6481881799534) node[anchor=north west] {\normalsize $\mathbf{t_{max}}$};
\draw (-2.8082561282223146,-20.68862952877444) node[anchor=north west] {\normalsize $\mathbf{t=0}$};
\draw [line width=1.6pt,color=dcrutc] (-2.5,-19.42)-- (-0.15934778044424203,-16.);
\draw [line width=1.6pt,color=dcrutc] (-0.15934778044424203,-16.)-- (2.590663939386428,-20.526864133490324);
\draw [line width=1.6pt,color=dcrutc] (2.590663939386428,-20.526864133490324)-- (4.895820822185666,-16.);
\draw [line width=1.6pt,color=dcrutc] (4.895820822185666,-16.)-- (8.677086936952838,-20.506643459079765);
\draw [line width=1.6pt,color=dcrutc] (8.677086936952838,-20.506643459079765)-- (9.5,-17.5);
\draw [line width=1.6pt,color=xfqqff] (9.5,-17.5)-- (13.044752609625078,-17.089349483701923);
\draw [line width=1.6pt,color=xfqqff] (14.60174453923509,-20.506643459079765)-- (15.5,-16.866922065186206);
\draw [line width=1.6pt,color=xfqqff] (13.044752609625078,-17.089349483701923)-- (14.60174453923509,-20.506643459079765);
\begin{scriptsize}
\draw [fill=qqqqff] (-2.5,-19.42) circle (2.0pt);
\draw[color=qqqqff] (-2.767814779401275,-19.546161424580074) node {\normalsize $1$};
\draw [fill=qqqqff] (-0.15934778044424203,-16.) circle (2.0pt);
\draw[color=qqqqff] (-0.13912710603372316,-15.80533665863392) node {\normalsize $2$};
\draw [fill=qqqqff] (2.590663939386428,-20.526864133490324) circle (2.0pt);
\draw[color=qqqqff] (2.5906639393864275,-20.96160863331646) node {\normalsize $3$};
\draw [fill=qqqqff] (4.895820822185666,-16.) circle (2.0pt);
\draw[color=qqqqff] (5.057586217469823,-15.80533665863392) node {\normalsize $4$};
\draw [fill=qqqqff] (8.677086936952838,-20.506643459079765) circle (2.0pt);
\draw[color=qqqqff] (8.677086936952838,-20.98182930772698) node {\normalsize $5$};
\draw [fill=qqqqff] (9.5,-17.5) circle (2.0pt);
\draw[color=qqqqff] (9.526355262194663,-17.3218872394229) node {\normalsize $6$};
\draw [fill=qqqqff] (13.044752609625078,-17.089349483701923) circle (2.0pt);
\draw[color=qqqqff] (13.287400702551315,-17.139901169728223) node {\normalsize $7$};
\draw [fill=qqqqff] (14.60174453923509,-20.506643459079765) circle (2.0pt);
\draw[color=qqqqff] (14.884833980982364,-20.678519191569183) node {\normalsize $8$};
\draw [fill=qqqqff] (15.5,-16.866922065186206) circle (2.0pt);
\draw[color=qqqqff] (15.794764329455749,-16.998356448854587) node {\normalsize $9$};
\end{scriptsize}
\end{tikzpicture}
}
\caption{\label{fig:EX_Step9}
Since there no more non-overlapping \flexis in the schedule,
we try to move the one that remains in the leaps section 
to the end of the schedule.
In this case, as all of them are already located after the leaps section,
this step is skipped.
Next, we will apply the same procedure but with the first timed action if it is a flat one
or with the last timed action if it does not reach neither $\vmin$ nor $\vmax$ (and so shrink and stretch operations can be applied to it).
In this case we apply this operation to \flexi 6-7-8 (of type up-down) and the last timed action 8-9.
This results in moving point 9 up until it reaches $\vmax$.
}
\end{figure}

\begin{figure}[htp]
\centering
\resizebox{\textwidth}{!}{

\begin{tikzpicture}[line cap=round,line join=round,>=triangle 45,x=1.0cm,y=1.0cm]
\draw (-2.5,-16.)-- (15.5,-16.);
\draw (-2.5,-20.5)-- (15.5,-20.5);
\draw (-3.8597311975693356,-15.690812833218265) node[anchor=north west] {\normalsize $\mathbf{V_{max}}$};
\draw (-3.7586278255167374,-20.22024390117469) node[anchor=north west] {\normalsize $\mathbf{V_{min}}$};
\draw [dash pattern=on 5pt off 5pt] (-2.5,-15.5)-- (-2.5,-21.);
\draw [dash pattern=on 5pt off 5pt] (15.5,-15.5)-- (15.5,-21.);
\draw (15.1274820739086,-20.6481881799534) node[anchor=north west] {\normalsize $\mathbf{t_{max}}$};
\draw (-2.8082561282223146,-20.68862952877444) node[anchor=north west] {\normalsize $\mathbf{t=0}$};
\draw [line width=1.6pt,color=dcrutc] (-2.5,-19.42)-- (-0.15934778044424203,-16.);
\draw [line width=1.6pt,color=dcrutc] (-0.15934778044424203,-16.)-- (2.590663939386428,-20.526864133490324);
\draw [line width=1.6pt,color=dcrutc] (2.590663939386428,-20.526864133490324)-- (4.895820822185666,-16.);
\draw [line width=1.6pt,color=dcrutc] (4.895820822185666,-16.)-- (8.677086936952838,-20.506643459079765);
\draw [line width=1.6pt,color=xfqqff] (8.677086936952838,-20.506643459079765)-- (9.5,-17.5);
\draw [line width=1.6pt,color=xfqqff] (9.5,-17.5)-- (12.680780470235725,-17.109570158112444);
\draw [line width=1.6pt,color=dcrutc] (14.338875771898335,-20.506643459079765)-- (15.5,-16.);
\draw [line width=1.6pt,color=xfqqff] (12.680780470235725,-17.109570158112444)-- (14.338875771898335,-20.506643459079765);
\begin{scriptsize}
\draw [fill=qqqqff] (-2.5,-19.42) circle (2.0pt);
\draw[color=qqqqff] (-2.767814779401275,-19.546161424580074) node {\normalsize $1$};
\draw [fill=qqqqff] (-0.15934778044424203,-16.) circle (2.0pt);
\draw[color=qqqqff] (-0.13912710603372316,-15.80533665863392) node {\normalsize $2$};
\draw [fill=qqqqff] (2.590663939386428,-20.526864133490324) circle (2.0pt);
\draw[color=qqqqff] (2.5906639393864275,-20.96160863331646) node {\normalsize $3$};
\draw [fill=qqqqff] (4.895820822185666,-16.) circle (2.0pt);
\draw[color=qqqqff] (5.057586217469823,-15.80533665863392) node {\normalsize $4$};
\draw [fill=qqqqff] (8.677086936952838,-20.506643459079765) circle (2.0pt);
\draw[color=qqqqff] (8.677086936952838,-20.98182930772698) node {\normalsize $5$};
\draw [fill=qqqqff] (9.5,-17.5) circle (2.0pt);
\draw[color=qqqqff] (9.526355262194663,-17.3218872394229) node {\normalsize $6$};
\draw [fill=qqqqff] (12.680780470235725,-17.109570158112444) circle (2.0pt);
\draw[color=qqqqff] (12.94364923757248,-17.139901169728223) node {\normalsize $7$};
\draw [fill=qqqqff] (14.338875771898335,-20.506643459079765) circle (2.0pt);
\draw[color=qqqqff] (14.621965213645609,-20.678519191569183) node {\normalsize $8$};
\draw [fill=qqqqff] (15.5,-16.) circle (2.0pt);
\draw[color=qqqqff] (15.794764329455749,-16.128867449202236) node {\normalsize $9$};
\end{scriptsize}
\end{tikzpicture}
}
\caption{\label{fig:EX_Step10}
Our schedule is already partitioned into three distinct sections: head, leaps, and tail. 
However, the tail section does not follow any of the 10 patterns in Figure \ref{fig:tail}.
We cannot apply become the flexes 5-6-7 and 6-7-8 are overlapping.
At the same time points 6 and 7 still have some flexibility in them. 
We apply the wedge operation to the 5-6-7-8 segment to resolve this.
In this case, points 6 and 7 are moved up until one of them reaches $\vmax$
and the first one to do so is point 7.}
\end{figure}
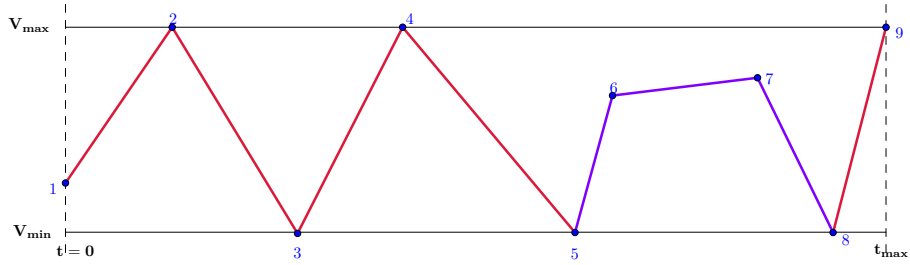

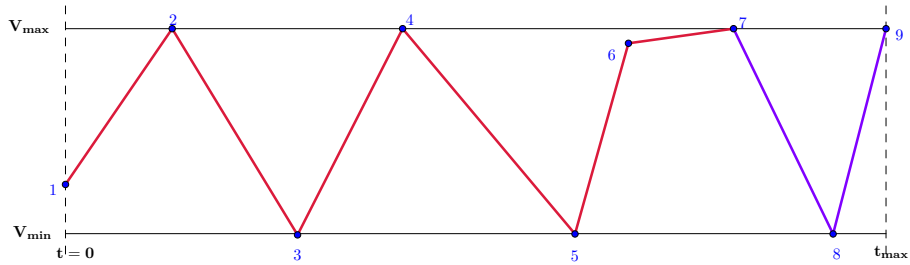
\begin{figure}[htp]
\centering
\resizebox{\textwidth}{!}{

\begin{tikzpicture}[line cap=round,line join=round,>=triangle 45,x=1.0cm,y=1.0cm]
\draw (-2.5,-16.)-- (15.5,-16.);
\draw (-2.5,-20.5)-- (15.5,-20.5);
\draw (-3.8597311975693356,-15.690812833218265) node[anchor=north west] {\normalsize $\mathbf{V_{max}}$};
\draw (-3.7586278255167374,-20.22024390117469) node[anchor=north west] {\normalsize $\mathbf{V_{min}}$};
\draw [dash pattern=on 5pt off 5pt] (-2.5,-15.5)-- (-2.5,-21.);
\draw [dash pattern=on 5pt off 5pt] (15.5,-15.5)-- (15.5,-21.);
\draw (15.1274820739086,-20.6481881799534) node[anchor=north west] {\normalsize $\mathbf{t_{max}}$};
\draw (-2.8082561282223146,-20.68862952877444) node[anchor=north west] {\normalsize $\mathbf{t=0}$};
\draw [line width=1.6pt,color=dcrutc] (-2.5,-19.42)-- (-0.15934778044424203,-16.);
\draw [line width=1.6pt,color=dcrutc] (-0.15934778044424203,-16.)-- (2.590663939386428,-20.526864133490324);
\draw [line width=1.6pt,color=dcrutc] (2.590663939386428,-20.526864133490324)-- (4.895820822185666,-16.);
\draw [line width=1.6pt,color=dcrutc] (4.895820822185666,-16.)-- (8.677086936952838,-20.506643459079765);
\draw [line width=1.6pt,color=xfqqff] (14.338875771898335,-20.506643459079765)-- (15.5,-16.);
\draw [line width=1.6pt,color=xfqqff] (14.338875771898335,-20.506643459079765)-- (12.155042935562214,-16.);
\draw [line width=1.6pt,color=dcrutc] (8.677086936952838,-20.506643459079765)-- (9.849886052762976,-16.320963856102175);
\draw [line width=1.6pt,color=dcrutc] (9.849886052762976,-16.320963856102175)-- (12.155042935562214,-16.);
\begin{scriptsize}
\draw [fill=qqqqff] (-2.5,-19.42) circle (2.0pt);
\draw[color=qqqqff] (-2.767814779401275,-19.546161424580074) node {\normalsize $1$};
\draw [fill=qqqqff] (-0.15934778044424203,-16.) circle (2.0pt);
\draw[color=qqqqff] (-0.13912710603372316,-15.80533665863392) node {\normalsize $2$};
\draw [fill=qqqqff] (2.590663939386428,-20.526864133490324) circle (2.0pt);
\draw[color=qqqqff] (2.5906639393864275,-20.96160863331646) node {\normalsize $3$};
\draw [fill=qqqqff] (4.895820822185666,-16.) circle (2.0pt);
\draw[color=qqqqff] (5.057586217469823,-15.80533665863392) node {\normalsize $4$};
\draw [fill=qqqqff] (8.677086936952838,-20.506643459079765) circle (2.0pt);
\draw[color=qqqqff] (8.677086936952838,-20.98182930772698) node {\normalsize $5$};
\draw [fill=qqqqff] (14.338875771898335,-20.506643459079765) circle (2.0pt);
\draw[color=qqqqff] (14.419758469540414,-20.94138795890594) node {\normalsize $8$};
\draw [fill=qqqqff] (15.5,-16.) circle (2.0pt);
\draw[color=qqqqff] (15.794764329455749,-16.128867449202236) node {\normalsize $9$};
\draw [fill=qqqqff] (12.155042935562214,-16.) circle (2.0pt);
\draw[color=qqqqff] (12.357249679667412,-15.865998681865477) node {\normalsize $7$};
\draw [fill=qqqqff] (9.849886052762976,-16.320963856102175) circle (2.0pt);
\draw[color=qqqqff] (9.485913913373622,-16.57372228623367) node {\normalsize $6$};
\end{scriptsize}
\end{tikzpicture}
}
\caption{\label{fig:EX_Step11} 
There is only one point between $\vmin$ and $\vmax$ left (point 6), but
the tail still does not follow any of the 10 patterns.
We use the shift-down operation to segment 7-8 and move it after 5.
}
\end{figure}

\begin{figure}[htp]
\centering
\resizebox{\textwidth}{!}{

\begin{tikzpicture}[line cap=round,line join=round,>=triangle 45,x=1.0cm,y=1.0cm]
\draw (-2.5,-16.)-- (15.5,-16.);
\draw (-2.5,-20.5)-- (15.5,-20.5);
\draw (-3.8597311975693356,-15.690812833218265) node[anchor=north west] {\normalsize $\mathbf{V_{max}}$};
\draw (-3.7586278255167374,-20.22024390117469) node[anchor=north west] {\normalsize $\mathbf{V_{min}}$};
\draw [dash pattern=on 5pt off 5pt] (-2.5,-15.5)-- (-2.5,-21.);
\draw [dash pattern=on 5pt off 5pt] (15.5,-15.5)-- (15.5,-21.);
\draw (15.1274820739086,-20.6481881799534) node[anchor=north west] {\normalsize $\mathbf{t_{max}}$};
\draw (-2.8082561282223146,-20.68862952877444) node[anchor=north west] {\normalsize $\mathbf{t=0}$};
\draw [line width=1.6pt,color=dcrutc] (-2.5,-19.42)-- (-0.15934778044424203,-16.);
\draw [line width=1.6pt,color=dcrutc] (-0.15934778044424203,-16.)-- (2.590663939386428,-20.526864133490324);
\draw [line width=1.6pt,color=dcrutc] (2.590663939386428,-20.526864133490324)-- (4.895820822185666,-16.);
\draw [dash pattern=on 5pt off 5pt] (2.590663939386428,-15.5)-- (2.590663939386428,-20.526864133490324);
\draw [line width=1.6pt,color=dcrutc] (4.895820822185666,-16.)-- (8.677086936952838,-20.506643459079765);
\draw [line width=1.6pt,color=dcrutc] (8.677086936952838,-20.506643459079765)-- (10.,-16.);
\draw [line width=1.6pt,color=dcrutc] (10.,-16.)-- (12.155042935562214,-20.5);
\draw [line width=1.6pt,color=dcrutc] (12.155042935562214,-20.5)-- (13.246959353730274,-16.34118453051269);
\draw [dash pattern=on 5pt off 5pt] (12.155042935562214,-15.5)-- (12.155042935562214,-20.5);
\draw [line width=1.6pt,color=dcrutc] (13.246959353730274,-16.34118453051269)-- (15.5,-16.);
\begin{scriptsize}
\draw [fill=qqqqff] (-2.5,-19.42) circle (2.0pt);
\draw[color=qqqqff] (-2.767814779401275,-19.546161424580074) node {\normalsize $1$};
\draw [fill=qqqqff] (-0.15934778044424203,-16.) circle (2.0pt);
\draw[color=qqqqff] (-0.13912710603372316,-15.80533665863392) node {\normalsize $2$};
\draw [fill=qqqqff] (2.590663939386428,-20.526864133490324) circle (2.0pt);
\draw[color=qqqqff] (2.5906639393864275,-20.96160863331646) node {\normalsize $3$};
\draw [fill=qqqqff] (4.895820822185666,-16.) circle (2.0pt);
\draw[color=qqqqff] (5.057586217469823,-15.80533665863392) node {\normalsize $4$};
\draw [fill=qqqqff] (8.677086936952838,-20.506643459079765) circle (2.0pt);
\draw[color=qqqqff] (8.677086936952838,-20.98182930772698) node {\normalsize $5$};
\draw [fill=qqqqff] (10.,-16.) circle (2.0pt);
\draw[color=qqqqff] (10.17341684333129,-15.906440030686518) node {\normalsize $6$};
\draw [fill=qqqqff] (12.155042935562214,-20.5) circle (2.0pt);
\draw[color=qqqqff] (12.175263609972735,-20.96160863331646) node {\normalsize $7$};
\draw [fill=qqqqff] (13.246959353730274,-16.34118453051269) circle (2.0pt);
\draw[color=qqqqff] (12.963869911983,-16.51306026300211) node {\normalsize $8$};
\draw [fill=qqqqff] (15.5,-16.) circle (2.0pt);
\draw[color=qqqqff] (15.794764329455749,-16.007543402739117) node {\normalsize $9$};
\end{scriptsize}
\end{tikzpicture}
}
\caption{\label{fig:EX_Step12} 
Finally, both the head section (1-2-3) and tail section (7-8-9) follows one of the standard patterns.
The head section follows the \PUD pattern (Figure \ref{fig:head}(e)) and the tail section follows \PUU pattern (Figure \ref{fig:tail}(b)).
The leaps section (3-4-5-6-7) consists of two (complete) leaps.}
\end{figure}
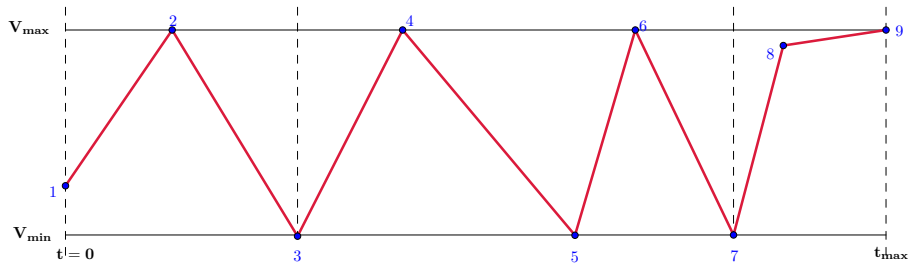 

\end{document}